\documentclass[11pt]{article}

\usepackage[american]{babel}
\usepackage[T1]{fontenc}
\usepackage[utf8]{inputenc}
\usepackage[text={15.0cm,22.0cm}]{geometry}
\usepackage{lmodern}
\usepackage{microtype}

\usepackage{amsmath, amssymb, amsthm, mathtools, thmtools}
\usepackage{wrapfig}
\usepackage{nicefrac}
\usepackage{color}
\usepackage{xcolor}
\usepackage{bm}
\usepackage{mdframed}
\usepackage{wrapfig}
\usepackage{subfig}
\usepackage{paralist}

\newmdtheoremenv[linecolor=gray,leftmargin=10,rightmargin=10,
backgroundcolor=gray!40,innertopmargin=0pt,ntheorem]{mydraft}{Draft}

\definecolor{myred}{RGB}{238,51,119}
\definecolor{myorange}{RGB}{238,119,51}
\definecolor{mygreen}{RGB}{0,153,136}
\colorlet{small}{white}
\colorlet{medium}{black!25!white}
\colorlet{big}{black!60!white}

\usepackage{tikz}
\usetikzlibrary{calc, positioning, decorations,decorations.pathreplacing,
  patterns, arrows}

\def\itempadding{.7pt}
\def\boxpadding{5pt}

\newcommand{\drawBin}[3]{
  \begin{scope}[x= 1cm, y=2cm]
  \draw[dashed] [#2] ($(#1,0)+(-\boxpadding,-\itempadding)$)
  rectangle node (#3) {}
  ($(#1,1)+(\boxpadding,0)+(1,0)$);
  \end{scope}
}

\newcommand{\addToBin}[6]{
  \begin{scope}[x= 1cm, y=2cm]
       \draw [#5]
       ($(#1,#2)+(\itempadding,\itempadding)$)
       rectangle node (#6) {#4}
       ($(#1,#2)+(-\itempadding,-\itempadding)+(1,#3)$);
\end{scope}
}

\newcommand{\rightBin}[1]{
  ($ (#1) + (.65cm,0) $)
}

\newcommand{\leftBin}[1]{
  ($ (#1) - (.65cm,0) $)
}

\usepackage[inline]{enumitem}

\usepackage[ruled,vlined,linesnumbered]{algorithm2e}

\usepackage{hyperref}

\usepackage{authblk}

\usepackage{booktabs}
\usepackage{multirow}

\usepackage{fontawesome}

\declaretheorem{theorem}
\declaretheorem[sibling=theorem]{lemma}

\declaretheorem[sibling=theorem]{proposition}
\theoremstyle{definition}

\theoremstyle{remark}
\declaretheorem[sibling=theorem]{remark}
\declaretheorem[sibling=theorem]{claim}

\newcommand{\Opt}{\operatorname{\text{\textsc{opt}}}}
\newcommand{\Alg}{\operatorname{\text{\textsc{alg}}}}

\DeclareMathOperator{\OPT}{OPT}
\newcommand{\items}{\Gamma}
\newcommand{\BBbins}{\mathtt{BB}}
\newcommand{\BMbins}{\mathtt{BM}}

\newcommand{\BSbins}{\mathtt{BS}}
\newcommand{\Xbins}{\mathtt{XB}}
\newcommand{\BSfullbins}{\mathtt{BSC}}
\newcommand{\BSwaitbins}{\mathtt{BSP}} \newcommand{\Sbins}{\mathtt{S}}
\newcommand{\Sfullbins}{\mathtt{SC}}
\newcommand{\Mbins}{\mathtt{M}}
\newcommand{\Mfullbins}{\mathtt{MC}}

\newcommand{\bins}{\mathtt{Bins}}

\newcommand{\schain}{\mathtt{SCh}}
\newcommand{\group}{\mathtt{Gr}}
\newcommand{\VGbuffer}{\mathtt{BGr}}

\newcommand{\Sbuffer}{\mathtt{SB}}
\newcommand{\Gbuffer}{\mathtt{GB}}

\newcommand{\chainpush}{\mathsf{ChainPush}}
\newcommand{\chainpull}{\mathsf{ChainPull}}
\newcommand{\greedypush}{\mathsf{GreedyPush}}
\newcommand{\greedypull}{\mathsf{GreedyPull}}
\newcommand{\gapfill}{\mathsf{GapFill}}
\newcommand{\grouppush}{\mathsf{GroupPush}}
\newcommand{\grouppull}{\mathsf{GroupPull}}

\newcommand{\groupinsert}{\mathsf{GroupInsert}}
\newcommand{\groupdelete}{\mathsf{GroupDelete}}

\newcommand{\ihuge}{\mathtt{huge}}
\newcommand{\ilarge}{\mathtt{large}}
\newcommand{\imedium}{\mathtt{medium}}
\newcommand{\itiny}{\mathtt{tiny}}

\DeclareMathOperator{\fsmall}{small}
\DeclareMathOperator{\fbig}{big}
\DeclareMathOperator{\fmedium}{medium}

\newcommand\Class[1]{\mathchoice {\text{\normalfont\fontsize{9pt}{10pt}\selectfont$\mathrm{#1}$}}{\text{\normalfont\fontsize{9pt}{10pt}\selectfont$\mathrm{#1}$}}{\text{\normalfont$\mathrm{#1}$}}{\text{\normalfont$\mathrm{#1}$}}}

\newcommand{\eps}{\varepsilon}
\newcommand{\Oh}{\mathcal{O}}

\newcommand{\off}{\operatorname{off}}
\newcommand{\on}{\operatorname{on}}

\newcommand{\ZZ}{\mathbb{Z}}

\DeclarePairedDelimiter\ceil{\lceil}{\rceil}
\DeclarePairedDelimiter\abs{|}{|}

\DeclarePairedDelimiter\set{\lbrace}{\rbrace}
\DeclarePairedDelimiterX\sett[2]{\lbrace}{\rbrace}{ #1 \,\delimsize| \,\mathopen{} #2 }

\usepackage{xspace}
\newcommand{\ie}{i.\,e.,\xspace}

\title{Online Bin Covering with Limited Migration\footnote{This work was partially supported
by DFG Project, "Robuste Online-Algorithmen für Scheduling- und Packungsprobleme", JA 612 /19-1, and
by GIF-Project "Polynomial Migration for Online Scheduling".}}

\author[1]{Sebastian Berndt}
\author[2]{Leah Epstein}
\author[1]{Klaus Jansen}
\author[3]{Asaf Levin}
\author[1]{Marten Maack}
\author[1]{Lars Rohwedder}

\affil[1]{Department of Computer Science, Kiel University, Kiel, Germany}
\affil[2]{Department of Mathematics, University of Haifa, Haifa, Israel}
\affil[3]{Faculty of Industrial Engineering and Management, The Technion, Haifa,
  Israel}

\begin{document}

\maketitle

\begin{abstract}
  Semi-online models where decisions may be revoked in a limited way
  have been studied extensively in the last years.

  This is motivated by the fact that
  the pure online model is often too restrictive to model real-world
  applications, where some changes might be allowed.
  A well-studied measure of the amount of decisions that can be revoked is the migration factor
  $\beta$:
  When an object $o$ of size $s(o)$ arrives,
  the decisions for objects of total size at most $\beta\cdot s(o)$ may be revoked.
  Usually $\beta$ should be a constant.
  This means that a small object only leads to small changes.
  This measure has been successfully investigated for
  different, classic problems such as bin packing or makespan minimization.
  The dual of makespan minimization~--~the Santa Claus or machine covering
  problem~--~has also been studied, whereas the dual of bin packing~--~the bin
  covering problem~--~has not been looked at from such a perspective.

  In this work, we extensively study the bin covering problem with migration in
  different scenarios. We develop algorithms both for the static case~--~where
  only insertions are allowed~--~and for the dynamic case, where items may also
  depart. We also develop lower bounds for these scenarios both for amortized
  migration and for worst-case migration showing that our algorithms have
  nearly optimal migration factor and asymptotic competitive ratio (up to an arbitrary
  small $\eps$). We therefore resolve the competitiveness of the bin covering
  problem with migration.

\end{abstract}

\section{Introduction}

\emph{Online} algorithms aim to maintain a competitive solution
without knowing future parts of the input. The competitive ratio
of such an algorithm (for a maximization problem) is thus defined
as the worst-case ratio between the value of an  optimal solution
produced by an offline algorithm knowing the complete input and
the value of the solution produced by the online algorithm.
Furthermore, once a decision is made by these algorithms, this
decision is fixed and irreversible. While a surprisingly large
number of problems do have such algorithms, the complete
irreversibility requirement is often too strict, leading to high
competitive ratios. Furthermore, if the departure of objects from
the instance is also allowed, irreversible online algorithms are
rarely able to be competitive at all. From a practical point of
view, this is quite alarming, as the departure of objects is part
of many applications. We call such a problem \emph{dynamic} and
the version with only insertions \emph{static}.

A number of different scenarios to loosen the strict requirement of
irreversibility~--~called \emph{semi-online} scenarios~--~have been developed
over time in order to find algorithms that achieve good competitive ratios for
some of the scenarios
bounded reversibility. In the last few years, the concept of the
\emph{migration factor} has been studied intensively \cite{Berndt2018,
  DBLP:journals/mp/EpsteinL09,DBLP:journals/siamjo/EpsteinL13,DBLP:journals/algorithmica/EpsteinL14,
  DBLP:conf/icalp/FeldkordFGGKRW18,DBLP:conf/esa/GalvezSV18,DBLP:conf/approx/JansenKKL17,DBLP:journals/mor/SandersSS09,DBLP:journals/mor/SkutellaV16}.
Roughly speaking, a migration factor of $\beta$ allows to reverse
a total size of $\beta\cdot s(o)$ decisions, where $s(o)$ denotes
the \emph{size} of the newly arrived object $o$. For a packing
problem, this means that the algorithm is allowed to repack
objects with a total size of $ \beta \cdot s(o)$. This notion of
reversibility is very natural, as it guarantees that a small
object can only lead to small changes in the solution structure.
Furthermore, algorithms with bounded migration factor often show a
very clear-cut tradeoff between their migration and the
competitive ratio: Many algorithms in this setting have a bounded
migration factor which can be defined as a function $f(\eps)$
(growing with $\frac 1{\eps}$) and a small competitive ratio
$g(\eps)$ (growing with $\eps$), where the functions $f$ and $g$
can be defined for all $\eps > 0$
\cite{Berndt2018,DBLP:journals/mp/EpsteinL09,DBLP:conf/icalp/FeldkordFGGKRW18,DBLP:conf/esa/GalvezSV18,DBLP:conf/approx/JansenKKL17,DBLP:journals/mor/SandersSS09,DBLP:journals/mor/SkutellaV16}.
Such algorithms are called \emph{robust}, as the amount of
reversibility allowed only depends on the solution guarantee that
one wants to achieve. Such robust algorithms thus serve as
evidence for the possibility for sensitivity analysis in
approximated settings.

Many different problems have been studied in online and
semi-online scenarios, but two problems that have been considered
in nearly every scenario are classic scheduling problems: The
\emph{bin packing} problem and the
\emph{makespan minimization} problem. Both of these problems have been studied intensively in the migration model
\cite{Berndt2018,DBLP:journals/mp/EpsteinL09,DBLP:journals/algorithmica/EpsteinL14,DBLP:conf/icalp/FeldkordFGGKRW18,DBLP:conf/approx/JansenKKL17,DBLP:journals/mor/SandersSS09,DBLP:journals/mor/SkutellaV16}.
Both of these problems also have corresponding dual maximization variants.
The dual version of the makespan minimization problem, often
called the \emph{Santa Claus} or \emph{machine covering} problem,
has also been studied with migration
\cite{DBLP:conf/esa/GalvezSV18,DBLP:journals/mor/SkutellaV16}. In
contrast, the dual version of bin packing, called \emph{bin
covering} has not yet been studied in this model. The aim of this
paper is to remedy this situation by taking a look at this classic
scheduling problem in the migration model.

\subsection*{Formal Problem Statement}
In the \emph{bin covering problem}, a set of items $\items$ with
sizes $s\colon\items\to (0,1]$ is used to cover as many unit sized
bins as possible, that is, $\items$ has to be partitioned
maximizing the number of partitions with summed up item size of at
least one. An instance of the problem will usually be denoted as
$I$ and is given as a sequence of entries $(i,s(i))$ where $i$ is
the identifier of the item and $s(i)$ is the size of the item. A
\emph{solution} to such an instance $I$ with items $\items$ is a
partition $P\colon \items \to \mathbb{N}$ and a set $B=P^{-1}(k)$
with $B\neq \emptyset$ is called a \emph{bin} and we say that the
items in $B$ are packed into the $k$-th bin. For a subset $I'
\subseteq I$, let $s(I')=\sum_{i\in I'} s(i)$. A bin $B$ is
\emph{covered} if $s(B) \geq 1$, where $s(B)$ is called the load
of $B$ or its size, and the goal is to maximize the number of such
covered bins. The optimal (maximum) number of covered bins of
instance $I$ is denoted as $\Opt(I)$.

We also use the following notations throughout our work: The smallest size of an item in bin $B$ is defined as $s_{\min}(B)
:= \min_{i\in B}\{s(i)\}$ and the largest size is defined as
$s_{\max}(B) := \max_{i\in B}\{s(i)\}$. If $\mathcal{B}$ is a set
of bins, we also define its total size $s(\mathcal{B}) :=
\sum_{B\in \mathcal{B}} s(B)$, its minimal size
$s_{\min}(\mathcal{B}) := \min_{B\in \mathcal{B}} s_{\min}(B)$,
and its maximal size $s_{\max}(\mathcal{B}) := \max_{B\in
\mathcal{B}} s_{\max}(B)$. Furthermore, we define
$s_{\min}(\emptyset) = +\infty$ and $s_{\max}(\emptyset) = 0$.

We consider variants of static and dynamic online bin covering in
which algorithms are allowed to reassign a bounded amount of
previously assigned items. In particular, an algorithm has a
\emph{migration factor} of $\beta$, if the total size of items
that it reassigns upon arrival or departure of an item of size $s$
is bounded by $\beta s$. Moreover, it has an \emph{amortized
  migration factor} of $\beta$, if at any time the total size of items that have
been reassigned by the algorithm in total is bounded by $\beta S$,
where $S$ is the total size of all items that arrived before.
Intuitively, an item of size $s$ creates a migration potential of
size $\beta s$ upon arrival, and this potential may be used by an
algorithm to reassign items right away (non-amortized) or anytime
from then on (amortized). Note that if an algorithm has a
non-amortized migration factor of $\beta$, then it also has an
amortized migration factor of at most $\beta$. Thus, we study four
variants in this work.

Offline bin covering is $\Class{NP}$-hard and therefore there is
little hope for a polynomial time algorithm solving the problem to
optimality, and in the online setting there is no algorithm that
can maintain an optimal solution regardless of its running time.
We prove that this non-existence of algorithms that maintain an
optimal solution holds also for (static or dynamic) algorithms
with bounded amortized migration factor (and thus also for
algorithms with bounded non-amortized migration factor).

Hence, algorithms satisfying some performance guarantee are
studied. In particular an offline algorithm $\Alg$ for a
maximization problem has an \emph{asymptotic performance
guarantee} of $\alpha \geq 1$, if $\Opt(I) \leq \alpha \cdot
\Alg(I) + c$, where $\Opt(I)$ and $\Alg(I)$ are the objective
values of an optimal solution or the one produced by $\Alg$
respectively for some instance $I$, and $c$ is an input
independent constant. If $c=0$ holds, $\alpha$ is called
\emph{absolute} rather than asymptotic. An online algorithm has a
(asymptotic or absolute) \emph{competitive ratio} of $\alpha$, if
after each arrival or departure an (asymptotic or absolute)
performance guarantee of $\alpha$ for the instance of the present
items holds. Note that we use the convention of competitive ratios
larger than $1$ for maximization problems. For minimization
problems similar definitions are used but they use the required
inequality $\Alg(I)\leq \alpha\cdot \Opt(I)+c$. As we study
asymptotic competitive ratios in this paper, we will sometimes
omit the word asymptotic (and we always use the word absolute for
absolute competitive ratios).

In what follows we assume $s\colon\items\to (0,1)$, \ie that there
are no items of size $1$. This is justified by the following. For
any algorithm it is possible to add a rule that such an item will
always be packed into its own bin, which will be covered. Without
loss of generality, we can consider an optimal solution that
applies this rule as well. This does not affect the (asymptotic or
absolute) {competitive ratio}.

\subsection*{Known Results for Bin Covering}
The offline bin covering problem was first studied by Assmann et
al.~\cite{DBLP:journals/jal/AssmannJKL84}. It was shown that a
simple greedy strategy achieves approximation ratio $2$. For the
online version of the bin covering problem, Csirik and Totik
showed in~\cite{DBLP:journals/dam/CsirikT88} that this simple
greedy algorithm also works in the online setting and that the
competitive ratio of $2$ reached by this algorithm is the best
possible. Csirik, Johnson, and Kenyon presented an asymptotic
polynomial time approximation scheme (APTAS) with approximation
ratio $1+\eps$ in~\cite{DBLP:conf/soda/CsirikJK01}. This was
improved to an asymptotic fully polynomial time approximation
scheme (AFPTAS) by Jansen and Solis-Oba
in~\cite{jansen2003bincovering}. Many different variants of this
problem have also been investigated: If a certain number of
classes needs to be part of each
bin~\cite{DBLP:journals/mst/EpsteinIL10,
  DBLP:conf/latin/FischerR18}; if items are drawn
probabilistically~\cite{DBLP:conf/latin/FischerR16,DBLP:conf/latin/FischerR18};
if bins have different sizes~\cite{DBLP:journals/iandc/Epstein01,
  DBLP:journals/orl/WoegingerZ99}; if the competitiveness is not measured with
regard to an optimal offline
algorithm~\cite{DBLP:journals/tcs/ChristFL14,DBLP:journals/scheduling/EpsteinFK12}.
More variants are discussed for example
in~\cite{DBLP:reference/crc/2007aam} and lower bounds for several
variants are studied in~\cite{balogh2018lower}. The dynamic
variant was not studied in the online scenario since there is no
algorithm with a finite competitive ratio. Specifically, if all
items are very small, it can happen that items depart in a way
that no bin remains covered, while an optimal solution could pack
items differently such that all items depart from a small number
of bins, and in this way it will still cover many bins.

\subsection*{Our Results}
We present competitive algorithms using both amortized migration
and non-amortized migration and develop nearly matching lower
bounds (up to an arbitrary small additive term of $\eps$). These
bounds show the optimality of all of our algorithms for both the
static and the dynamic version of the bin covering problem. The
main technical contribution of our work is an algorithm with
competitive ratio $3/2+\eps$ and non-amortized migration of
$\Oh(\frac{\log^{2}(1/\eps)} {\eps ^{5}})$ for the dynamic bin
covering problem where items arrive and depart. A major obstacle
in the design of competitive algorithms for dynamic problems is
the impossibility of moving large items on the arrival or
departure of small items. We overcome this obstacle by developing
a delicate technique to combine the packing of large and small
items. The main results of this work are summarized in the
following table. Note that the lower bound of $1$ in the third row
indicates that there is no online algorithm that maintains an
optimal solution with amortized migration factor $\Oh(1)$. All of
our algorithms run in polynomial time. Curiously, we achieve a
polynomial migration factor, while most known migration factors
are exponential (e.g., for the makespan minimization
problem~\cite{DBLP:journals/mor/SandersSS09}) with the exception
of bin packing ~\cite{DBLP:conf/icalp/JansenK13}.

Thus, we show how to overcome the shortcoming of online algorithms
in the dynamic setting by using migration and get a constant
competitive ratio in the most general setting. Surprisingly, the
constant is even smaller than the best possible competitive ratio
for pure online algorithms in the static case for which the
competitive ratio is $2$, while our bound is $\frac 32+\eps$ for
any $\eps>0$.

{\begin{table}[h!]
\renewcommand{\arraystretch}{1.35   }

\begin{tabular}{lllll}
  \toprule
  Amortization & Departures & Lower Bound &  Competitive Ratio & Migration \\
  \cmidrule(lllll){1-5}
  \faRemove & \faRemove & $3/2$ & $3/2+\eps$ & $\Oh(1/\eps)$\\
  \faRemove & \faCheck & $3/2$ & $3/2+\eps$ & $\Oh(\frac{\log^{2}(1/\eps)} {\eps ^{5}})$\\
  \faCheck & \faRemove &  $1$ & $1+\eps$ & $\Oh(1/\eps)$\\
  \faCheck & \faCheck & $3/2$ & $3/2+\eps$ & $\Oh(\frac{\log^{2}(1/\eps)} {\eps ^{5}})$\\
  \bottomrule
\end{tabular}
\end{table}}
\subsection*{Related Results}

\paragraph*{Makespan Minimization and Santa Claus:}
The migration factor model was introduced by Sanders, Sivadasan,
and Skutella in~\cite{DBLP:journals/mor/SandersSS09}. The paper
investigated several algorithms for the makespan minimization
problem and also presents an approximation scheme with absolute
competitive ratio $1+\eps$ and non-amortized migration factor
$2^{O((1/\eps) \cdot\log^{2}(1/\eps))}$. Skutella and Verschae
\cite{DBLP:journals/mor/SkutellaV16} studied a dynamic setting
with amortized migration, where jobs may also depart from the
instance. They achieved the same absolute competitive ratio, but
their algorithm needs an amortized migration of
$2^{O((1/\eps)\cdot
  \log^{2}(1/\eps))}$. Their algorithm also works for the Santa Claus (or
machine covering) problem, for which they show that even in the
static setting no algorithm has absolute competitive ratio
$1+\eps$ and a bounded migration factor. If one aims for a
polynomial migration factor for the Santa Claus problem, G\'alvez,
Soto, and Verschae presented an online variant of the LPT (least
processing time) algorithm achieving an absolute competitive ratio
of $4/3+\eps$ with non-amortized migration factor
$O((1/\eps^{3})\cdot
\log(1/\eps))$~\cite{DBLP:conf/esa/GalvezSV18}.

For the makespan minimization problem, if jobs can be preempted, Epstein and
Levin showed in~\cite{DBLP:journals/algorithmica/EpsteinL14} that an optimal
algorithm with a non-amortized migration factor of $1-1/m$ is achievable and the best
possible.

\paragraph*{Bin Packing:}
Epstein and Levin presented an approximation scheme with the same
ratio $1+\eps$ and the same non-amortized migration factor
$2^{O((1/\eps)\cdot \log^{2}(1/\eps))}$ as in the makespan
minimization for the bin packing problem
in~\cite{DBLP:journals/mp/EpsteinL09}. This result was improved by
Jansen and Klein in~\cite{DBLP:conf/icalp/JansenK13}, who
drastically reduced the migration factor to $O(1/\eps^{4})$.
Berndt, Jansen, and Klein used a similar approach to also handle
the dynamic bin packing problem, where items may also depart over
time~\cite{Berndt2018}. They also showed that a non-amortized
migration factor of $\Omega(1/\eps)$ is needed for this. A
generalized model, where an item $i$ has arbitrary movement costs
$c_{i}$~--~not necessarily linked to the size of an item~--~was
studied by Feldkord et
al.~\cite{DBLP:conf/icalp/FeldkordFGGKRW18}. They showed that for
$\alpha\approx 1{.}387$ and every $\eps > 0$, a competitive ratio
of $\alpha+\eps$ is achievable with migration $O(1/\eps^{2})$, but
no algorithm with migration $o(n)$ and ratio $\alpha-\eps$ exists.
By strengthening the lower bound of~\cite{Berndt2018}, they also
showed that an amortized migration factor of $\Omega(1/\eps)$ is
needed for the standard migration model~--~where the movement
costs $c_{i}$ are equal to the items sizes $s_{i}$~--~if one wants
to achieve a competitive ratio of $1+\eps$. A generalization of
bin packing, where $d$-dimensional cubic items are packed into as
few unit size cubes was studied by Epstein and
Levin~\cite{DBLP:journals/siamjo/EpsteinL13}.

\section{Non-amortized Migration in the Static Case}\label{sec:Static_Non_Amortized}
We begin our study by analyzing the static case with non-amortized migration. We
will first present a lower bound showing that no algorithm with constant
non-amortized migration factor can have a competitive ratio below $3/2$. Then, we present an
algorithm that achieves for all $\eps > 0$ a competitive ratio of $3/2+\eps$
with non-amortized migration factor $O(1/\eps)$.

\subsection{Lower Bound}
We start with a simple lower bound on the asymptotic competitive
ratio of all algorithms with constant non-amortized migration
factor. This lower bound can also be proved for a different
definition of the asymptotic competitive ratio $\alpha$, where we
require $\Opt(I) \leq \alpha \cdot \Alg(I) + o(\Opt(I))$.

\begin{proposition}
  \label{thm:wc_and_static_lb}
There is no algorithm for static online bin covering with a constant non-amortized migration factor and an asymptotic competitive ratio smaller than $3/2$.
\end{proposition}
\begin{proof}

  Let $N$ be an integer and $\Alg$ an algorithm for online bin covering with a
  non-amortized migration factor of $\beta$. Furthermore, let $R $ be the
  asymptotic competitive ratio of $\Alg$, \ie we have $\Opt(I)\leq R\cdot
  \Alg(I)+c$ for some constant value $c$. We construct an instance that is
  divided into two phases, in each of which of equal size items
  arrive. The instance corresponding to the first phase is called $I_1$ and the
  one corresponding to both phases $I_2$. The optimal (offline) objective value
  of $I_1$ and $I_2$ will be a multiple of $N$ and we will show that the ratio
  between this value and the objective value that is achieved by $\Alg$ is at
  least $3/2$ either for $I_1$ or for~$I_2$.

  In both phases $6N$ items arrive, in the first they have size $1-\eps$ and in
  the second one size $\eps$. We choose the parameter $\eps$ such that the items
  from the second phase do not fill a whole bin and such that they cannot be
  used to migrate items from the first, more precisely, let $\eps = \min\set{(2\beta + 2)^{-1},
    (12N)^{-1}}$. Note that the optimal objective value for $I_1$ and $I_2$ is
  $3N$ and $6N$ respectively. See Fig.~\ref{fig:lb_static} for feasible packings
  with these objective values. The optimality follows as each covered bin must
  have at least two items since the sizes of the items are strictly smaller than
  $1$. As each of the $\Alg(I_1)$ covered bins has to contain at least $2$
  items from the first phase, there are at most $6N-2\Alg(I_{1})$ bins each of
  which with exactly one item of the first phase. Furthermore, items from the
  first phase cannot be migrated during the second, as
  $\eps \cdot \beta \leq (2\beta+2)^{-1}\cdot \beta  < 1/2 < 11/12 \leq 1-\eps$.

Hence, we have $\Alg(I_2)\leq \Alg(I_1) + (6N -
2\Alg(I_1))=6N-\Alg(I_{1})$. We thus have
$\Alg(I_{1})+\Alg(I_{2})\leq 6N$. As the asymptotic competitive
ratio of $\Alg$ is $R$, we can also conclude that $3N =
\Opt(I_{1}) \leq R\cdot \Alg(I_{1})+c$ and $6N = \Opt(I_{2})\leq
R\cdot \Alg(I_{2})+c$. Summing these two inequalities gives $9N
\leq R(\Alg(I_{1})+\Alg(I_{2}))+2c$. As
$\Alg(I_{1})+\Alg(I_{2})\leq 6N$, we have $9N\leq 6RN +2c$ and
thus $R\geq 3/2 - (2c)/(6N)$. Letting $N$ grow to infinity thus
yields $R\geq 3/2$.
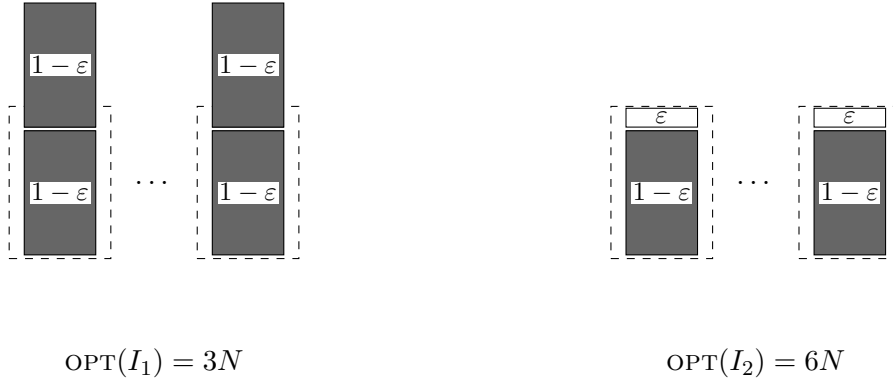
\begin{figure}[h]
  \centering
\begin{tikzpicture}
  \drawBin{2}{}{b1}

  \drawBin{4.5}{}{b2}
  \node at ($ (b1)!0.5!(b2)$) {\ldots};
  \addToBin{2}{0}{.85}{\tikz\node[fill=white,inner sep=.5]{\small $1-\eps$};}{fill=big}{}
  \addToBin{2}{.85}{.85}{\tikz\node[fill=white,inner sep=.5]{\small $1-\eps$};}{fill=big}{}
  \addToBin{4.5}{0}{.85}{\tikz\node[fill=white,inner sep=.5]{\small $1-\eps$};}{fill=big}{}
  \addToBin{4.5}{.85}{.85}{\tikz\node[fill=white, inner sep=.5]{\small $1-\eps$};}{fill=big}{}
  \node[yshift=-2.4cm] at ($ (b1)!0.5!(b2)$) {$\Opt(I_{1})=3N$};

  \drawBin{10}{}{b3}
  \drawBin{12.5}{}{b4}
  \node at ($ (b3)!0.5!(b4)$) {\ldots};
  \addToBin{10}{0}{.85}{\tikz\node[fill=white, inner sep=.5]{\small $1-\eps$};}{fill=big}{}
  \addToBin{10}{.85}{.15}{\small $\eps$}{fill=small}{}
  \addToBin{12.5}{0}{.85}{\tikz\node[fill=white, inner sep=.5]{\small $1-\eps$};}{fill=big}{}
  \addToBin{12.5}{.85}{.15}{\small $\eps$}{fill=small}{}
  \node[yshift=-2.4cm] at ($ (b3)!0.5!(b4)$) {$\Opt(I_{2})=6N$};

\end{tikzpicture}

  \caption{Feasible packings of $I_1$ and $I_2$}
  \label{fig:lb_static}
\end{figure}

\end{proof}

\subsection{Upper Bound}
We will now give our algorithm $\Alg$ for this scenario. In
addition to the instance $I$, a parameter $\eps > 0$ is also given
that regulates the asymptotic competitive ratio and the used
migration.  The assumption $\eps \leq 0.5$ is justified as for
$\eps \geq 0.5$, the result follows by the online algorithm with
an asymptotic competitive ratio of $2$ presented
in~\cite{DBLP:journals/dam/CsirikT88}, or by using the algorithm
below with $\eps = \frac 12$.

\begin{theorem}
  \label{thm:wc_and_static_alg}
For each $\eps \in (0,0.5]$, there is an algorithm $\Alg$ for
static online bin covering with polynomial running time, an
asymptotic competitive ratio of $1.5 + \eps$ with an additive
constant of $3$, and a non-amortized migration factor of
$\Oh(1/\eps)$.
\end{theorem}

The algorithm distinguishes between big, medium and small items.
For each item, it calls a corresponding insertion procedure based
on this classification into three classes. An item $i$ is called
\emph{big} if $s(i)\in(0.5,1]$, \emph{medium} if
$s(i)\in(\eps,0.5]$, and \emph{small} otherwise. For a bin $B$,
let $\fsmall(B)$ be the set of small items of $B$. We define
$\fmedium(B)$ and $\fbig(B)$ accordingly and also extend these
notions to sets of bins $\mathcal{B}$. Furthermore, we call a
covered bin \emph{barely covered}, if removing the biggest item of
the smallest class of items, \ie big, medium or small, contained
in the bin, results in the bin not being covered anymore. For
instance, consider a bin containing four items with sizes $0.65$,
$0.3$, $\eps$ and $0.25\eps$, for $\eps \geq 0.04$. This bin
contains items of all three classes if $\eps < 0.3$. In this case,
the biggest item of the smallest class has size $\eps$, and if
$\eps = 0.1$, the bin is indeed barely covered. However, if
$\eps=0.22$, the bin is not barely covered. If $\eps=0.3$, the bin
only has items of two classes, and it is not barely covered, since
removing one item of size $0.3$ results in a total size above $1$.
Note that showing that removing an arbitrary item of the smallest
class of items for a given covered bin results in an uncovered bin
is sufficient for showing that it is barely covered.

Let $B$ be a barely covered bin. If $B$ contains at most one big
item, its load is bounded from above by $1.5$, and if $B$
additionally contains no medium item the bound is reduced to $1 +
\eps$. This holds due to the following. Since the big item has
size below $1$, the bin contains at least one medium or small
item. If the bin has no small items, removing the largest medium
item reduces the load to below $1$, and together with the medium
item the load is below $1.5$. If it has a small item, a similar
calculation shows that the load is below $1+\eps$.

The last two types of bins are benign in the sense that they allow
analysis using arguments that are based on sizes. This is not the
case for bins containing two big items. Such bins could have a
size arbitrarily close to $2$ (even if they have no other items).
Bins of this type are needed for instances with many big items
(for an example we refer to the construction of the lower bound in
Proposition~\ref{thm:wc_and_static_lb}). However, they cause
problems not only because they are wastefully packed, but also
because they exclusively contain big items that should only be
moved if suitably large items arrive in order to bound migration.
The basic idea of the algorithm is to balance the number of bins
containing two big items and the number of bins containing one big
item and no medium items. The two numbers will be roughly the
same, which is obtained using migration on arrivals of big and
medium items. As described in the previous paragraph, the
guarantees of these bins that are based on loads cancel each other
out in the sense that an average load not exceeding $1.5+\eps/2$
can be achieved. In order to keep the number of bins with two big
items in check, our algorithm will only produces very few bin
types and we will maintain several invariants. This structured
approach allows us also to bound the migration needed. We
elaborate on the details of the algorithm.

\paragraph{Bin Types and Invariants.}

We distinguish different types of bins packed by the algorithm.
See Fig.~\ref{fig:types} for an overview of covered bins (only).

\begin{figure}
  \centering
\begin{tikzpicture}
  \drawBin{0}{}{bb}
  \addToBin{0}{0}{.6}{}{fill=big}{bb1}
  \addToBin{0}{.6}{.6}{}{fill=big}{bb2}
  \node[below=of bb, yshift=-1cm]{$\BBbins$};

  \drawBin{2}{}{bm}
  \addToBin{2}{0}{.6}{}{fill=big}{}
  \addToBin{2}{.6}{.25}{}{fill=medium}{}
  \addToBin{2}{.85}{.2}{}{fill=medium}{}
  \node[below=of bm, yshift=-1cm]{$\BMbins$};

  \drawBin{4}{}{bs}
  \addToBin{4}{0}{.6}{}{fill=big}{}
  \addToBin{4}{.6}{.1}{}{fill=small}{}
  \addToBin{4}{.7}{.1}{}{fill=small}{}
  \addToBin{4}{.8}{.12}{}{fill=small}{}
  \addToBin{4}{.92}{.13}{}{fill=small}{}
  \node[below=of bs, yshift=-1cm]{$\BSfullbins$};

  \drawBin{6}{}{m}
  \foreach \i in {0,1,...,4}{
    \addToBin{6}{\i*.22}{.22}{}{fill=medium}{}
    }
  \node[below=of m, yshift=-1cm]{$\Mfullbins$};

  \drawBin{8}{}{s}
  \foreach \i in {0,1,...,8}{
    \addToBin{8}{\i*.12}{.12}{}{fill=small}{}
    }
  \node[below=of s, yshift=-1cm]{$\Sfullbins$};
\end{tikzpicture}

  \caption{Different types of bins produced by the algorithm containing
    big items in dark gray, medium items in light gray, and
    small items in white.}
  \label{fig:types}
\end{figure}
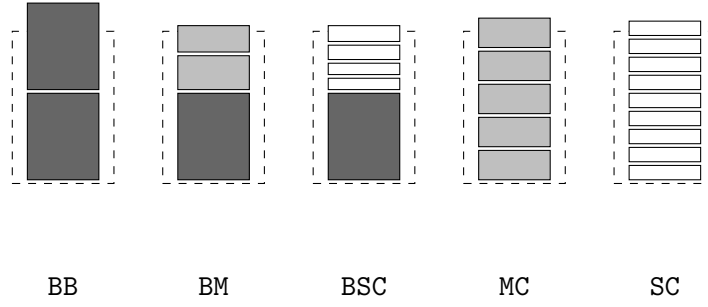
The bins are partitioned into bins containing two big items (and
no other items) $\BBbins$; barely covered bins containing one big
item and some medium items $\BMbins$; {\em barely covered} bins
containing one big item and some small items $\BSfullbins$; bins
that are not covered ({\em partially covered}) and contain one big
item and no medium items (but it could contain small items)
$\BSwaitbins$; and bins that are at most barely covered (they are
barely-covered, or not covered) and exclusively contain small or
medium items $\Sbins$ or $\Mbins$ respectively. Furthermore, let
$\Mfullbins\subseteq\Mbins$ and $\Sfullbins\subseteq\Sbins$ be the
corresponding subsets of barely covered bins (while $M\setminus
MC$ and $S\setminus SC$ are sets of bins that are not covered). We
denote the (disjoint) union of $\BSfullbins$ and $\BSwaitbins$ as
$\BSbins$. The set of bins packed by the algorithm (covered or not
covered) is denoted as $\bins$. All bins covered by the algorithm
are in fact barely covered, and no bin (covered or not) contains
items of all three classes. Among bins that are not covered, there
are no bins containing a big item and a non-empty set of medium
items.

We will now introduce the invariants needed. We first give a
formal definition and will sometimes follow this by a mnemonic in
italics to help the reader in remembering the purpose of this
invariant. The first invariant the algorithm ensures that this bin
structure is maintained and the second invariant was already
indicated above ($A\ \dot{\cup}\ B$ denotes the disjoint union of
$A$ and $B$):
\begin{enumerate}[label=I\arabic*, series=staticinvariants]
\item The solution has the proposed bin type structure, \ie
$$\bins = \BBbins\ \dot{\cup}\ \BMbins\ \dot{\cup}\ \BSbins\
\dot{\cup}\ \Mbins\ \dot{\cup}\ \Sbins ,$$  and $\BSbins =
\BSfullbins\ \dot{\cup}\ \BSwaitbins$.
\label{Invariant:Binstructure} \item The sets $\BBbins$ and
$\BSbins$ are balanced in size, \ie
  $\abs[\big]{|\BBbins|-|\BSbins|} \leq 1 $.
 \label{Invariant:BB=BS}
\end{enumerate}
Recall that the average load of any bin of $\BBbins$ together with
any bin of $\BSbins$ does not exceed $1.5+\eps/2$, so by invariant
\ref{Invariant:BB=BS} the average load of bins of $\BBbins \cup
\BSbins$ possibly excluding one bin is not larger than $1.5
+\eps/2$.

Therefore we have $\Alg(I) = |\BBbins| + |\BMbins| + |\BSfullbins|
+ |\Mfullbins| + |\Sfullbins|$. Furthermore, we have several
invariants concerning the distribution of items to different bin
types. The intuition behind these invariants is always the same:
We have to ensure that no other algorithm is able to use the small
and medium items to cover a much large number of bins.
\begin{enumerate}[resume*=staticinvariants]
\item The big items contained in $\BMbins$ are at least as big as the ones in
  $\BSfullbins$ which in turn are at least as big as the ones in $\BSwaitbins$,
  and the smallest big items are placed in $\BBbins$, \ie $s(i) \geq s(i')$ for
  each (i) $i\in\fbig(\BMbins)$ and $i'\in\fbig(\BSbins\cup\BBbins)$; each (ii)
  $i\in\fbig(\BMbins\cup\BSfullbins)$ and $i'\in\fbig(\BSwaitbins\cup\BBbins)$;
  each (iii)
  $i\in\fbig(\BMbins\cup\BSbins)$ and $i'\in\fbig(\BBbins)$.
  \emph{The smallest biggest items are in $\BBbins$ and the largest in
    $\BMbins$. Informally, $\BBbins\leq \BSwaitbins\leq \BSfullbins\leq \BMbins$.} \label{Invariant:BigItems}
\item The union of items in $\Mbins$ cannot be used to cover a bin
together with a big
  item from $\BSbins$ or $\BBbins$, \ie $\BSbins\cup\BBbins \neq \emptyset
  \implies s(\Mbins) < 1 - s_{\max}(\BSbins\cup\BBbins)$. In the case $\BSbins\cup\BBbins = \emptyset$, there
    are no restrictions on $s(\Mbins)$.
  \emph{We can not build another $\BMbins$ bin from $\Mbins$ and $\BSbins\cup
    \BBbins$.}\label{Invariant:MediumBins}
\item If a bin containing only small items exists, all bins in $\BSbins$ are
  covered, \ie $|\Sbins| > 0 \implies |\BSwaitbins| = 0$.
  \emph{We do not have unnecessary $\Sbins$ bins in the sense that if they exist, their items cannot be used elsewhere.}
  \label{Invariant:SmallBins}
\end{enumerate}

Lastly, there are some bin types with bins that are not covered,
and we have to ensure that they are not wastefully packed:
\begin{enumerate}[resume*=staticinvariants]
\item If there are small items in $\BSwaitbins$, they are all included in the
  bin containing the biggest item in $\BSwaitbins$.
  \emph{The next bin of $\BSwaitbins$ that is planned to be covered by adding small items to it is the most
    loaded one.}
  \label{Invariant:BSPonlyonewithsmall}
\item Each of $\Sbins$ and $\Mbins$ will contain at most one bin
that is not covered, \ie $|\Sfullbins|\ge |\Sbins|-1$ and
$|\Mfullbins|\ge |\Mbins|-1$.
  \emph{As long as medium and small items arrive and they are packed into $\Mbins$ and $\Sbins$, bins of these sets are opened one by one.}
  \label{Invariant:SmallandMedium}
\end{enumerate}

This concludes the definition of all invariants. It is easy to see
that the invariants all hold in the beginning when no item has
arrived yet. Next, we describe the insertion procedures and argue
that the invariants are maintained.

\paragraph{Examples.}

\begin{figure}
  \centering
\begin{tikzpicture}
\pgfmathsetmacro{\sgap}{1.6}
\pgfmathsetmacro{\mgap}{1.7}
\pgfmathsetmacro{\bgap}{1.85}

\drawBin{0}{}{}
\addToBin{0}{0}{.55}{}{fill=big}{}
\addToBin{0}{.55}{.55}{}{fill=big}{}
  
\drawBin{\sgap}{}{bb}
\addToBin{\sgap}{0}{.55}{}{fill=big}{}
\addToBin{\sgap}{.55}{.55}{}{fill=big}{}

\draw[decorate,decoration={brace,amplitude=6pt,raise=0pt, mirror},yshift=0pt] (-0.2,-0.2) -- (\sgap + 1.2,-0.2) node [midway,yshift=-13pt]{$\BBbins$};

\drawBin{1*\sgap + 0*\mgap + 1*\bgap}{}{}
\addToBin{1*\sgap + 0*\mgap + 1*\bgap}{0}{.65}{}{fill=big}{}

\drawBin{2*\sgap + 0*\mgap + 1*\bgap}{}{}
\addToBin{2*\sgap + 0*\mgap + 1*\bgap}{0}{.65}{}{fill=big}{}
\addToBin{2*\sgap + 0*\mgap + 1*\bgap}{.65}{.07}{}{fill=small}{}
\addToBin{2*\sgap + 0*\mgap + 1*\bgap}{.72}{.12}{}{fill=small}{}

\draw[decorate,decoration={brace,amplitude=6pt,raise=0pt, mirror},yshift=0pt] (1*\sgap + 0*\mgap + 1*\bgap-0.2,-0.2) -- (2*\sgap + 0*\mgap + 1*\bgap + 1.2,-0.2) node [midway,yshift=-13pt]{$\BSwaitbins$};

\drawBin{2*\sgap + 1*\mgap + 1*\bgap}{}{}
\addToBin{2*\sgap + 1*\mgap + 1*\bgap}{0}{.7}{}{fill=big}{}
\addToBin{2*\sgap + 1*\mgap + 1*\bgap}{.7}{.08}{}{fill=small}{}
\addToBin{2*\sgap + 1*\mgap + 1*\bgap}{.78}{.1}{}{fill=small}{}
\addToBin{2*\sgap + 1*\mgap + 1*\bgap}{.88}{.2}{}{fill=small}{}

\draw[decorate,decoration={brace,amplitude=6pt,raise=0pt, mirror},yshift=0pt] ( 2*\sgap + 1*\mgap + 1*\bgap -0.2,-0.2) -- (2*\sgap + 1*\mgap + 1*\bgap + 1.2,-0.2) node [midway,yshift=-13pt]{$\BSfullbins$};

\draw[decorate,decoration={brace,amplitude=6pt,raise=0pt, mirror},yshift=0pt] (1*\sgap + 0*\mgap + 1*\bgap-0.2,-0.8) -- (2*\sgap + 1*\mgap + 1*\bgap + 1.2,-0.8) node [midway,yshift=-13pt]{$\BSbins$};

\drawBin{2*\sgap + 1*\mgap + 2*\bgap}{}{}
\addToBin{2*\sgap + 1*\mgap + 2*\bgap}{0}{.7}{}{fill=big}{}
\addToBin{2*\sgap + 1*\mgap + 2*\bgap}{.7}{.25}{}{fill=medium}{}
\addToBin{2*\sgap + 1*\mgap + 2*\bgap}{.95}{.4}{}{fill=medium}{}

\drawBin{3*\sgap + 1*\mgap + 2*\bgap}{}{}
\addToBin{3*\sgap + 1*\mgap + 2*\bgap}{0}{.9}{}{fill=big}{}
\addToBin{3*\sgap + 1*\mgap + 2*\bgap}{.9}{.45}{}{fill=medium}{}

\draw[decorate,decoration={brace,amplitude=6pt,raise=0pt, mirror},yshift=0pt] ( 2*\sgap + 1*\mgap + 2*\bgap -0.2,-0.2) -- (3*\sgap + 1*\mgap + 2*\bgap + 1.2,-0.2) node [midway,yshift=-13pt]{$\BMbins$};

\drawBin{3*\sgap + 1*\mgap + 3*\bgap}{}{}
\addToBin{3*\sgap + 1*\mgap + 3*\bgap}{0}{.29}{}{fill=medium}{}

\draw[decorate,decoration={brace,amplitude=6pt,raise=0pt, mirror},yshift=0pt] ( 3*\sgap + 1*\mgap + 3*\bgap -0.2,-0.2) -- (3*\sgap + 1*\mgap + 3*\bgap + 1.2,-0.2) node [midway,yshift=-13pt]{$\Mbins$};

\begin{scope}[yshift=-5cm]

\drawBin{0}{}{}
\addToBin{0}{0}{.52}{}{fill=big}{}
\addToBin{0}{.52}{.52}{}{fill=big}{}
  
\drawBin{\sgap}{}{}
\addToBin{\sgap}{0}{.52}{}{fill=big}{}
\addToBin{\sgap}{.52}{.52}{}{fill=big}{}

\drawBin{2*\sgap + 0*\bgap}{}{}
\addToBin{2*\sgap + 0*\bgap}{0}{.58}{}{fill=big}{}
\addToBin{2*\sgap + 0*\bgap}{.58}{.58}{}{fill=big}{}

\draw[decorate,decoration={brace,amplitude=6pt,raise=0pt, mirror},yshift=0pt] (-0.2,-0.2) -- (2*\sgap + 0*\mgap + 0*\bgap + 1.2,-0.2) node [midway,yshift=-13pt]{$\BBbins$};

\drawBin{2*\sgap + 1*\bgap}{}{}
\addToBin{2*\sgap + 1*\bgap}{0}{.6}{}{fill=big}{}
\addToBin{2*\sgap + 1*\bgap}{.6}{.08}{}{fill=small}{}
\addToBin{2*\sgap + 1*\bgap}{.68}{.08}{}{fill=small}{}
\addToBin{2*\sgap + 1*\bgap}{.76}{.1}{}{fill=small}{}
\addToBin{2*\sgap + 1*\bgap}{.86}{.1}{}{fill=small}{}
\addToBin{2*\sgap + 1*\bgap}{.96}{.13}{}{fill=small}{}

\drawBin{3*\sgap + 1*\bgap}{}{}
\addToBin{3*\sgap + 1*\bgap}{0}{.6}{}{fill=big}{}
\addToBin{3*\sgap + 1*\bgap}{.6}{.12}{}{fill=small}{}
\addToBin{3*\sgap + 1*\bgap}{.72}{.12}{}{fill=small}{}
\addToBin{3*\sgap + 1*\bgap}{.84}{.14}{}{fill=small}{}
\addToBin{3*\sgap + 1*\bgap}{.98}{.2}{}{fill=small}{}

\draw[decorate,decoration={brace,amplitude=6pt,raise=0pt, mirror},yshift=0pt] ( 2*\sgap + 1*\bgap-0.2,-0.2) -- (3*\sgap + 1*\bgap + 1.2,-0.2) node [midway,yshift=-13pt]{$\BSbins = \BSfullbins$};

\drawBin{3*\sgap + 2*\bgap}{}{}
\addToBin{3*\sgap + 2*\bgap}{0}{.36}{}{fill=medium}{}

\draw[decorate,decoration={brace,amplitude=6pt,raise=0pt, mirror},yshift=0pt] (3*\sgap + 2*\bgap -0.2,-0.2) -- (3*\sgap + 2*\bgap + 1.2,-0.2) node [midway,yshift=-13pt]{$\Mbins$};

\drawBin{3*\sgap +  3*\bgap}{}{}
\addToBin{3*\sgap +  3*\bgap}{0}{.1}{}{fill=small}{}
\addToBin{3*\sgap +  3*\bgap}{.1}{.1}{}{fill=small}{}
\addToBin{3*\sgap +  3*\bgap}{.2}{.1}{}{fill=small}{}
\addToBin{3*\sgap +  3*\bgap}{.3}{.1}{}{fill=small}{}
\addToBin{3*\sgap +  3*\bgap}{.4}{.15}{}{fill=small}{}
\addToBin{3*\sgap +  3*\bgap}{.55}{.2}{}{fill=small}{}
\addToBin{3*\sgap +  3*\bgap}{.75}{.2}{}{fill=small}{}
\addToBin{3*\sgap +  3*\bgap}{.95}{.2}{}{fill=small}{}

\drawBin{4*\sgap +  3*\bgap}{}{}
\addToBin{4*\sgap +  3*\bgap}{0}{.1}{}{fill=small}{}
\addToBin{4*\sgap +  3*\bgap}{.1}{.1}{}{fill=small}{}
\addToBin{4*\sgap +  3*\bgap}{.2}{.15}{}{fill=small}{}
\addToBin{4*\sgap +  3*\bgap}{.35}{.2}{}{fill=small}{}
\addToBin{4*\sgap +  3*\bgap}{.55}{.2}{}{fill=small}{}
\addToBin{4*\sgap +  3*\bgap}{.75}{.2}{}{fill=small}{}
\addToBin{4*\sgap +  3*\bgap}{.95}{.2}{}{fill=small}{}

\draw[decorate,decoration={brace,amplitude=6pt,raise=0pt, mirror},yshift=0pt] (3*\sgap + 3*\bgap-0.2,-0.2) -- (4*\sgap + 3*\bgap + 1.2,-0.2) node [midway,yshift=-13pt]{$\Sfullbins$};

\drawBin{4*\sgap + \mgap +  3*\bgap}{}{}
\addToBin{4*\sgap + \mgap +  3*\bgap}{0}{.1}{}{fill=small}{}
\addToBin{4*\sgap + \mgap +  3*\bgap}{.1}{.15}{}{fill=small}{}
\addToBin{4*\sgap + \mgap +  3*\bgap}{.25}{.2}{}{fill=small}{}

\draw[decorate,decoration={brace,amplitude=6pt,raise=0pt, mirror},yshift=0pt] (3*\sgap + 3*\bgap-0.2,-0.8) -- (4*\sgap + \mgap + 3*\bgap + 1.2,-0.8) node [midway,yshift=-13pt]{$\Sbins$};

\end{scope}

\begin{scope}[yshift=-10cm]

\drawBin{0}{}{}
\addToBin{0}{0}{.7}{}{fill=big}{}
\addToBin{0}{.7}{.25}{}{fill=medium}{}
\addToBin{0}{.95}{.4}{}{fill=medium}{}

\drawBin{\sgap}{}{}
\addToBin{\sgap}{0}{.8}{}{fill=big}{}
\addToBin{\sgap}{.8}{.49}{}{fill=medium}{}
  
\drawBin{2*\sgap + 0*\bgap}{}{}
\addToBin{2*\sgap + 0*\bgap}{0}{.9}{}{fill=big}{}
\addToBin{2*\sgap + 0*\bgap}{.9}{.3}{}{fill=medium}{}

\draw[decorate,decoration={brace,amplitude=6pt,raise=0pt, mirror},yshift=0pt] (-0.2,-0.2) -- (2*\sgap + 0*\mgap + 0*\bgap + 1.2,-0.2) node [midway,yshift=-13pt]{$\BMbins$};

\drawBin{2*\sgap + 1*\bgap}{}{}
\addToBin{2*\sgap + 1*\bgap}{0}{.24}{}{fill=medium}{}
\addToBin{2*\sgap + 1*\bgap}{0.24}{.25}{}{fill=medium}{}
\addToBin{2*\sgap + 1*\bgap}{0.49}{.31}{}{fill=medium}{}
\addToBin{2*\sgap + 1*\bgap}{0.8}{.33}{}{fill=medium}{}

\drawBin{3*\sgap + 1*\bgap}{}{}
\addToBin{3*\sgap + 1*\bgap}{0}{.4}{}{fill=medium}{}
\addToBin{3*\sgap + 1*\bgap}{0.4}{.45}{}{fill=medium}{}
\addToBin{3*\sgap + 1*\bgap}{0.85}{.47}{}{fill=medium}{}

\draw[decorate,decoration={brace,amplitude=6pt,raise=0pt, mirror},yshift=0pt] ( 2*\sgap + 1*\bgap-0.2,-0.2) -- (3*\sgap + 1*\bgap + 1.2,-0.2) node [midway,yshift=-13pt]{$\Mfullbins$};

\drawBin{3*\sgap + 1*\mgap + 1*\bgap}{}{}
\addToBin{3*\sgap + 1*\mgap  + 1*\bgap}{0}{.34}{}{fill=medium}{}
\addToBin{3*\sgap + 1*\mgap  + 1*\bgap}{0.34}{.36}{}{fill=medium}{}

\draw[decorate,decoration={brace,amplitude=6pt,raise=0pt, mirror},yshift=0pt] ( 2*\sgap + 1*\bgap-0.2,-0.8) -- (3*\sgap + 1*\mgap  + 1*\bgap + 1.2,-0.8) node [midway,yshift=-13pt]{$\Mbins$};

\drawBin{3*\sgap + 1*\mgap + 2*\bgap}{}{}
\addToBin{3*\sgap + 1*\mgap + 2*\bgap}{0}{.1}{}{fill=small}{}
\addToBin{3*\sgap + 1*\mgap + 2*\bgap}{.1}{.1}{}{fill=small}{}
\addToBin{3*\sgap + 1*\mgap + 2*\bgap}{.2}{.1}{}{fill=small}{}
\addToBin{3*\sgap + 1*\mgap + 2*\bgap}{.3}{.1}{}{fill=small}{}
\addToBin{3*\sgap + 1*\mgap + 2*\bgap}{.4}{.15}{}{fill=small}{}
\addToBin{3*\sgap + 1*\mgap + 2*\bgap}{.55}{.2}{}{fill=small}{}
\addToBin{3*\sgap + 1*\mgap + 2*\bgap}{.75}{.2}{}{fill=small}{}
\addToBin{3*\sgap + 1*\mgap + 2*\bgap}{.95}{.2}{}{fill=small}{}

\draw[decorate,decoration={brace,amplitude=6pt,raise=0pt, mirror},yshift=0pt] (3*\sgap + 1*\mgap + 2*\bgap-0.2,-0.2) -- (3*\sgap + 1*\mgap + 2*\bgap + 1.2,-0.2) node [midway,yshift=-13pt]{$\Sbins = \Sfullbins$};

\draw (1,-1.9) -- (14,-1.9);
\end{scope}

 \begin{scope}[yshift=-16.5cm]

\drawBin{0*\mgap }{}{b1}
\addToBin{0*\mgap}{0}{.9}{}{fill=big}{}
\addToBin{0*\mgap}{.9}{.9}{}{fill=big}{}

\node[below=of b1, yshift=-0.05cm]{$B_1$};
  
\drawBin{\mgap}{}{b2}
\addToBin{\mgap}{0}{.6}{}{fill=big}{}

\node[below=of b2, yshift=-0.05cm]{$B_2$};

\drawBin{2*\mgap }{}{b3}
\addToBin{2*\mgap}{0}{.7}{}{fill=big}{}
\addToBin{2*\mgap}{.7}{.2}{}{fill=small}{}

\node[below=of b3, yshift=-0.05cm]{$B_3$};

\drawBin{3*\mgap }{}{b4}
\addToBin{3*\mgap}{0}{.65}{}{fill=big}{}
\addToBin{3*\mgap}{.65}{.15}{}{fill=small}{}
\addToBin{3*\mgap}{.8}{.15}{}{fill=small}{}

\node[below=of b4, yshift=-0.05cm]{$B_4$};

\drawBin{4*\mgap }{}{b5}
\addToBin{4*\mgap}{0}{.6}{}{fill=big}{}
\addToBin{4*\mgap}{.6}{.25}{}{fill=medium}{}
\addToBin{4*\mgap}{.85}{.1}{}{fill=small}{}

\node[below=of b5, yshift=-0.05cm]{$B_5$};

\drawBin{5*\mgap }{}{b6}
\addToBin{5*\mgap}{0}{.3}{}{fill=medium}{}
\addToBin{5*\mgap}{0.3}{.4}{}{fill=medium}{}

\node[below=of b6, yshift=-0.05cm]{$B_6$};

\drawBin{6*\mgap }{}{b7}
\addToBin{6*\mgap}{0}{.22}{}{fill=medium}{}
\addToBin{6*\mgap}{0.22}{.28}{}{fill=medium}{}
\addToBin{6*\mgap}{0.50}{.33}{}{fill=medium}{}

\node[below=of b7, yshift=-0.05cm]{$B_7$};

\drawBin{7*\mgap }{}{b8}
\addToBin{7*\mgap}{0}{.3}{}{fill=medium}{}
\addToBin{7*\mgap}{0.3}{.15}{}{fill=small}{}
\addToBin{7*\mgap}{0.45}{.15}{}{fill=small}{}
\addToBin{7*\mgap}{0.6}{.15}{}{fill=small}{}
\addToBin{7*\mgap}{0.75}{.2}{}{fill=small}{}
\addToBin{7*\mgap}{0.95}{.2}{}{fill=small}{}

\node[below=of b8, yshift=-0.05cm]{$B_8$};

\drawBin{8*\mgap }{}{b9}
\addToBin{8*\mgap}{0}{.1}{}{fill=small}{}
\addToBin{8*\mgap}{0.1}{.15}{}{fill=small}{}
\addToBin{8*\mgap}{0.25}{.15}{}{fill=small}{}
\addToBin{8*\mgap}{0.4}{.15}{}{fill=small}{}
\addToBin{8*\mgap}{0.55}{.15}{}{fill=small}{}
\addToBin{8*\mgap}{0.7}{.2}{}{fill=small}{}
\addToBin{8*\mgap}{0.9}{.2}{}{fill=small}{}
\addToBin{8*\mgap}{1.1}{.2}{}{fill=small}{}
\addToBin{8*\mgap}{1.3}{.2}{}{fill=small}{}

\node[below=of b9, yshift=-0.05cm]{$B_9$};

 \end{scope}

\end{tikzpicture}

  \caption{An exemplary packings with big items in dark gray, medium items in light gray, and
    small items in white.  The first three examples are consistent with the invariants and the last one is not.
  }
  \label{fig:examples_static}
\end{figure}
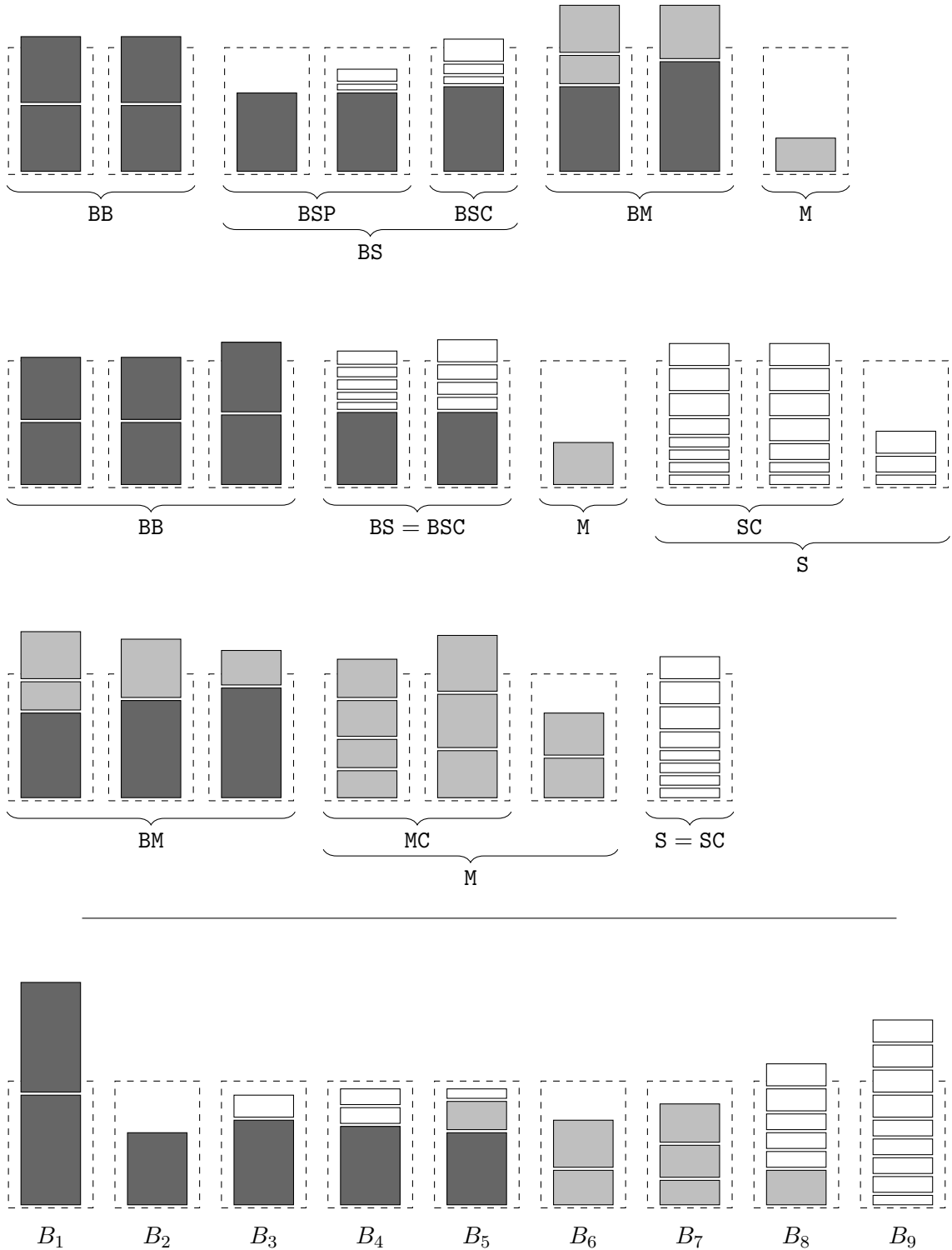  In the following, we briefly discuss several examples to build intuition for the invariants.
Figure \ref{fig:examples_static} shows four exemplary packings.
The first three are consistent with the invariants and the last one is not, and we first discuss the former three and then the latter.

As indicated in the captions, only the allowed bin types
(invariant \ref{Invariant:Binstructure}) are present in examples
1, 2, 3; and we have $|\BBbins| + 1 = |\BSbins|$, $|\BBbins| =
|\BSbins| + 1$ and $|\BBbins| = |\BSbins|$ in examples 1, 2, and
3, respectively, yielding invariant \ref{Invariant:BB=BS} for the
three first examples. Concerning the distribution of big items,
that is, invariant \ref{Invariant:BigItems}, note that the big
items get smaller when going from left to right in each of the
examples which corresponds to the correct distribution. In the
first two examples, there is exactly one bin in $\Mbins$
containing only one medium item, which already gives invariant
\ref{Invariant:SmallandMedium} with respect to $\Mbins$ and it is
easy to see that this invariant also holds for example 3. This
item is relatively big in the second example, however each big
item present in the solution is too small to be combined with the
medium item to form a covered bin. In the first example, the big
items are bigger, but the medium item is smaller, and the same
holds. Hence, invariant \ref{Invariant:MediumBins} holds, and this
is also the case for example 3, where each big item is contained
in $\BMbins$. Note that in the third example $\Mbins$ contains
multiple bins, which may only happen if
$\BSbins=\BBbins=\emptyset$, due to invariant
\ref{Invariant:MediumBins}. In examples 2 and 3, we have
$\BSwaitbins = \emptyset$ and hence invariants
\ref{Invariant:BSPonlyonewithsmall} and \ref{Invariant:SmallBins}
trivially hold. In example 1, on the other hand, there is exactly
one bin in $\BSwaitbins$ containing small items, and this bin also
contains a big item of maximal size among the ones in
$\BSwaitbins$, yielding \ref{Invariant:BSPonlyonewithsmall} for
this example, and we have $\Sbins=\emptyset$, yielding
\ref{Invariant:SmallBins}. Lastly, in each of the three examples
there is at most one bin in $\Sbins$ which is not covered, and
hence invariant \ref{Invariant:SmallandMedium} also holds with
respect to $\Sbins$.

The last example, on the other hand, is carefully designed such
that all invariants are infringed or meaningless: The two bins
$B_5$ and $B_8$ contain both small and medium items; and bin $B_9$
is overpacked (covered but not barely covered). Therefore
invariant \ref{Invariant:Binstructure} is infringed. Furthermore,
there are three bins containing one big and no medium items, but
only one bin with two big items. This is in conflict to invariant
\ref{Invariant:BB=BS}. The two biggest big items are combined in
bin $B_1$, violating invariant \ref{Invariant:BigItems}. Moreover,
the bins $B_6$ and $B_7$ both are not covered, and exclusively
contain medium items, which could be combined with several big
items to form barely covered bins. Hence, both
\ref{Invariant:SmallandMedium} and \ref{Invariant:MediumBins} are
infringed. There are two bins in $\BSwaitbins$ containing small
items, infringing invariant \ref{Invariant:BSPonlyonewithsmall};
and additionally $\Sbins\neq \emptyset$, violating invariant
\ref{Invariant:SmallBins}.

\paragraph{Insertion Procedures.}

We start with the definition of two simple auxiliary procedures used in the following:
\begin{itemize}
\item $\greedypush(i, \mathcal{B})$ is given an item $i$ and a set
of bins $\mathcal{B}$. If all the bins contained in $\mathcal{B}$
are covered, it creates a new bin containing item $i$, and
otherwise it inserts $i$ into the most loaded bin that is not
covered. \item $\greedypull(B,\mathcal{B})$ is given a bin $B$ and
a set of bins $\mathcal{B}$. It successively removes a largest
non-big item from a least loaded bin from
$\sett{B'\in\mathcal{B}}{\fsmall(B')\cup\fmedium(B')\neq\emptyset}$
and inserts it into $B$. This is repeated until either $B$ is
covered or $\mathcal{B}$ does not contain non-big items.
\end{itemize}

Consider one application of $\greedypull$ such that $B$ already
has a big item. The total size of moved items is smaller than $1$.

Both procedures are used to insert and repack non-big items. Note
that calling $\greedypush$ for a small item and bin set
$\BSwaitbins$ or $\Sbins$, or a medium item and bin set $\Mbins$,
the last two invariants \ref{Invariant:BSPonlyonewithsmall} and
\ref{Invariant:SmallandMedium} are maintained. For $\BSwaitbins$,
the most loaded bin always contains the largest big item, and if
there is at least one small item, such a bin is unique. It could
happen that as a result of inserting a small item into this bin of
$\BSwaitbins$ the bin is covered and moves to $\BSfullbins$.

For each insertion procedure, we will argue that the invariants
are maintained and focus on the critical ones, that is, in each
context the invariants, that are not discussed explicitly,
trivially hold. For example, we do not discuss
\ref{Invariant:Binstructure} in the following, because it will
always be easy to see that it is maintained.

\paragraph{Insertion of Small Items.}

If the arriving item $i^*$ is small, we call
$\greedypush(i^*,\BSwaitbins)$, if $\BSwaitbins \neq \emptyset$,
and $\greedypush(i^*,\Sbins)$ otherwise. Insertion into a bin of
$\BSwaitbins$ (the most loaded one) may lead to a covered bin, in
which case the bin becomes a bin of $\BSfullbins$ (but remains in
$\BSbins$. It is easy to verify, that all invariants, and
\ref{Invariant:SmallBins}, \ref{Invariant:BSPonlyonewithsmall} and
\ref{Invariant:SmallandMedium} in particular, are maintained by
this. Furthermore, there is no migration in this case. The
insertion of a medium or big item, however, is more complicated.

\paragraph{Insertion of Big Items.}

In the case that a big item $i^*$ arrives, we have to be careful
where we place it exactly, because, on the one hand, the
distributions of big and medium items, that is, invariants
\ref{Invariant:BigItems} and \ref{Invariant:MediumBins}, have to
be maintained, and, on the other hand, we have to balance out
$\BSbins$ and $\BBbins$ (\ref{Invariant:BB=BS}). We consider
placing the item in $\BMbins$, $\BSbins$ or $\BBbins$ in this
order, \ie we first try to insert $i^{*}$ into $\BMbins$, then
into $\BSbins$ and finally into $\BBbins$.
Figure~\ref{fig:static_insert_big} illustrates this process.

\begin{itemize}

\item{Insertion into ${\BMbins}$.}

We insert $i^*$ into $\BMbins$, if either $s(i^*) + s(\Mbins) \geq
1$ or $s(i^*) > s_{\min}(\fbig(\BMbins))$. Note that the first
condition implies $s(i^*) \geq
s_{\max}(\fbig(\BBbins\cup\BSbins))$, because
of~\ref{Invariant:MediumBins}, and therefore the insertion of
$i^*$ into $\BMbins$ maintains \ref{Invariant:BigItems} in both
situations. The second condition implies $\BMbins \neq \emptyset$,
because we set $s_{\min}(\emptyset) = +\infty$. In either of these
cases, we create a new bin $B^* = \set{i^*}$ and call
$\greedypull(B^*,\Mbins)$, thereby ensuring that
\ref{Invariant:MediumBins} is maintained if the new bin is
covered. If the first condition did hold, $B^*$ is covered
afterwards and we do nothing else. Otherwise, there is a bin
$B\in\BMbins$ containing a big item $i$ with $s(i) =
s_{\min}(\fbig(\BMbins)) < s(i^*)$, and we have $\Mbins =
\emptyset$. We remove $i$ from $B$, yielding $\Mbins = \set{B}$,
and call $\greedypull(B^*,\Mbins)$ a second time. Afterwards,
$B^*$ is covered, because $s(i^*) > s(i)$. Furthermore, $s(i) +
s(\Mbins) < 1$, because $B$ was barely covered before and the
biggest medium item was removed from $B$ due to the second call of
$\greedypull$. The item~$i$ is reinserted using a recursive call
to the procedure of inserting a big item. However, item $i$ will
not be considered for insertion into $\BMbins$, because neither
the first nor second condition holds for this item, and the other
insertion options have no recursive calls for insertion  into
$\BMbins$. It is easy to verify that the distribution of medium
items in $\Mbins$ (\ref{Invariant:SmallandMedium}) is maintained.

\item{Insertion into ${\BSbins}$.}

This step is possible for item $i^*$ that satisfies $s(i^*) +
s(\Mbins) < 1$ and $s(i^*) \leq s_{\min}(\fbig(\BMbins))$. Thus,
$\BMbins$ will have the largest big items as required in Invariant
\ref{Invariant:BigItems} after the insertion is performed. In this
case there is no recursive call for inserting a big item.

We insert $i^*$ into $\BSbins$, if either $s(i^*) >
s_{\min}(\fbig(\BSbins))$ or the following two conditions hold:
$s(i^*) \geq s_{\max}(\fbig(\BBbins))$ and $ |\BSbins| \leq
|\BBbins|$. Note that $s(i^*) \geq s_{\max}(\fbig(\BBbins))$
trivially holds, if $\BBbins = \emptyset$. Inserting $i^*$ into
$\BSbins$ under these conditions already ensures the correct
distribution of big items (\ref{Invariant:BigItems}) with respect
to $\BSbins$ and $\BBbins$, but we still have to be careful
concerning the distribution within the two subsets of $\BSbins$.
The procedure is divided into three simple steps. As a first step,
we create a new bin $B^* = \set{i^*}$ and call
$\greedypull(B^*,\Sbins)$. No matter whether $B^*$ is now covered
or not, Invariant \ref{Invariant:SmallBins} is satisfied as either
$B^*$ is covered and therefore $\BSwaitbins = \emptyset$ both
before and after the call, or $B^*$ is not covered but now $\Sbins
= \emptyset$ (it is possible that both will hold). Note that all
properties of the invariants are satisfied, if $B^*$ is already
covered. In particular, Invariant \ref{Invariant:BigItems} holds
within $\BSbins$ because $\BSwaitbins=\emptyset$. In the remainder
of the second step of the insertion into ${\BSbins}$ algorithm we
deal with the case that $B^*$ is not covered.

Let $\Xbins$ denote the set of bins $B\in\BSbins$ that include
small items as well as a big item $i$ with $s(i)<s(i^*)$. Recall
that any bin of $\BSfullbins$ has at least one small item, while
at most one bin of $\BSwaitbins$ has small items.

First, assume that $\Xbins=\emptyset$ but $B^*$ is not covered.
There are two cases. In the first case, at least one item of
$\Sbins$ was moved. In this case before we started dealing with
$i^*$, the set $\BSwaitbins$ was empty, and now $B^*$ is the
unique bin of $\BSwaitbins$, and its big item is not larger than
those of $\BSfullbins$ (if the last set is not empty) so the
invariants \ref{Invariant:BigItems} and
\ref{Invariant:BSPonlyonewithsmall} are maintained. In the second
case, $\Sbins$ was empty, and $B^*$ is now a bin of $\BSwaitbins$
with only a big item. Since $\Xbins = \emptyset$, adding $B^*$ to
$\BSwaitbins$ maintains the invariants \ref{Invariant:BigItems}
and \ref{Invariant:BSPonlyonewithsmall}. Thus, it is left to deal
with the case $\Xbins \neq \emptyset$. In the remainder of the
second step of the insertion into ${\BSbins}$ algorithm we deal
with the case that $\Xbins$ is not empty.

As $B^{*}$ is not yet covered, we now have $\Sbins = \emptyset$,
and this might have been the case before the call of
$\greedypull$, in particular if we had $\BSwaitbins\neq\emptyset$.
Due to the existence of a big item that is smaller than $i^*$ in
$\BSbins$ (such items exist in all bins of $\Xbins$), we have to
be careful in order to maintain the correct distribution of big
and small items inside of $\BSbins$ (\ref{Invariant:BigItems} and
\ref{Invariant:BSPonlyonewithsmall}).

In the second step, we construct a set of bins
$\tilde{\mathtt{B}}\subseteq\BSbins$ from which small items are
removed in order to cover $B^{*}$. If $\BSwaitbins \cap \Xbins
\neq \emptyset$, this set has exactly one bin (containing small
items) by Invariant \ref{Invariant:BSPonlyonewithsmall}. If such a
bin exists, we denote it by $B_{1}$. If $\BSfullbins \cap \Xbins
\neq \emptyset$, the set $\BSfullbins$ includes a bin that
contains a big item $i'$ with $s_{\min}(\fbig(\BSfullbins)) =
s(i') < s(i^*)$ and we denote one such bin (with a big item of
minimum size in $\BSfullbins$) by $B_2$. As $\Xbins \neq
\emptyset$, at least one of the bins $B_{1}$ or $B_{2}$ must
exist, but it can also be the case that both exist. Let
$\tilde{\mathtt{B}}$ be the set of cardinality $1$ or $2$, which
contains these bins. The next operation of the second step is to
call $\greedypull(B^*,\tilde{\mathtt{B}})$. It is easy to see that
no matter whether $B^*$ is covered or not after this operation,
the invariants \ref{Invariant:BigItems} and
\ref{Invariant:BSPonlyonewithsmall} hold. Specifically, if $B_2$
does not exist, $B^*$ is not necessarily covered, but
\ref{Invariant:BigItems} and \ref{Invariant:BSPonlyonewithsmall}
hold as all big items of $\BSfullbins$ are not smaller than $i^*$
(as every such bin has at least one small item). If $B_2$ exists,
then $B^*$ keeps receiving items coming first from $B_1$ and then
possibly also from $B_2$, until it is covered. As the total size
of small items of $B_2$ is sufficient for covering $B^*$ since the
big item of $B_2$ is smaller than $i^*$, $B^*$ will be covered, so
all big items of $\BSwaitbins$ are not larger than $i^*$.

Lastly, we describe the third step, which is performed for all
cases above, after $i^*$ has been inserted. The insertion of $i^*$
might have violated \ref{Invariant:BB=BS}, that is, we now have
$|\BSbins| = |\BBbins| + 2$. In this case, we perform the last
step, namely, we select two bins $B_3,B_4\in\BSbins$ with minimal
big items, merge the big items into a $\BBbins$ bin and remove and
reinsert all small items from $B_{3}$ and $B_{4}$, using insertion
of small items. This yields, $|\BSbins| = |\BBbins| - 1$ and
\ref{Invariant:BB=BS} holds.

\item{Insertion into ${\BBbins}$.}

If $i^*$ was not inserted in any of the last steps, it is inserted
into $\BBbins$. In this case, we know from the conditions above
that $ s(i^*) \leq s_{\min}(\fbig(\BSbins\cup\BMbins))$, and
additionally that $ s(i^*) \geq s_{\max}(\fbig(\BBbins))$ implies
$|\BSbins| = |\BBbins| + 1$. We consider two cases.

If $|\BSbins| = |\BBbins| + 1$ (and hence $\BSbins\neq\emptyset$), we select a bin $B\in\BSbins$ with a big item of minimal size.
We insert $i^*$ into $B$ to obtain a $\BBbins$ bin and remove and reinsert all small items from $B$.
This yields $|\BSbins| = |\BBbins| - 1$ and \ref{Invariant:BB=BS} holds.

If $|\BSbins| < |\BBbins| + 1$, we have $s(i^*) < s_{\max}(\fbig(\BBbins))$ (and hence $\BBbins\neq\emptyset$).
In this case, we select a bin $B\in\BBbins$ with a big item $i$ of maximal size, insert $i^*$ into $B$, and remove and reinsert $i$.
Because of its size and invariant \ref{Invariant:MediumBins}, the item $i$ will be inserted into $\BSbins$.
Note that in both cases invariant \ref{Invariant:BigItems} is maintained.

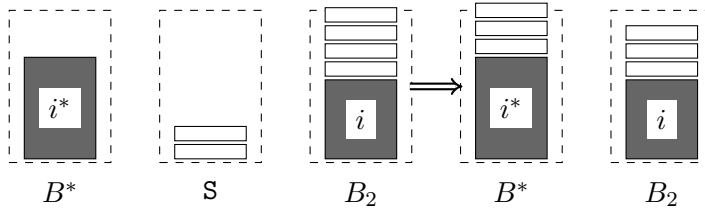
\begin{figure}[h]
  \centering
\begin{enumerate}[label=(\alph*)]
\item Insert into $\BMbins$: open a new bin $B^{*}$ for $i^{*}$;
pull $\Mbins$ into $B^{*}$; pull from bin $B\in
  \BMbins$ containing the smallest big item $i$; remove and reinsert
  $i$.

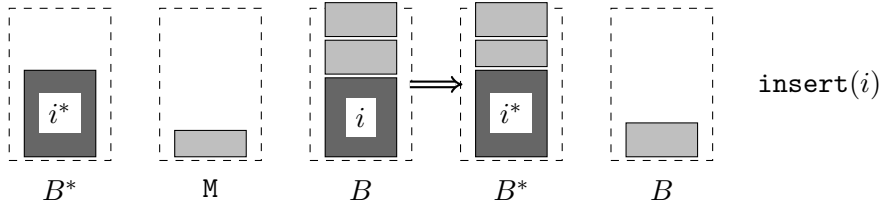
\begin{tikzpicture}
      \drawBin{0}{}{b0}
      \addToBin{0}{0}{.6}{\tikz\node[fill=white]{$i^{*}$};}{fill=big}{}
      \node[below = of b0]{$B^{*}$};

      \drawBin{2}{}{b1}
      \foreach \i in {0}{
        \addToBin{2}{\i*.2}{0.2}{}{fill=medium}{}
      }
      \node[below = of b1]{$\Mbins$};

      \drawBin{4}{}{b2}
      \addToBin{4}{0}{.55}{\tikz\node[fill=white]{$i$};}{fill=big}{}
      \foreach \i in {0,1}{
        \addToBin{4}{.55+\i*.25}{0.25}{}{fill=medium}{}
      }
      \node[below = of b2]{$B$};

      \drawBin{6}{}{b3}
      \addToBin{6}{0}{.6}{\tikz\node[fill=white]{$i^{*}$};}{fill=big}{}
      \addToBin{6}{.6}{0.2}{}{fill=medium}{}
      \addToBin{6}{.8}{0.25}{}{fill=medium}{}
      \node[below = of b3]{$B^{*}$};

      \drawBin{8}{}{b4}
      \addToBin{8}{0}{0.25}{}{fill=medium}{}
      \node[below = of b4]{$B$};

      \node[right = of b4]{\texttt{insert}($i$)};

      \draw[-implies, double equal sign distance, thick] \rightBin{b2} -- \leftBin{b3};
    \end{tikzpicture}

  \item Insert into $\BSbins$: open a new bin $B^{*}$ for $i^{*}$; pull $\Sbins$ into $B^{*}$; pull from bin $B\in
  \BSbins$ containing the smallest big item $i$.

\begin{tikzpicture}
      \drawBin{0}{}{b0}
      \addToBin{0}{0}{.7}{\tikz\node[fill=white]{$i^{*}$};}{fill=big}{}
      \node[below = of b0]{$B^{*}$};

      \drawBin{2}{}{b1}
      \foreach \i in {0,1}{
        \addToBin{2}{\i*.12}{.12}{}{fill=small}{}
      }
      \node[below = of b1]{$\Sbins$};

      \drawBin{4}{}{b2}
      \addToBin{4}{0}{.55}{\tikz\node[fill=white]{$i$};}{fill=big}{}
      \foreach \i in {.55,.67,.79,.91}{
        \addToBin{4}{\i}{.12}{}{fill=small}{}
      }
      \node[below = of b2]{$B_{2}$};

      \drawBin{6}{}{b3}
      \addToBin{6}{0}{.7}{\tikz\node[fill=white]{$i^{*}$};}{fill=big}{}
      \foreach \i in {0,1,2}{
        \addToBin{6}{.7+\i*.12}{.12}{}{fill=small}{}
      }

      \node[below = of b3]{$B^{*}$};

      \drawBin{8}{}{b4}
      \addToBin{8}{0}{.55}{\tikz\node[fill=white]{$i$};}{fill=big}{}
      \foreach \i in {.55,.67,.79}{
        \addToBin{8}{\i}{.12}{}{fill=small}{}
      }
      \node[below = of b4]{$B_{2}$};

      \draw[-implies, double equal sign distance, thick] \rightBin{b2} -- \leftBin{b3};
    \end{tikzpicture}

\end{enumerate}

  \caption{Insertion of a big item $i^*$. Big items are drawn in dark gray, medium items in light gray, and
    small items in white.}
  \label{fig:static_insert_big}
\end{figure}
 
\end{itemize}

\begin{lemma}\label{lem:static_migration_big}
The overall size of items migrated due to the insertion of a big item $i^*$ is upper bounded by $11$.
\end{lemma}
\begin{proof}
  First, note that an insertion into $\BSbins$ can not trigger the reinsertion
  of a big item. The insertion into $\BBbins$ can only trigger the reinsertion
  of a single big item into $\BSbins$ and the insertion into $\BMbins$ can only
  trigger the reinsertion into $\BBbins$ or $\BSbins$. Hence, each insertion of
  a big item can trigger at most two other insertions of big items in total and thus only
  move a total size of $2$ this way.

  Concerning the direct reassignments:
  \begin{itemize}
  \item If $i^*$ is inserted into $\BMbins$, this may directly cause
    reassignments of items with size at most $3$: there may be two calls of
    $\greedypull$ in each of which items with overall size at most $1$ may be
    moved and additionally a big item with size at most $1$ may be moved.
  \item If $i^*$ is inserted into $\BSbins$, the corresponding bound is $5$:
    there may be two calls of $\greedypull$ in both of which items with overall
    size at most $1$ may be moved and additionally a big item with size at most
    $1$ and small items with overall size at most $2$ may be moved.
\item If $i^*$ is inserted into $\BBbins$, the corresponding bound is $1$:
  either small items with overall size at most $1$ or a big item are moved.
\end{itemize}
Hence, the total size of items migrated is at most $2+3+5+1=11$.

\end{proof}

\paragraph{Insertion of Medium Items.}

If a medium item $i^*$ arrives, $\greedypush(i^*,\Mbins)$ is
called. Afterwards, the invariant \ref{Invariant:MediumBins} may
be infringed and if this happens, we have $\BSbins\cup\BBbins \neq
\emptyset$ and $s(\Mbins) \geq 1 - s_{\max}(\BSbins\cup\BBbins)$,
and we continue as follows. We will now describe how to pack a
barely covered bin using the items from $\Mbins$ and a largest big
item from $\BSbins\cup\BBbins$ to maintain
\ref{Invariant:MediumBins} again.

If $\BSbins=\emptyset$, \ref{Invariant:BB=BS} implies that
$\BBbins$ contains a single bin $B$ including two items $i$ and
$i'$ with $s_{\max}(\BSbins\cup \BBbins)=s(i)\geq s(i')$. We
remove~$i'$ from $B$, and call $\greedypull(B,\Mbins)$ to create a
$\BMbins$ bin. Afterwards, $s(\Mbins)$ and
$s_{\max}(\BSbins\cup\BBbins)$ are at most as big as they were
before $i^*$ arrived as the first item we pulled from $\Mbins$ is
at least as big as $i^*$, and therefore \ref{Invariant:MediumBins}
holds. Furthermore $\BSbins = \BBbins = \emptyset$ and
\ref{Invariant:BB=BS} still holds. Lastly, we reinsert the big
item $i'$.

If, on the other hand, $\BSbins\neq\emptyset$, the corresponding
big item $i$ with $s_{\max}(\BSbins\cup \BBbins)=s(i)$ is
contained in a bin $B\in\BSbins$, because of
\ref{Invariant:BigItems}. In this case, we remove the small items
from $B$ and call $\greedypull(B,\Mbins)$. Afterwards
\ref{Invariant:MediumBins} holds, but \ref{Invariant:BB=BS} may be
infringed due to the removal of a bin from $\BSbins$, \ie
$|\BBbins| = |\BSbins| + 2$. In this case, we remove the two
biggest items $i_1$ and $i_2$ from the bins $B_1,B_2\in\BBbins$
and if $B_1\neq B_2$ merge the two bins. This yields $|\BBbins| =
|\BSbins| + 1$ and \ref{Invariant:BB=BS} holds. Afterwards, we
reinsert the two items $i_1$ and $i_2$, which both will be
inserted in $\BSbins$ due to their sizes. No matter whether we had
to rebalance $|\BBbins|$ and $|\BSbins|$ or not, we reinsert the
removed small items from $B$ as a last step.

Figure~\ref{fig:static_insert_medium} contains an illustration of this process.

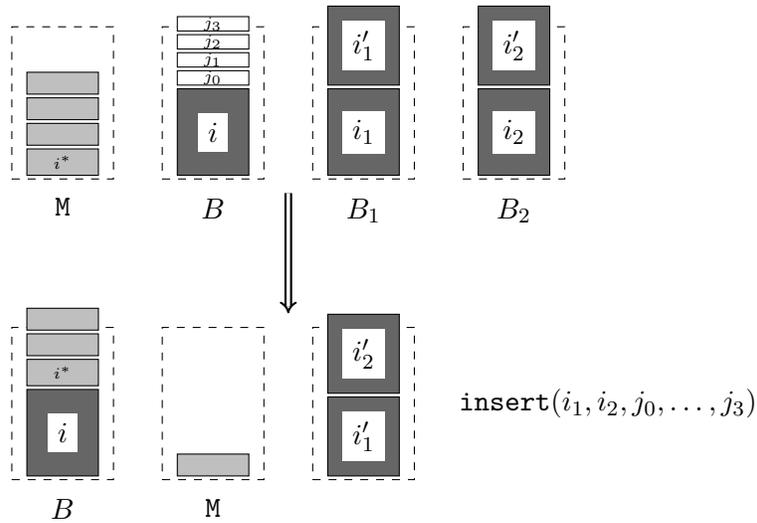
\begin{figure}[h]
  \centering

  \begin{tikzpicture}
    \drawBin{0}{}{b0}
    \addToBin{0}{0}{.2}{\tiny $i^{*}$}{fill=medium}{}
    \addToBin{0}{.2}{.17}{}{fill=medium}{}
    \addToBin{0}{.37}{.17}{}{fill=medium}{}
    \addToBin{0}{.54}{.17}{}{fill=medium}{}
    \node[below = of b0]{$\Mbins$};

    \drawBin{2}{}{b1}
    \addToBin{2}{0}{.6}{\tikz\node[fill=white]{$i$};}{fill=big}{}
      \foreach \i in {0,1,2,3}{
        \addToBin{2}{.6+\i*.12}{0.12}{\tiny $j_{\i}$}{fill=small}{}
      }
      \node[below = of b1]{$B$};

      \drawBin{4}{}{b2}
      \addToBin{4}{0}{.6}{\tikz\node[fill=white]{$i_{1}$};}{fill=big}{}
      \addToBin{4}{.6}{.55}{\tikz\node[fill=white]{$i'_{1}$};}{fill=big}{}
      \node[below = of b2]{$B_{1}$};

      \drawBin{6}{}{b3}
      \addToBin{6}{0}{.6}{\tikz\node[fill=white]{$i_{2}$};}{fill=big}{}
      \addToBin{6}{.6}{.55}{\tikz\node[fill=white]{$i'_{2}$};}{fill=big}{}
      \node[below = of b3]{$B_{2}$};
      \begin{scope}[yshift=-4cm]
      \drawBin{0}{}{b1p}
      \addToBin{0}{0}{.6}{\tikz\node[fill=white]{$i$};}{fill=big}{}
      \addToBin{0}{.6}{.2}{\tiny $i^{*}$}{fill=medium}{}
      \foreach \i in {0,1}{
        \addToBin{0}{.8+\i*.17}{0.17}{}{fill=medium}{}
      }
      \node[below = of b1p]{$B$};

      \drawBin{2}{}{b2p}
      \addToBin{2}{0}{0.17}{}{fill=medium}{}
      \node[below = of b2p]{$\Mbins$};

      \drawBin{4}{}{b3p}
      \addToBin{4}{0}{.55}{\tikz\node[fill=white]{$i'_{1}$};}{fill=big}{}
      \addToBin{4}{.55}{.55}{\tikz\node[fill=white]{$i'_{2}$};}{fill=big}{}

      \node[right = of b3p]{\texttt{insert}$(i_{1},i_{2},j_{0},\ldots,j_{3})$};
      \draw[-implies, double equal sign distance, thick, shorten >=1.2cm,shorten <=1.2cm] ($ (b1) !0.5! (b2)$) --  ($ (b2p) !0.5! (b3p)$);
      \end{scope}

  \end{tikzpicture}
  \caption{Insertion of a medium item $i^*$: push $i^{*}$ into $\Mbins$; remove
    small items $j_{0},\ldots,j_{3}$ from bin $B\in \BSbins$ containing the largest big i; pull from
    $\Mbins$ into $B$; merge $\BBbins$ bins $B_{1}$ and $B_{2}$ containing the
    largest big items $i_{1}, i_{2}$; reinsert remaining items $i_{1},i_{2},j_{0},\ldots,j_{3}$.\\
    Big items are drawn in dark gray, medium items in light gray, and
    small items in white.}
  \label{fig:static_insert_medium}
\end{figure} 
\begin{lemma}\label{lem:static_migration_medium}
The overall size of items migrated due to the insertion of a medium item $i^*$ is upper bounded by $27$.
\end{lemma}
\begin{proof}
In the first case, one big item is reinserted, resulting in the migration of items with size at most $1+11=12$ (see Lemma \ref{lem:static_migration_big}).
Furthermore, there is one call of $\greedypull$ moving items with overall size at most $1$.
This yields a size of at most $13$.

In the second case, medium items with size at most $1$ are moved
using $\greedypull$ and small items with overall size at most $1$
are moved in the last step. Furthermore, three big items may be
moved and two of them may be reinserted, causing total migration
of size at most $1+1+3+2\cdot 11=27$. Hence, the overall size of
migrated items is upper bounded by~$27$.
\end{proof}

\paragraph{Analysis.}

The migration bound stated in Theorem~\ref{thm:wc_and_static_alg}
or more precisely $\frac{27}{\eps}$ is already implied by
Lemma~\ref{lem:static_migration_big} and
Lemma~\ref{lem:static_migration_medium}, as a medium item has size
above $\eps$. Furthermore, it is easy to see that:
\begin{remark}
  \label{rem:static:running_time}
The presented algorithm for static online bin covering has a polynomial running time.
\end{remark}
\noindent Hence, the only thing left to show is the stated asymptotic competitive ratio:
\begin{lemma}
  \label{lem:static:ratio}
The presented algorithm has an asymptotic competitive ratio of
$1.5 + \eps$ with an additive constant of $3$.
\end{lemma}
\begin{proof}
First, we consider the case $\BSwaitbins = \emptyset$ (see Figure
\ref{fig:examples_static}, second and third).
In this case, the claim holds because the bins on average have not too much excess size.
More precisely, we obviously have $\Opt(I) \leq s(I)$, and invariants \ref{Invariant:Binstructure} and \ref{Invariant:SmallandMedium} imply:
\[ s(I) < 2|\BBbins| + (1+\eps)|\BSbins| + 1.5 |\BMbins| + 1.5|\Mfullbins| + (1+\eps)|\Sfullbins| + 2\]
Furthermore, we have $0.5|\BBbins|\leq 0.5(|\BSbins| + 1)$, due to
Invariant \ref{Invariant:BB=BS}, and $|\BSbins|=|\BSfullbins|$, as
$\BSwaitbins=\emptyset$ holds in the case we are currently
considering. Hence:
\[\Opt(I) < (1.5 + \eps)(|\BBbins| + |\BSfullbins| + |\BMbins| + |\Mfullbins| + |\Sfullbins|) + 2.5 < (1.5 + \eps)\Alg(I) + 3.\]

A similar argument holds, if $\BSwaitbins \neq \emptyset$ but
$\BBbins = \emptyset$. In this case, we have $|\BSwaitbins| = 1$,
because of invariant \ref{Invariant:BB=BS} and thus
$\BSbins=\BSwaitbins$; and $\Sbins = \emptyset$, because of
Invariant \ref{Invariant:SmallBins}. Hence:
\[\Opt(I) \leq s(I) < 1.5 |\BMbins| + 1.5|\Mfullbins| + 2 = 1.5 \Alg(I) + 2.\]

Next, we consider the case $\BSwaitbins \neq \emptyset$ and
$\BBbins \neq \emptyset$ (see Figure \ref{fig:examples_static},
first example). In this case, we have $\Mfullbins = \emptyset$,
because of Invariant \ref{Invariant:MediumBins}, and $\Sbins =
\emptyset$, because of Invariant \ref{Invariant:SmallBins}. Note
that every bin of $\BSfullbins \cup \BMbins$ has at least one item
that is not big, since big items have sizes below $1$, and these
bins are covered.

Let $\xi = s_{\max}(\BSwaitbins)$ be the size of a big item from
$\BSwaitbins$ with maximal size. Then all items in
$\BBbins\cup\BSwaitbins$ are upper bounded in size by $\xi$
(\ref{Invariant:BigItems}) and $\xi > 0.5$. We construct a
modified instance $I^*$ as follows:
\begin{enumerate}
\item The size of each big item with size below $\xi$ is increased
to $\xi$. \item Every big item of size larger than $\xi$ is split
into a big item of size $\xi$ and a
  medium or small item, such that the total size of these two items is equal to the size of the original item.
  Let $X$ be the set of items with sizes of $\xi$, which we will call $\xi$-items in the instance after these transformations ($X$ includes also items whose sizes were $\xi$ in $I$).
\item For each bin from $\BSfullbins\cup\BMbins$, select the
largest item of $I$ that is not big and call it {\it special}. By
increasing item sizes if necessary, change the sizes of all
special items to $0.5$. Let $Y$ be the set of special items (whose
sizes are now all equal to $0.5$). Let $Z$ be the set of the
remaining items not belonging to $X$ or $Y$ (in the instance $I^*$
after the transformations, so there maybe be items that did not
exist in $I$ resulting from splitting a big item).
\end{enumerate}

The set of items in $I^*$ is just $X\cup Y\cup Z$. For the
instance $I^*$, any bin of $\BBbins$ contains two items of $X$
(and no other items). Any bin of $\BSwaitbins$ has an item of $X$,
and one of these bins may also have small items of $Z$, but it is
not covered. Any bin of $\BSfullbins\cup\BMbins$ has one item of
$X$, one item of $Y$, and possibly items of $Z$. There may be one
uncovered bin of $\Mbins$, containing items of $Z$.

Note that $\Opt(I)\leq\Opt(I^*)$, since any packing for $I$ can be
used as a packing for $I^*$ with at least the same number of
covered bins. Next, we investigate the relationship between
$\Opt(I^*)$ and the packing of the algorithm for the original
instance $I$. For some optimal solution for $I^*$ without
overpacked bins (more than barely covered), let $k_2$, $k_1$ and
$k_0$ be the sets of covered bins with $2$, $1$ and $0$ items from
$X\cup Y$, respectively. Then we have $\Opt(I^*) = k_2 + k_1 +
k_0$ and due to counting:
\[2k_2 + k_1 = |X\cup Y| = (2|\BBbins| + |\BSbins| + |\BMbins|) + (|\BMbins| + |\BSfullbins|).\]
Since each item in $X\cup Y$ is upper bounded by $\xi$, we have:
\[(1-\xi)k_1 + k_0 \leq s(Z) . \]

The total size of items (of $Z$ only) packed into the bin of
$\Mbins$ is below $1-\xi$ since $\BSwaitbins$ has a big item of
size $\xi$ in $I$, by Invariant \ref{Invariant:MediumBins}, since
every item of $\BSbins \cup \BBbins$ is smaller than $1-s(M)$. For
$\BSwaitbins$ only one bin may contain items of $Z$ by Invariant
\ref{Invariant:BSPonlyonewithsmall}, and this bin has an item of
size $\xi$ in $I$ (and it is not covered), so it also has items of
$Z$ of total size below $1-\xi$. Consider a bin of
$\BSfullbins\cup\BMbins$. The total size of items excluding the
special item is the same for $I$ and $I^*$. Since such a bin is
barely covered and for $I$ it has items of one class except for
the big item (small or medium), removing the special item results
in a load below $1$. The total size of items of $I^*$ in such a
bin excluding the $\xi$-item and the special item is below
$1-\xi$.

Therefore, we find that \[s(Z) \leq (1-\xi)(|\BMbins| +
|\BSfullbins| + 2) . \] Hence:
\begin{align*}
2\Opt(I)   &\leq 2(k_2 + k_1 + k_0)\\
           &\leq (2k_2 + k_1) + (k_1 + (1-\xi)^{-1}k_0)\\
           &\leq (2|\BBbins| + |\BSbins| + |\BMbins| + |\BMbins| + |\BSfullbins|) + (|\BMbins| + |\BSfullbins| + 2)\\
           &\leq 3|\BBbins| + 3|\BMbins| + 2|\BSfullbins| + 3 \\
           &\leq 3\Alg(I) + 3.
\end{align*}
In the second to last step, we again used invariant \ref{Invariant:BB=BS}.
\end{proof}

\newpage

\section{Non-amortized Migration in the Dynamic Case}

In this section we extend the result of the static case and show:
\begin{theorem}\label{thm:wc_and_dynamic_alg}
For each $0 < \eps < 1$ with $1/\eps\in\ZZ$, there is an algorithm
$\Alg$ for dynamic online bin covering with polynomial running
time, an asymptotic competitive ratio of $1.5 + \eps$ with
additive constant~$\Oh(\log 1/\eps)$, and a non-amortized
migration factor of $\Oh(\frac{\log^{2}(1/\eps)} {\eps ^{5}})$.
\end{theorem}

\paragraph{Motivation and Discussion.}
It is not too hard to see that the insertion procedures for big
and medium items designed for the static case can be reversed in
order to deal with departures of such items. However, small items
never cause migration in the static case algorithm and borrowing
this approach in the dynamic case causes immediate problems: Let
$N$ be some positive integer. Consider the case that $N^2$ items
of size $1/N$ arrived and were placed into $N$ bins, covering each
of them perfectly. Next, one item from each bin leaves yielding a
solution without any covered bin while the optimum number of
covered bins is $N - 1$. Hence, a migration strategy for
arbitrarily small items is needed in order to design a competitive
algorithm. Now, coming up with a migration strategy to deal with
the present example is rather simple, since all the items are of
the same size, but in principle small items may differ in size by
arbitrary multiplicative factors. That is, it is possible that
when a new small item which is relatively small arrives, some of
the existing small items cannot be migrated at all. Still, the
case with only small items can be dealt with comparatively easily
by adapting a technique that was developed for dynamic online bin
packing with migration \cite{DBLP:conf/approx/BerndtJK15}. The
basic idea is to sort the items non-increasingly and maintain a
packing that corresponds to a partition of this sequence into
barely covered bins. If an item arrives, it is inserted into the
correct bin and excess items are pushed to the right, that is, to
the neighboring bin containing the next items in the ordering, and
this process is repeated until the packing is restored.
Correspondingly, if an item departs, items are pulled in from the
next bin to the right. In this process the arrival or departure of
a small item can only cause movements of items that are at most as
big as the original one. While this is a useful property, it
obviously does not suffice to bound migration: Too many bins have
to be repacked. In order to deal with this, the bins are
partitioned into \emph{chains} of appropriate constant length with
a \emph{buffer bin} at the end, which is used to interrupt the
migration process. This technique is sufficient for the set of
bins $\Sbins$ containing only small items.

The small items are also placed together with big items in the
bins of $\BSbins$ and in principle the same problem can occur
here: Few items with very small overall size can leave many bins
uncovered when leaving and hence we need a migration strategy and
additional structure for these bins as well. Unfortunately, a
straight-forward combination of the chain approach with bins
containing big items fails. The main reason for this is that in
order to adapt our analysis, we need to cover the bins in
$\BSbins$ containing larger big items with higher priority and
furthermore guarantee that there are no (or only few) bins
contained in $\Sbins$ if there are bins containing big items that
are not covered, \ie $\BSwaitbins\neq\emptyset$. It is not hard to
see that spreading one sequence of chains out over the bins of
$\BSbins$ and $\Sbins$ will not suffice.

To overcome these problems, we developed a new technique: We
partition the bins of $\BSbins$ into few, that is, $\Oh(\log
1/\eps)$ many, \emph{groups}. Each of the groups is in turn
partitioned into \emph{parallel chains} of length $\Oh(1/\eps)$.
The groups are defined such that they comply with a non-increasing
ordering of both the big and the small items: the first group
contains the largest big and small items, the next group the
remaining largest, and so on. Similarly, the big or small items
contained in a bin of a \emph{single} parallel chain are at most
as big as the ones in the predecessor bin of that chain. However,
no such structure is maintained in between the parallel chains of
the same group. Now, whenever a buffer bin of a parallel chain
becomes empty or overfilled, items are pushed or pulled directly
into or out of the next group. This may spark a recursive
reaction, but since there are only few groups overall the total
migration can be bounded. Furthermore, we are able to guarantee
(i) that there is at most one group $\mathtt{G}$ containing
uncovered bins, and that all groups before are covered. Finally,
if $\Sbins\neq\emptyset$, we can guarantee (ii) that all bins of
$\BSbins$ are barely covered. Properties (i) and (ii) are the
essential properties we need in order to adapt our approach to the
dynamic case.

The buffer bins of the parallel chains act as the interface between one group and another, and therefore have a special role in this construction.
We require that the buffer bins of the same group have some special properties, e.\,g., that they can exchange items among each other without migration cost.
To guarantee these properties, we simulate the buffer bins of chains belonging
to a fixed group using one \emph{group buffer} bin. Hence, all chains within the
same group share the same buffer bin.
Because of their special role in the procedure, these group buffer bins are considered as a separate class, denoted by $\Gbuffer$.
These bins exclusively contain small items but may be overpacked, that is, unlike any other bin in the construction they may be covered but not barely covered.
To avoid confusion, we slightly alter the definition of $\Sbins$, excluding the group buffer bins regardless of whether they are at most barely covered or not.
Furthermore, we require the bins of $\BSfullbins$ to comply with a more
restricted notion of being barely covered: bins of this type need to be  \emph{well-covered}. A bin $B$ is called well-covered, if it is barely covered and additionally has the following property:
If $s_{\max}(\fbig(B)) + s_{\max}(B\setminus\fbig(B))\geq 1$, then $|B|\leq 2$.
In particular, consider the situation that $B$ contains one big item and more
than one small item. As $B$ is barely covered, removing the largest small item
from $B$ makes it not covered. But $B$ is also well-covered, hence removing all
of the small items \emph{but} the largest one will also leave $B$ not covered.
We call a bin $B$ \emph{more than well-covered}, if one can remove a subset of
items in $B$ to make $B$ well-covered and \emph{at most well-covered}, if one can
add a (possibly empty) set of items to $B$ to make it well-covered.

\paragraph{The Algorithm.}
In the following, we will describe in detail the additional
structure needed to deal with the small items along with the
insertion procedures for small items. We first describe the chain
structure used for $\Sbins$ and then the group structure used for
$\BSbins$. Next, we discuss the insertion and deletion procedures
of medium and big items. The main challenge here is to properly
deal with the insertion and deletion of big items in $\BSbins$,
because these can interfere with the group structure. Lastly, we
will argue that the migration, competitive ratio and running time
of the overall algorithm is properly bounded. In order to do so we
will, in the course of this section, again introduce invariants
and argue that they are maintained by the algorithm. Indeed, some
of the invariants can be transferred directly or with only small
changes from the static case:
\begin{enumerate}[label = I\arabic*,series=invariants]

\item The solution has the proposed bin type structure, \ie $\bins
= \BBbins\ \dot{\cup}\ \BMbins\ \dot{\cup}\ \BSbins\ \dot{\cup}\
\Mbins\ \dot{\cup}\ \Sbins\ \dot{\cup}\ \Gbuffer$  and $\BSbins =
\BSfullbins\ \dot{\cup}\ \BSwaitbins$.
\label{DInvariant:Binstructure}

\item The sets $\BBbins$ and $\BSbins$ are balanced in size, \ie
$\abs[\big]{|\BBbins|-|\BSbins|} \leq 1 $.
\label{DInvariant:BB=BS}

\item The big items contained in $\BMbins$ are at least as big as
the ones in
  $\BSbins$ which in turn are at least as big as the ones in $\BBbins$, \ie
  $s(i) \geq s(i')$ for each $i\in\fbig(\BMbins)$ and
  $i'\in\fbig(\BSbins\cup\BBbins)$; or $i\in\fbig(\BMbins\cup\BSbins)$ and
  $i'\in\fbig(\BBbins)$. \label{DInvariant:BigItems}

\item The items in $\Mbins$ cannot be used to cover a bin together
with a big  item from $\BSbins$ or $\BBbins$, \ie
$\BSbins\cup\BBbins \neq \emptyset
  \implies s(\Mbins) < 1 - s_{\max}(\BSbins\cup\BBbins)$. \label{DInvariant:MediumBins}

\item If there is a bin which only contains small items and is not a group
  buffer bin, then all bins in $\BSbins$ are covered, \ie $|\Sbins| > 0 \implies
  |\BSwaitbins| = 0$. \label{DInvariant:SmallBins}

\item The set $\Mbins$ contains at most one bin that is not covered. \label{DInvariant:MediumoneUncovered}

\end{enumerate}
In the following, we slightly change our notions of bins, that is, we allow empty bins.
Until now, we assumed that a bin ceases to exist as soon as it becomes empty.
In the following, however, there are some situations in which empty bins are maintained by the algorithm.
This will be made clear in each case.  Furthermore, without loss of generality we assume that $1/\eps$ is an integer.

\subsection{Chains}

\paragraph{Preliminaries.}

We formally define a \emph{chain} as a totally ordered finite set
of bins, where each bin contains at most one big item and no
medium items, such that the following conditions are fulfilled.
The ordering of the chain $\mathtt{C}$ is consistent with a
non-increasing ordering of the small items contained in
$\mathtt{C}$, that is, for bins $B,B'\in\mathtt{C}$ and small
items $i\in B$ and $i'\in B'$, we have $s(i)\geq s(i')$, if $B$
precedes $B'$ in $\mathtt{C}$. Furthermore, the corresponding
condition for big items holds as well, and if there is a bin $B$
in $\mathtt{C}$ containing no big item, then the successors of $B$
contain no big items as well.

Let $\mathtt{C}$ be a chain, $B\in\mathtt{C}$, and $X$ a set of
small items. We say that $X$ is \emph{eligible} for $\mathtt{C}$
and $B$, if $\mathtt{C}$ remains a chain after $X$ is inserted
into $B$. Furthermore, we denote the $j$-th bin in the ordering of
$\mathtt{C}$ as $B(\mathtt{C},j)$. We define two basic operations
on chains, which are recursive procedures. Intuitively,
$\chainpush(X,B)$ puts a set of small items $X$ into bin $B$ and
pushes the now superfluous small items to the next bin. In
contrast, $\chainpull(B)$ pulls small items from these next bins
to cover $B$. More formally, these operations behave as follows:
\begin{itemize}

\item $\chainpush(X,B,\mathtt{C})$ is given a chain $\mathtt{C}$, a bin $B\in\mathtt{C}$, and a set $X$ of small items eligible for $\mathtt{C}$ and $B$.
The set of items $X$ is inserted into $B$.
If $B$ is now overpacked, that is, more than well-covered, we repeatedly remove a smallest small item from $B$ until $B$ is well-covered and call the set of removed items $X'$.
If $B$ is the last bin in the chain, the set $X'$ is returned.
Otherwise there is a successor $B'$ of $B$ and $\chainpush(X',B',\mathtt{C})$ is
called.

\item $\chainpull(B,\mathtt{C})$ is given a chain $\mathtt{C}$ and a bin $B\in\mathtt{C}$.
If $B$ is not covered and not the last bin in the chain, we repeatedly move the biggest small item from $B'$ to $B$, where $B'$ is the successor of $B$ in $\mathtt{C}$, until either $B$ becomes well-covered or $B'$ does not contain small items anymore.
Afterwards $\chainpull(B',\mathtt{C})$ is called.

\end{itemize}
Remember that a bin $B$ is called well-covered, if it is barely covered and additionally we have that $s_{\max}(\fbig(B)) + s_{\max}(B\setminus\fbig(B))\geq 1$ implies $|B|\leq 2$.
Note that in both procedures the inputs of the recursive calls are legal, e.g.,
when $\chainpush(X',B',\mathtt{C})$ or $\chainpull(B',\mathtt{C})$ are called, $\mathtt{C}$ is still a chain, $B'\in\mathtt{C}$, and $X'$ is eligible for $\mathtt{C}$ and $B'$.
This immediately yields:
\begin{remark}\label{rem:dyn_chain_remains_chain}
After a call of $\chainpush(X,B,\mathtt{C})$ or $\chainpull(B,\mathtt{C})$, the set $\mathtt{C}$ remains a chain.
\end{remark}
If $Y$ is the set of items moved away from a bin $B$ due to a call of $\chainpush$, we say that $Y$ was pushed out of $B$ and moreover, if $Y$ is pushed out of the last bin of the chain, we say that it was pushed out of the chain.
Correspondingly, if $Y$ was moved away from a bin $B$ due to a call of $\chainpull$, we say that $Y$ was pulled out of $B$.
Our next goal is to formulate and prove some easy technical results concerning
these operations. These results will be among the main tools for bounding the migration due to the small items.
We consider $\chainpush$ first and $\chainpull$ second.

\begin{lemma}\label{lem:chainpush_basic}
Let $\mathtt{C}$ be a chain, $B\in\mathtt{C}$ be at most well-covered, and $X$
be a
set of small items eligible for $\mathtt{C}$ and $B$ such that $X\cup\fbig(B)$ forms a bin that is at most well-covered. Furthermore, assume that $s_{\min}(X)\geq s_{\max}(\fsmall(B))$ or $|X| = 1$.
Finally, let $Y\neq \emptyset$ be the set of items that are pushed out of $B$, due to a call of $\chainpush(X,B,\mathtt{C})$.
Then $B$ is well-covered after the call and we have:
\begin{enumerate}
\item $Y\cup\fbig(B)$ forms a bin that is at most well-covered.
\item \label{assertion2Lemma9} $s_{\max}(X)\geq s_{\max}(Y)$.
\item \label{assertion3Lemma9} $s(Y)\leq s(X) + s_{\max}(X)$.
\item If $B$ has a direct successor $B'$ in $\mathtt{C}$, then $Y$
is eligible for $\mathtt{C}$ and $B'$ and $s_{\min}(Y)\geq
s_{\max}(\fsmall(B'))$ (before the recursive call of
$\chainpush(Y,B',\mathtt{C})$).

\end{enumerate}
\end{lemma}
\begin{proof}

Note that $B$ is well-covered, as $Y\neq \emptyset$. We now
distinguish two cases, namely $s_{\min}(X)\geq
s_{\max}(\fsmall(B))$ and $s_{\min}(X)< s_{\max}(\fsmall(B))$.
In the first case, note that no item belonging to $X$ will be
pushed out of $B$, or, more precisely: if $Y'$ is the subset of
$X$ that is pushed out, there is a set $Z$ of items remaining in
$B$ that does not belong to $X$ and that is equivalent to $Y'$ in
the sense that there exists a bijection between the two, mapping
items of equal size onto each other. Therefore, $Y\cup\fbig(B)$ is
at most well-covered, because $B$ was at most well-covered. In the
second case, we have $s_{\min}(X)< s_{\max}(\fsmall(B))$ and hence
$|X| = 1$ by assumption. Let $X = \set{i^*}$. Note that the items
that are pushed out of $B$ each have size at most $s(i^*)$, and
there can be at most one item with size $s(i^*)$ in $Y$ because
$B$ was at most well-covered before $i^*$ was inserted. Hence
$Y\cup\fbig(B)$ is at most well-covered, because replacing the
item with size $s(i^*)$ from $Y$~---~if it exists~---~with the
biggest item from $\fsmall(B)$, yields a set of items completely
contained in $B$ before the call to $\chainpush$. The above
considerations also show $s_{\max}(X)\geq s_{\max}(Y)$.

Since $B$ was at most well-covered, removing the smallest items
from $B$ with overall size $s(X)$ after $X$ was inserted will
suffice to make it at most well-covered again. The exact value
$s(X)$ may not be met and therefore one extra item may be needed
whose size is bounded by $s_{\max}(X)$, due to the already
established property that  $s_{\max}(X)\geq s_{\max}(Y)$. Note
that these considerations also hold, if $B$ was well-covered with
only one small item, or $|X| = 1$ and $X\cup\fbig(B)$ is
well-covered. This implies the third point and the fourth holds,
because the smallest items are pushed out of $B$ and $\mathtt{C}$
is a chain.
\end{proof}

\begin{lemma}\label{lem:chainpull_basic}
Let $\mathtt{C}$ be a chain, $B\in\mathtt{C}$ a bin that is not covered and not the last bin in $\mathtt{C}$, and $B'$ its successor.
Furthermore, let $Y$ be the set of items that are pulled out of $B'$ due to a call of $\chainpull(B,\mathtt{C})$.
We have:
\begin{enumerate}
\item $Y\subseteq\fsmall(B)$.
\item $B$ is well-covered after the call, if $B'$ was well-covered.
\item\label{assertion3lemma10} $s(Y) \leq (1 - s(B)) + s_{\max}(\fsmall(B'))$.
\end{enumerate}
\end{lemma}
\begin{proof}
The first assertion follows directly from the definition of
$\chainpull$. Note that $\fsmall(B')\cup\fbig(B)$ forms a covered
bin due to the definition of chains, if $B'$ was well-covered.
This together with the definition of $\chainpull$ implies the
second assertion and this also holds if $|\fsmall(B')|=1$.
Furthermore, it suffices to pull items with size $(1 - s(B))$ out
of $B'$ in order to cover $B$. However, the exact value $(1 -
s(B))$ may not be met and therefore one extra item may be needed
whose size is bounded by $s_{\max}(\fsmall(B'))$ and this proves
the last claim.
\end{proof}

\paragraph{Sequential Chains.}

We maintain a partition of $\Sbins$ into sequential chains $\schain_1$, \dots,
$\schain_C$ with $|\schain_c| \in [(1/\eps + 1), (2/\eps + 1)]$ for $c<C$ and
$|\schain_C|\leq 2/\eps + 2$. Herein, the exact value of $C$ depends on the
instance. The last bin of each chain is its buffer bin. See Figure~\ref{fig:sequential_chains} for an
illustration.
\begin{figure}[h]
  \centering
  \begin{tikzpicture}
    \begin{scope}[yscale=.3, xscale=.5]
\foreach \j in {0, 5, 10}{
    \foreach \i in {0,1,2}{
      \pgfmathtruncatemacro{\ipj}{\i+\j}
      \drawBin{2*\ipj}{}{b\ipj}
    }
    \pgfmathtruncatemacro{\jfive}{3+\j}
    \drawBin{2*\jfive}{pattern=north west lines}{b\jfive}

    \foreach \i in {0,1,2}{
      \pgfmathtruncatemacro{\ipj}{\i+\j}
      \pgfmathtruncatemacro{\ipjpo}{\i+\j+1}
      \draw ($ (b\ipj) +(.65cm,0) $) -- ($ (b\ipjpo) - (.65cm,0) $); 
    }
  }
  \draw ($ (b3) +(.65cm,0) $) -- ($ (b5) - (.65cm,0) $);
  \draw ($ (b8) +(.65cm,0) $) -- ($ (b10) - (.65cm,0) $);
      \end{scope}
  \end{tikzpicture}
  
  \caption{Sketch of a sequential chain decomposition. Buffer bins are hatched.}
  \label{fig:sequential_chains}
\end{figure}
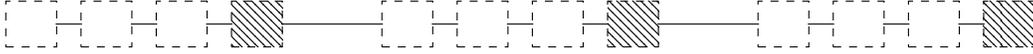 

We require that all bins that are not buffer bins are well-covered.
Buffer bins, on the other hand, are non-empty and at most well-covered.
The set of buffer bins is denoted by $\Sbuffer$.
These properties are summarized in the following invariant:
\begin{enumerate}[resume*=invariants]
\item $\Sbins$ has the described chain structure, in particular, $\Sbins\setminus\Sfullbins \subseteq \Sbuffer$ and $|\Sbuffer| \leq \eps |\Sfullbins| + 1$.\label{DInvariants:Chainstructure}
\end{enumerate}
The distribution of small items in $\Sbins$ is already limited by our definition of chains.
However, we will require an even stronger property such that all of the
orderings of the small items are consistent with each other:
\begin{enumerate}[resume*=invariants]
\item The non-increasing ordering of the small items in $\Sbins$ by size is consistent with both the ordering of the chains and their internal ordering.\label{DInvariants:Distribution_Small_sequential_chains}
\end{enumerate}
More precisely, due to this invariant for each item $i\in B(\schain_c,j)$ and $i'\in B(\schain_{c'},j')$, we have $s(i)\geq s(i')$, if $(c,j)$ is lexicographically smaller than $(c',j')$.
We next describe the insertion and deletion of a small item $i^*$ into $\Sbins$ maintaining these invariants.

\paragraph{Insertion.}

We describe the insertion of a set of small items $X$ into $\Sbins$ in two insertion scenarios:
In scenario 1, a new item $i^*$ was inserted into the instance, that is, $X=\set{i^*}$, and in
scenario 2, a set of items $X$ with $s(X)\leq 1$ and
$s_{\min}(X)\geq s_{\max}(\Sbins)$ was moved from $\BSbins$ to $\Sbins$.
In both scenarios, we treat the case $\Sbins = \emptyset$ as if there was exactly one chain with one bin, which is empty.
In scenario 1, we select the first bin $B$ (with respect to the ordering of the chains and their internal ordering) containing items that are smaller than~$i^*$, or the last bin if no such bin exists.
Let $\mathtt{C}$ be the chain $B$ is contained in.
In scenario 2, $\mathtt{C}$ is the first chain and $B$ its first bin.
In both cases, we call $\chainpush(X,B,\mathtt{C})$.
Due to the choice of $B$, Remark \ref{rem:dyn_chain_remains_chain} and a simple
inductive application of Lemma \ref{lem:chainpush_basic}, we can conclude  that the ordering of the small items is maintained (invariant \ref{DInvariants:Distribution_Small_sequential_chains}) and each non-buffer bin in $\mathtt{C}$ remains well-covered (\ref{DInvariants:Chainstructure}).
Let $Y$ be the set of items that are pushed out of $\mathtt{C}$.
If $Y \neq \emptyset$, we know that the buffer bin $B'$ of $\mathtt{C}$ is now well-covered, due to Lemma \ref{lem:chainpush_basic}, and furthermore $Y$ forms a bin that is at most well-covered, again due to Lemma \ref{lem:chainpush_basic}.
Hence, we turn the buffer bin into a regular bin and declare $Y$ the new buffer bin of $\mathtt{C}$.
If $\mathtt{C}$ now has length $2/\eps + 2$, we split it evenly into two chains each of which contains $\frac{1}{\eps}+1$ bins and declare the last bin of the first chain its buffer bin.
This yields two chains with length $1/\eps + 1$, thereby maintaining invariant
\ref{DInvariants:Chainstructure}. An example of this insertion is illustrated in
Figure~\ref{fig:dynamic_insert_small}.
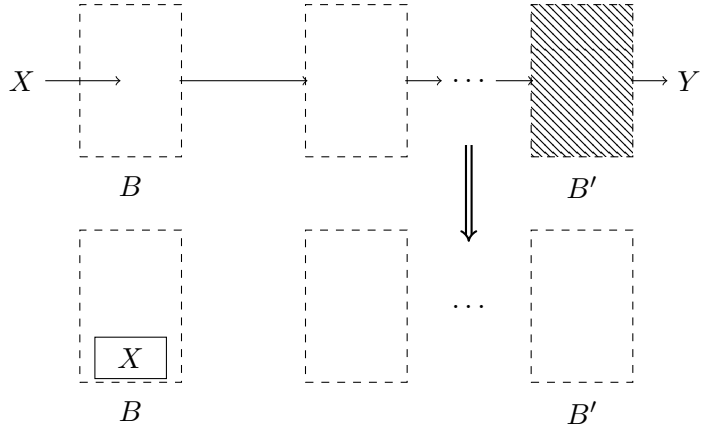
\begin{figure}[h]
  \centering

  \begin{tikzpicture}
    \drawBin{0}{}{b0}
    \node[below = of b0]{$B$};
    \node[left = of b0] (X) {$X$};
    \draw[->] (X) to (b0);

    \drawBin{3}{}{b1}

    \drawBin{6}{pattern=north west lines}{b2}
    \node[below = of b2]{$B'$};

    \node[right = of b1] (dots) {$\ldots$};
    \node[right = of b2] (y) {$Y$};
    \draw[->] \rightBin{b0} -- \leftBin{b1};
    \draw[->] \rightBin{b1} -- (dots);
    \draw[->] (dots) -- \leftBin{b2};
    \draw[->] \rightBin{b2} -- (y);

    \begin{scope}[yshift=-3cm]
      \drawBin{0}{}{b0}
      \node[below = of b0]{$B$};
      \addToBin{0}{0}{.3}{$X$}{fill=small}{}
      
    \drawBin{3}{}{b1}

    \drawBin{6}{}{b2}
    \node[below = of b2]{$B'$};

    \drawBin{9}{pattern=north west lines}{b3}
      \addToBin{9}{0}{.3}{$Y$}{fill=small}{}
    
    \node[right = of b1] (dotsp) {$\ldots$};

    \end{scope}
    \draw[-implies, double equal sign distance, thick, shorten >=.7cm,shorten <=.7cm] (dots) -- (dotsp);

  \end{tikzpicture}
  \caption{Insertion of a set of small items $X$. The current buffer bin is
    drawn hatched. The set $Y$ is pushed out of the chain and put into a new
    buffer bin.}
  \label{fig:dynamic_insert_small}
\end{figure}

\begin{remark}\label{rem:dyn_ins}
The overall size of migrated items due to the insertion of $X$ into $\Sbins$ is at most $\Oh(1/\eps^2)s(X)$, and therefore $\Oh(1/\eps^2)s(i^*)$ in scenario 1.
\end{remark}
\begin{proof}
Let $X_k$ be the $k$-th set pushed out of a bin due to the call of
$\chainpush$. A simple inductive application of assertions
(\ref{assertion2Lemma9}) and (\ref{assertion3Lemma9}) of Lemma
\ref{lem:chainpush_basic} yields that $s(X_k) \leq s(X) + k
s_{\max}(X)\leq (k+1)s(X)$. There are at most $2/\eps + 2$ bins in
the chain, and hence the overall size of migrated items is at most
$(2/\eps + 2)(2/\eps + 2)s(X) = \Oh(1/\eps^2)s(X)$.
\end{proof}

\paragraph{Deletion.}

Next, we consider the case that $i^*$ is deleted from a bin $B$
belonging to chain~$\mathtt{C}$ (that is a part of $S$). If $B$ is
still well-covered or a buffer bin, we do nothing. Hence, we
design a procedure to deal with the case that a non-buffer bin $B$
has become not covered. Similar to the insertion case, this
procedure will work regardless of whether $B$ is not covered
because $i^*$ was deleted, or because some set of small items has
been removed from $B$ to be placed in another group of bins. We
call $\chainpull(B,\mathtt{C})$. Since the invariants
\ref{DInvariants:Chainstructure} and
\ref{DInvariants:Distribution_Small_sequential_chains} did hold
before and due to  Remark \ref{rem:dyn_chain_remains_chain} and a
simple inductive application of Lemma \ref{lem:chainpull_basic},
we know that the ordering of the small items is maintained
(invariant \ref{DInvariants:Distribution_Small_sequential_chains})
and each non-buffer bin in $\mathtt{C}$ except for the bin $B^*$
directly preceding the buffer remains well-covered. If $B^*$ is
well-covered as well, the chain structure (invariant
\ref{DInvariants:Chainstructure}) is maintained, and we do nothing
else. Otherwise, the old buffer bin of $\mathtt{C}$ is deleted and
$B^*$ becomes the new buffer bin, which is not a problem, if
$|\mathtt{C}|\geq 1/\eps + 1$ or $\mathtt{C}$ is the last chain.
Hence, we consider the case that $|\mathtt{C}| = 1/\eps$ and there
is a direct successor chain $\mathtt{C}'$ of $\mathtt{C}$. If
$|\mathtt{C}'| > 1/\eps + 1$ the first bin of $\mathtt{C}'$ is
removed from this chain and declared the new buffer bin of
$\mathtt{C}$, and otherwise $\mathtt{C}$ and $\mathtt{C}'$ are
concatenated and we call the resulting chain again $\mathtt{C}$.
In both cases, we call $\chainpull(B^*,\mathtt{C})$, and the
invariants \ref{DInvariants:Distribution_Small_sequential_chains}
and \ref{DInvariants:Chainstructure} hold afterwards (and this
call to $\chainpull$ will not result in the need to concatenate
two chains).
\begin{remark}\label{rem:dyn_del}
The overall size of migrated items is at most $\Oh(1/\eps)(1-s(B)) + \Oh(1/\eps^2)\psi$, where $\psi$ is the maximum size of a small item in any bin succeeding $B$.
In particular, if $i^*$ was deleted from $B$ this amounts to $\Oh(1/\eps^2)s(i^*)$.
\end{remark}
\begin{proof}
Let $X_k$ be the $k$-th set pulled out of a bin due to one of the at most two calls of $\chainpull$.
Note that there can be at most $2/\eps + 1$ such bins.
A simple inductive application of assertion (\ref{assertion3lemma10}) of Lemma \ref{lem:chainpull_basic} yields that $s(X_k) \leq (1 - s(B)) + k s_{\max}(B')$.
In the case that an item $i^*$ was deleted from $B$, we additionally have $(1 - s(B))\leq s(i^*)$ and $s_{\max}(B')\leq s(i^*)$.
\end{proof}

\subsection{Groups}

We maintain a partition of $\BSbins$ into groups $\group_1$, \dots, $\group_G$ with $|\group_g| = 2^{g-1}\ceil{\eps |\BSbins|}$ for $g<G$ and $|\group_G| \leq 2^{G-1}\ceil{\eps |\BSbins|}$.
Note that $|\group_g| \leq \ceil{\eps |\BSbins|} + \sum_{g'= 1}^{g-1} |\group_{g'}| $ for each $g\in [G]$ and $G \leq \log (1/\eps) + 2 $.
These properties are summarized in the following invariant:
\begin{enumerate}[resume*=invariants]
\item $\BSbins$ has the described group structure. \label{DInvariants:Groupstructure}
\end{enumerate}
Furthermore, it will be useful in the following to consider the bins $\Sbins$ containing exclusively small items as the last group, that is, $\Sbins=\group_{G+1}$.
Each group $\group_g$, for $g\leq G$, is in turn partitioned into $C_g = \ceil{\eps|\group_g|}$ sets such that each of these sets has cardinality $1/\eps$ except for the last which may be smaller.
For each of these sets there is a buffer bin and the set of these buffer bins is called $\VGbuffer_g$.
As mentioned above, these buffer bins have a special role that will be described in more detail below.
The sets together with their buffer bins form \emph{parallel chains}, and since
all the bins of such a chain are contained in the corresponding group, we
sometimes say that the parallel chains are contained in or part of their
respective group. See Figure~\ref{fig:parallel_chains} for an illustration.

  \begin{figure}[h]
    \centering
  
  \begin{tikzpicture}
    \begin{scope}[yscale=.3, xscale=.5]
      \foreach \k in {0, -4cm, -8cm}{
        \begin{scope}[yshift=\k]
\foreach \j in {0}{
    \foreach \i in {0,1,...,4}{
      \pgfmathtruncatemacro{\ipj}{\i+\j}
      \drawBin{2*\ipj}{}{b\ipj\k}
    }
    \pgfmathtruncatemacro{\jfive}{5+\j}
    \drawBin{2*\jfive}{pattern=north west lines}{b\jfive\k}

    \foreach \i in {0,1,2,3,4}{
      \pgfmathtruncatemacro{\ipj}{\i+\j}
      \pgfmathtruncatemacro{\ipjpo}{\i+\j+1}
      \draw \rightBin{b\ipj\k} -- \leftBin{b\ipjpo\k};
    }
  }
          \end{scope}
        }
        \begin{scope}[yshift=-4cm]
        \drawBin{13}{pattern=north west lines}{bb}  
      \end{scope}
      \draw[-, double equal sign distance] \rightBin{b50} -- \leftBin{bb};
      \draw[-, double equal sign distance] \rightBin{b5-8cm} -- \leftBin{bb};
      \draw[-, double equal sign distance] \rightBin{b5-4cm} -- \leftBin{bb};
            
      \end{scope}
  \end{tikzpicture}
  \caption{Sketch of a parallel chain decomposition. Buffer bins are hatched. }
  \label{fig:parallel_chains}
\end{figure}
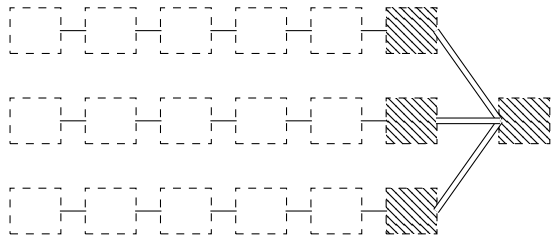 

Each bin in a parallel chain is at most well-covered.
The well-covered bins are all placed in a sequence at the start of the chain, and there can be at most one bin of the chain containing small items that is not covered, and this bin is positioned directly following the last well-covered bin of the chain.
Note that this implies $|\sett{B\in\bigcup_{g\in[G]}\VGbuffer_g}{\fsmall(B)\neq \emptyset}| \leq \eps|\BSfullbins| + G$.
Again, we summarize:
\begin{enumerate}[resume*=invariants]
\item Each group has the described internal parallel chain structure. \label{DInvariant:parallel_chain_structure}
\end{enumerate}
Note that due to the definition of chains, the above invariant also restricts the \emph{internal} distribution of big and small items in each group.
With the next two invariants, we control the distribution of big and small items \emph{between} the different groups.
\begin{enumerate}[resume*=invariants]

\item The ordering of the groups (including $\Sbins=\group_{G+1}$) is consistent with a non-increasing ordering of both the big and the small items. \label{DIvariants:distribution_big_small_groups}

\item There is a group-index $g^*\in [G+1]$ such that each bin in each group
  preceding $\group_{g^*}$ is well-covered and each succeeding group does not
  contain small items. \emph{Small items will always be used to create well-covered bins.
  }
\label{DInvariants:wellpacked_groups}

\end{enumerate}
More precisely, the former invariant together with
\ref{DInvariant:parallel_chain_structure} yields the following.
Let $i\in B(\mathtt{C},j)$ and $i'\in B(\mathtt{C}',j')$, where
$\mathtt{C}$ and $\mathtt{C}'$ are (parallel or sequential) chains
belonging to group $\group_{g}$ and $\group_{g'}$ respectively and
$i,i'$ belong to a common class (either both are big or both are
small). Then $g < g'$, implies $s(i)\geq s(i')$ and the same
holds, if $\mathtt{C}=\mathtt{C}'$ and $j<j'$. Furthermore, due to
the latter invariant, in some sense, there is always exactly one
group being filled with or emptied of small items (\ie the group
of index $g^*$ of invariant \ref{DInvariants:wellpacked_groups}).
Next, we discuss the buffer bins of parallel chains in more
detail.

\paragraph{Group Buffer Bins.}
If we can always maintain the group structure and the chain structures described
above, we are able to guarantee an asymptotic competitive ratio of
$3/2+\eps$. Unfortunately, maintaining this structure will be rather
complicated. Firstly, the buffer bins of parallel chains behave very different
from other bins. In order to have a notion of well-covered buffer bins, we
assume that the buffer bin of each parallel chain  contains a
\emph{virtual} big item that is of the same size as the big item in the second
to last bin of the chain. This virtual big item will only be used in the
theoretical analysis.

Furthermore, in a parallel chain setting, each chain by itself is
ordered with regards to the small and big items. But the
distribution of item sizes of chains within the same group might
be very complicated. For example, we might want to add a small
item into a chain $\mathtt{C}_{1}$ in group $\group_{g}$. This
will trigger a push of small items along this chain and might
remove some small items $Y$ from the buffer bin of
$\mathtt{C}_{1}$. The natural approach would be to try to insert
these items into the subsequent group $\group_{g+1}$, but there
might be another chain $\mathtt{C}_{2}$ in group $\group_{g}$ that
contains items smaller than those in $Y$. Hence, we first would
need to insert the items of $Y$ into $\mathtt{C}_{2}$, triggering
even more repacking. To prevent this cascade of repacking, we
observe that the smallest small items in $\group_{g}$ must be
contained in the buffer bins. Now, if we push small items along
$\mathtt{C}_{1}$, we will first rearrange the small items in the
buffer bins such that the buffer bin of $\mathtt{C}_{1}$ contains
only the smallest items. In this way, we can indeed push the
removed items $Y$ to $\group_{g+1}$. To simplify this
rearrangement of small items in buffer bins, we will make the
following two assumptions that may seem counter-intuitive at first
glance: Firstly, the buffer bins of the same group are allowed to
exchange items without migration costs; and secondly, items may be
split between buffer bins of the same group, with part of an item
placed in one buffer bin and the other part in another.

We justify these assumptions as well as the virtual big items, by simulating
these bins using one \emph{group buffer bin} for each group $\group_g$ (as
indicated in Figure~\ref{fig:parallel_chains}).
All the small items from the buffer bins of $\group_g$ are in fact placed in the group buffer bin of $g$ and the buffer bins of the chains are only virtually maintained by the algorithm.
Clearly, splitting items and moving them without cost does not pose a problem
anymore, as all of these items are placed on the same bin and do not leave it.
The set of group buffer bins is denoted by $\Gbuffer$.
Note that we do neither consider the (virtual) buffer bins of the parallel chains to be a part of $\BSbins$, nor the actual group buffer bins to be a part of $\Sbins$.

Let $B$ be the buffer bin of a parallel chain $\mathtt{C}$. If an
item is placed fractionally in $B$ this item is at most as big as
any other item placed in a preceding bin of the chain.
Furthermore, if $B$ contains small item pieces, we guarantee that
there is a small item $i$ fully contained in $B$ with the property
that $i$ is at least as big as any other small item for which item
pieces are placed in $B$. Due to this property, the definition of
well-covered can be extended to the buffer bins quite naturally:
the bin $B$ is called \emph{well-covered}, if (i) it is covered,
if (ii) the small item pieces together with the virtual big item
but excluding $i$, cannot cover a bin where $i$ is a largest small
item fully contained in $B$, and if (iii) the virtual big item and
$i$ together can cover a bin, then there is no other small item
piece placed in $B$.
\begin{enumerate}[resume*=invariants]
\item The buffer bins of each group and the corresponding group buffer bin have the described structure. \label{DInvariants:group_buffer_bin}
\end{enumerate}

If $\chainpush$ is called on the parallel chain $\mathtt{C}$, fractional item pieces may be pushed out of~$\mathtt{C}$.
When this case occurs, we will carefully redistribute the item pieces in the buffer bins of the group, or remove the full item from the group.
The operation $\chainpull$, on the other hand, could pull fractional items
further into the chain and we would like to avoid that also.
However, as described above, we guarantee that the biggest small item in the buffer bin $B$ is fully placed in it.
Hence, it should be pulled first.
Afterwards, if $B$ still contains small item pieces, we simply choose a small item $i'$ with maximal size that is at least fractionally placed on $B$, and pull all pieces of $i'$ onto $B$.
Since we only move item pieces from other buffer bins of the group, no additional migration costs occur.

\paragraph{Insertion.}

When a small item $i^{*}$ arrives, we have to be careful to
maintain the described structure when inserting it. In particular,
we have to maintain the parallel chain structure in each group
(invariant \ref{DInvariant:parallel_chain_structure}) and the
correct distribution of small items
(\ref{DIvariants:distribution_big_small_groups} and
\ref{DInvariants:wellpacked_groups}). For the group structure
(\ref{DInvariants:Groupstructure}) and the other invariants, it
will be easy to see that they are maintained.

The first step is to choose the appropriate group $\group_g$
depending on the size of $i^*$ and with respect to the correct
distribution of the small items
(\ref{DIvariants:distribution_big_small_groups} and
\ref{DInvariants:wellpacked_groups}). More precisely, $\group_g$
is either (i) the first group containing an item that is smaller
than $i^*$ if such a group exists; or (ii) the last group
containing small items, if that group is the last group or
contains a bin that is not covered; or (iii) the group directly
after the last group containing small items, otherwise.

Instead of describing the insertion of $i^{*}$ into $\group_g$, we take a
slightly more general approach and consider two scenarios, where a set $X$ of
small items has to be inserted into~$\group_g$, similar to the insertion into
$\Sbins$.
We do so, because the insertion procedure works recursively.
In scenario 1 we have $X=\set{i^*}$, and in scenario 2 all items in $X$ are at least as big as any small item in $\group_g$, and adding a big item from $\group_g$ to $X$ forms a bin that is at most well covered.
For the case $g=G+1$ we already discussed the same scenarios and refer to the section on
sequential chains.
Hence we assume $g<G+1$.

We distinguish two cases in the following depending on the existence of bins
without any small items.
If $\group_g$ contains a chain $\mathtt{C}$ that includes a bin without any
small item, the procedure is rather easy: we choose an appropriate bin $B$ with respect to the parallel chain structure (\ref{DInvariant:parallel_chain_structure} and \ref{DInvariants:wellpacked_groups}).
More precisely, in scenario 1, we choose the first bin containing an item that is smaller than $i^*$ if such a bin exists; and otherwise the first bin that is not covered and is either the first bin of the chain or has a well-covered predecessor.
In scenario 2, we simply pick the first bin of the chain.
We perform $\chainpush(X,B,\mathtt{C})$.
Remark \ref{rem:dyn_chain_remains_chain} and an easy inductive application of Lemma \ref{lem:chainpush_basic} yield that no items are pushed out of $\mathtt{C}$ and all invariants are maintained.

On the other hand, if each bin in $\group_g$ contains small items, only the
buffer bins may not be covered. This complicates the maintenance of our
invariants.

We pick a chain $\mathtt{C}$ with the property that the size of a largest small item contained in its buffer bin is minimal.
This chain has the property that each small item placed in its buffer bin could also be placed in any other buffer bin of the group without violation of the ordering of small items in the chain (\ref{DInvariant:parallel_chain_structure}).
Again, we choose an appropriate bin $B$ with respect to the parallel chain structure, that is, the first bin in scenario 2, or in scenario 1 the first bin containing an item that is smaller than $i^*$ if such a bin exists, or the last bin of the chain.
We call $\chainpush(X,B,\mathtt{C})$ and denote the set of items pushed out of $\mathtt{C}$ as $Y$.

If $Y=\emptyset$, the insertion procedure is now complete and all invariants hold (see above).
Hence, we consider the case $Y\neq\emptyset$. Note that $Y$ may contain the biggest small item previously placed in the buffer bin of $\mathtt{C}$.
However, in this case all items now contained in the buffer bin are not fractional and hence, the new biggest item is neither.
In any case, invariant \ref{DInvariants:group_buffer_bin} is maintained. Ideally, we would like to simply
insert $Y$ into the next group $\group_{g+1}$, but the structure of $Y$
might be unsuited to do so without violating our invariants. We will thus construct a more suitable set $Z$ from
$Y$ and $\VGbuffer_{g}$ and insert this set into $\group_{g+1}$. The remaining
items $Y'=Y\setminus Z$ will be placed on the buffer bins of $\group_{g}$.

More formally, let $Z\subseteq \fsmall(\VGbuffer_g \cup Y)$ be a
set of small items constructed greedily by choosing the smallest
available item and putting it into $Z$ until $s(Z)\geq s(Y)$. Ties
are broken in favor of items (fractionally) contained in $Y$.
Clearly, we have $s(Z) - s_{\max}(Z) < s(Y)$ and $s_{\max}(Z)\leq
s_{\min}(Y\setminus Z)$. We remove all items or item pieces
belonging to $Z$ from $Y$ and $\group_g$ and define $Y'=Y\setminus
Z$. Note, that in this process of creating $Z$, no buffer bin can
become empty, because  no item that is the biggest small item of a
buffer bin can be included in $Z$ due to the choice of chain
$\mathtt{C}$. Hence, invariant \ref{DInvariants:group_buffer_bin}
is maintained. Next, we place the items and item pieces from $Y'$
onto buffer bins of $\group_g$: If a piece of an item $i\in Z$ was
previously placed on a buffer bin $B$, we can imagine that
removing it left a \emph{gap}. As $s(Z)\geq s(Y)$, the space
$s(Z\cap \VGbuffer_{g})$ freed up by constructing $Z$ is at least
$s(Y')$. Hence, we can place the items from $Y'$ fractionally into
these gaps. Note that at this point all invariants hold and $Z$
contains items that are appropriate for the insertion in
$\group_{g+1}$ with respect to the distribution of small items in
the groups (invariant
\ref{DIvariants:distribution_big_small_groups}). As a last step,
we split $Z$ into two sets $Z_1$ and $Z_2$ such that their
insertion into $\group_{g+1}$ is consistent with scenario 2. The
sets are defined as follows. Let $q$ be the size of the biggest
item piece contained in $Y$. The first set $Z_1$ is chosen such
that $s(Z_1)\geq q$, $s(Z_1) - s_{\min}(Z_1) < q$, and
$s_{\min}(Z_1)\geq s_{\max}(Z\setminus Z_1)$ (with
$s_{\max}(\emptyset) = 0$). We set $Z_2 = Z\setminus Z_1$ and
first insert $Z_2$ and then~$Z_1$.
\begin{lemma}
Inserting $Z_1$ and $Z_2$ into $\group_{g+1}$ is consistent with scenario 2, and this includes the case $g=G$.
\end{lemma}
\begin{proof}
For the case $g=G$ note that $s(Z_1)\leq 2\eps\leq 1 $ and $s(Z_2)<s(Y)\leq 1$ which is sufficient for scenario 2 in this case.
We assume $g<G$ in the following and consider two cases.
Let $i$ be the biggest item (fractionally) contained in $Y$.

In the first case, we assume that $i$ is big enough to cover a bin together with the biggest (big) item $b$ contained in $\group_{g+1}$.
Note that due to invariant \ref{DIvariants:distribution_big_small_groups} the
big items contained in $\group_{g}$ are at least as big as $b$.
Since the buffer bin $B'$ of chain $\mathtt{C}$ was at most well-covered, this
implies that $Y = \set{i}$.   Remember that we
choose the chain $\mathtt{C}$ in a way such that the size of the largest small
item, call it $i'$, in its buffer bin $B'$ is minimal. As $Y=\{i\}$,  either $i$ is this largest small item $i'$ placed on $B'$ before the push, or scenario 1 applies and $i=i^{*}$ was just inserted
into the instance.
In both cases we have $s(i)\leq s(i')$.
Due to our choice of $\mathtt{C}$, any largest small item contained in a buffer bin of $\group_{g}$ is large enough to cover a bin together with any big item contained in $\group_{g}$.
Therefore, each buffer bin contains exactly one small item and that item is at least as big as $i$.
This in turn implies that $Z = Y = \set{i}$, $Z_1 = Z$ and $Z_2 = \emptyset$, yielding the assertion in this case.

In the second case, $i$ is not big enough to cover a bin together with $b$.
Hence, $Z_1$ fits the requirements.
Furthermore, note that $Z_2$ is at most as big as $Y$ minus the biggest item piece contained in that set.
This implies that $Z_2$ is suitably small as well:
either $Y$ was fully contained in the buffer bin before (yielding the assertion), or scenario 1 applies and $i^*\in Y$.
In the latter case, $i^*$ has to be the biggest item (and item piece) contained in $Y$ and an item that is bigger than $i^*$ remains in the buffer bin $B'$ of $\mathtt{C}$.
Hence, $Z_2$ is also suitably small.
\end{proof}

\begin{lemma}\label{lem:dyn_in_small_migration}
The overall size of items that are migrated due to the insertion
of a small item $i^*$ is $\Oh((1/\eps^3)\cdot
\log^2(1/\eps))s(i^*)$.
\end{lemma}
\begin{proof}
Let $X_0 =\set{i^*}$ and $g_0$ be the index of the group $i^*$ is inserted to.
For the case $g_0 = G+1$, we refer to Remark \ref{rem:dyn_ins} that already
guarantees the claimed bound.
Hence, we assume $g_0\leq G$.
Furthermore, let $g_k = g_0 + k$ for $k\in\set{1,\dots,G+1-g_0}$ and $X_k$ be a set of items that is (recursively) inserted into $\group_{g_k}$.
Note that at most $2^k$ sets are inserted into $\group_{g_k}$ for each
$k\in\set{0,\dots,G+1-g_0}$ (due to the insertion of both parts of $Z$) and $k\leq G + 1\leq \log(1/\eps) + 3$ (invariant \ref{DInvariants:Groupstructure}).

We first consider $\group_{g_0}$. When $X_0$ is inserted into this
group, $\chainpush$ is called exactly once (not counting recursive
calls). Let $X_{0,r}$ be the $r$-th set pushed out of a bin due to
this call. A simple inductive application of Lemma
\ref{lem:chainpush_basic} yields that $s(X_{0,r}) \leq s(X_0) + r
s_{\max}(X_0) = (r+1)s(i^*)$. The corresponding chain has length
at most $1/\eps + 1$. Hence, the overall size of items moved in
this group due to the insertion of $X_0$ is at most
$\Oh(1/\eps^2)s(i^*)$. Let $Z_1$ and $Z_2$ be the sets that are
inserted into the next group. Considering the choice of these sets
and the lengths of the parallel chains, we have $s(Z_1) + s(Z_2)
\leq (1/\eps + 3)s(i^*)$, and therefore this upper bound holds for
both sets independently. This proves the base case of the
inductive claim that  $s(X_k) \leq (1 + k(1/\eps + 2))s(i^*)$. The
claim for $k$ holds by noting that $s_{\max}(X_k) \leq s(i^*)$
(using \ref{DIvariants:distribution_big_small_groups}), and that
the chain to which $X_k$ is inserted has at most $\frac 1{\eps}+1$
bins (including the buffer bin).

Furthermore, the overall size of migrated items in group
$\group_{g_k}$ due to the insertion of $X_k$ is at most $(1/\eps +
1)(1 + k(1/\eps + 2) + 1/\eps + 1)s(i^*) = \Oh((1/\eps^2)\cdot
\log(1/\eps))s(i^*)$ and this bound also holds for $k = G+1-g_0$
(see Remark \ref{rem:dyn_ins}). Hence, the overall size of
migrated items is bounded by $(G+1)2^{G+1}\Oh((1/\eps^2)\cdot
\log(1/\eps))s(i^*) = \Oh((1/\eps^3)\cdot \log^2(1/\eps))s(i^*)$.
\end{proof}

\paragraph{Deletion.}

We now consider the case that a small item $i^{*}$ is deleted.
Like in the case of insertion, we have to be careful to maintain the parallel chain structure in each group (invariant \ref{DInvariant:parallel_chain_structure}) and the correct distribution of small items (\ref{DIvariants:distribution_big_small_groups} and \ref{DInvariants:wellpacked_groups}).
Let $i^*$ be removed from a bin $B^*$ positioned in a chain $\mathtt{C}^*$ in a group $\group_{g}$.
We already considered the case $g=G+1$ and hence assume $g\leq G$.
There are several cases in which we do nothing:
If (i) $B^*$ is still well-covered; or if (ii) $B^{*}$ is a buffer bin; or if
(iii) we have $g<G+1$ and $\group_{g}$ is the last group containing small items and $B^*$ was the last bin in $\mathtt{C}^*$ that contained small items.
In these cases, it is easy to see that the group structure and all invariants
are maintained. For the steps that we take otherwise, we again take a slightly more general approach, because they are recursive in nature, and because they will also be useful when considering big and medium items.

A \emph{gap bin} $B$ is a bin that is part of the group structure belonging to some
parallel or sequential chain and should be covered with respect to the
invariants \ref{DInvariant:parallel_chain_structure},
\ref{DInvariants:wellpacked_groups} and
\ref{DInvariants:Distribution_Small_sequential_chains}, but (due to repacking)
is currently not.
More precisely, $B$ is not a buffer bin, and there are small items positioned in
subsequent bins, which may either be part of  the same parallel or sequential
chain, or may be part of a subsequent group or sequential chain.
Furthermore, we call a gap bin \emph{maximal}, if there are no gap bins among the subsequent bins.
In the following, we design a procedure $\gapfill(B)$ used to repair such a maximal gap bin $B$.
It may be called recursively in situations in which there are more gap bins
and therefore the parallel chain structure (invariant  \ref{DInvariant:parallel_chain_structure}) and the correct distribution of small items (\ref{DInvariants:wellpacked_groups}) may not hold.
However, we can guarantee that for each initial (non-recursive) call of
$\gapfill(B)$, the bin $B$ will be
the only gap bin, and we will guarantee that there are no gap bins left after the call (including all recursive calls).
Now, it is easy to see that the bin $B^*$ is indeed a maximal gap bin and we
will handle the remaining cases by a call of $\gapfill(B^*)$.

Let $B$ be a maximal gap bin positioned in a chain $\mathtt{C}$ in a group $\group_{g}$.
First, we discuss the special case $g=G+1$.
In the section on chains, we already designed a procedure for the case that a non-buffer bin $B$ becomes not covered, that is, in the scenario discussed there, we have only a single (maximal) gap bin in the group.
In all application scenarios of $\gapfill$ in the last group, this will actually
be enough, but it is also very easy to see that the described procedure will
work for any maximal gap bin, the only difference being that invariant
\ref{DInvariants:Chainstructure} may not hold because of other gap bins that
might be present.

From now on, we assume $g<G+1$.
As a first step, we call $\chainpull(B,\mathtt{C})$.
Afterwards, $\mathtt{C}$ remains a chain, and $B$ as well as each bin which was followed by a well-covered bin before is well-covered (see Remark \ref{rem:dyn_chain_remains_chain} and Lemma \ref{lem:chainpull_basic}).
Hence, if $g$ is the last group containing small items, or the second to last bin $B'$ in $\mathtt{C}$ remains well-covered, we do nothing else, because all invariants hold (except for gap bins which may be present in preceding groups or a parallel chain).
Otherwise, $B'$ is not covered, but all bins between $B$ and $B'$ are well-covered.
In this case, we repeatedly move the largest small item from the subsequent
group $\group_{g+1}$ to $B'$ until either $\group_{g+1}$ does not contain small items anymore or $B'$ becomes well-covered.
Afterwards, each bin in the chain $\mathtt{C}$ starting from $B$ has the correct structure with respect to the correct distribution of small items (invariants \ref{DIvariants:distribution_big_small_groups} and \ref{DInvariants:wellpacked_groups}).
Note that, if there is a small item in $\group_{g+1}$ that can cover a bin together with the big item from $B'$, than the same holds for the biggest small item from $\group_{g+1}$ and $B'$ did not contain small items before, because these would be even bigger (due to \ref{DIvariants:distribution_big_small_groups}) and cover the bin.
Hence, in that case only one item would be pulled from $\group_{g+1}$ and it would be the only small item in $B'$ afterwards.
Furthermore, note that at most one bin from $\group_{g+1}$ can lose all of its
small items in the above process, if it was well-covered to begin with, and if
this happens, all the pulled items were pulled from this very bin (due to the distribution of big items \ref{DIvariants:distribution_big_small_groups}).
Therefore, each small item that was pulled from $\group_{g+1}$ was previously
contained in the first bin of a chain.

If we stopped pulling items from $\group_{g+1}$ because it did not contain small items anymore, we do nothing else, because then $\group_{g}$ has become the last group containing small items, and $B'$ does not have to be covered with respect to invariant \ref{DInvariants:wellpacked_groups}.
Hence, we consider the case that $\group_{g+1}$ still contains small items.
By pulling the items from $\group_{g+1}$, we might have created multiple maximal
gap bins.
We now could consider each such bin individually to recurse, but in order to
have bounded migration, we need to guarantee that only a constant number of
maximal gap bins, namely four, have to be considered. We will achieve this by
moving items within the first bins of the chains of $\group_{g+1}$ such that at
most four gap bins remain.

To this end, we first consider the case that $\group_{g+2}$ does also contain
small items, \ie $\group_{g+1}$ is not the last group containing small items.
Let $\mathtt{Fst}$ be the set of first bins in chains belonging to $\group_{g+1}$.
We may assume that $\mathtt{Fst}$ contains more than four bins that are not well-covered because otherwise there is nothing to do.
In this case, each of the bins in $\mathtt{Fst}$ still contains at least one small item and no small item in the group can cover a bin together with the biggest big item in the group (see considerations above).
We pick a set $\mathtt{F}\subseteq\mathtt{Fst}$ with $|\mathtt{F}|=4$ and the
property that a smallest small item contained in a bin belonging to $\mathtt{F}$
is at least as big as a smallest small item from a bin belonging to
$\mathtt{Fst}\setminus\mathtt{F}$, \ie $s_{\min}(\fsmall(\mathtt{F}))\geq
s_{\min}(\fsmall(\mathtt{Fst}\setminus \mathtt{F}))$.
Due to that choice, any small item from $\mathtt{F}$ can be moved to any bin from $\mathtt{Fst}\setminus\mathtt{F}$ without violating the ordering of the small items (Invariant \ref{DInvariant:parallel_chain_structure}).
We now repeatedly pick a bin from $\mathtt{F}$ with maximal load and move a biggest small item from that bin to a bin $\mathtt{Fst}\setminus\mathtt{F}$ that is not covered.
By this process each bin can receive items whose size is at most twice as big as
the items that were pulled from this bin.
Furthermore, for each bin $B''$ contained in $\mathtt{F}$, the overall size of small items contained in $B''$ before small items were pulled to $\group_{g}$ was at least half as big as the overall size of items that were pulled.
This is due to the distribution of big items (invariant \ref{DIvariants:distribution_big_small_groups}) and due to the fact that there was no small item in $\group_{g+1}$ big enough to cover a bin together with the big item from $B'$ in this case (see considerations above).
Hence, each bin in $\mathtt{Fst}\setminus\mathtt{F}$ can be covered by this
process and we can recurse on the remaining at most four gap bins in
$\mathtt{F}$.

If, on the other hand, $\group_{g+1}$ is the last group containing small items there are two cases.
In the case that $g<G$, we can use the same approach as before with one difference.
In this case there can be bins that are positioned at the start of the chain but do not need to be covered, because there are no small items in any subsequent bin of the chain.
Hence, we can simply exclude these bins in the considerations.
Lastly, if $g=G$, we do not have parallel chains in $\group_{g+1} = \Sbins$ and all items have been pulled from the same bin.
In any case, for each maximal gap bin $B''$ in  $\group_{g+1}$ we call $\gapfill(B'')$.

Note that the call of $\gapfill(B)$ repairs the gap bin $B$ as well as all created gap bins via recursive calls.
Hence, after the call of $\gapfill(B^*)$ in the deletion procedure all invariants and the correct distribution of small items (\ref{DIvariants:distribution_big_small_groups} and \ref{DInvariants:wellpacked_groups}) in particular hold.

\begin{lemma}\label{lem:dyn_del_small_migration}
Let $B$ be a maximal gap bin.
The overall size of items that are migrated due to a call of
$\gapfill(B)$ is at most $\Oh((1/\eps^3)\cdot \log^2(1/\eps))$.
Furthermore, the migration due to the deletion of an item $i^*$ is
bounded by $\Oh((1/\eps^4)\cdot \log^2(1/\eps))s(i^*)$.
\end{lemma}
\begin{proof}
The initial call of $\gapfill(B)$ may cause several recursive
calls in later groups. Let $B^*$ be a maximal gap bin for which
such a call occurs, and $\mathtt{C}^*$ its chain, $\group_{g^*}$
its group, $\varphi^* = (1-s(B^*))$ be its free space (at the time
of the call), and $\psi^*$ be the maximal size of a small item in
any subsequent bin. We first consider the case $g^*< G+1$. The
call of $\gapfill(B^*)$ results in a call of $\chainpull$. Let
$X_{r}$ be the $r$-th set pulled out of a bin due to this call. A
simple inductive application of Lemma \ref{lem:chainpull_basic}
yields that $s(X_{r}) \leq (1 - s(B^*)) + r \psi^*$, and we have
$r\leq 1/\eps$. Hence, the overall size of items pulled from the
last non-buffer bin $B'$ of $\mathtt{C}^*$ is at most $(1 -
s(B^*)) + \psi^*/\eps$, and the overall size of items migrated due
to the call of $\chainpull$ is at most $1/\eps((1 - s(B^*)) +
\psi^*/\eps )$. In the next step, items may be pulled to $B'$ from
the next group, and because of the above observation their size is
as well upper bounded by $(1 - s(B^*)) + \psi^*/\eps+\psi^*$. This
in turn implies, that the repacking that may be caused in the next
group in order to bound the number of emerging gap bins is upper
bounded by $\Oh((1 - s(B^*)) + \psi^*/\eps)$. After this
repacking, at most $4$ maximal gap bins remain in the next group.
It is easy to see that for each of these gap bins $B''$ we have
$(1-s(B'')) \leq (1 - s(B^*)) + \psi^*/\eps + \psi^*$ as well.
Summing up, the migration directly caused by the call of $\gapfill(B^*)$ (excluding recursive calls) is bounded by $\Oh(1/\eps((1 - s(B^*)) + \psi^*/\eps ))$.
Due to Remark \ref{rem:dyn_del} this bounds also holds for the case $g^{*}= G+1$.

Now, let $g_0$ be the index of the group $B$ is positioned in, and $g_0 + K$ the last group in which a (recursive) call of $\gapfill$ occurs.
Furthermore, let $g_k = g_0 + k$ for $k\in\set{0,\dots, K}$, $B_k$ a maximal gap bin created in $\group_{g_k}$, and $\psi_k$ the maximal size of a small item in any subsequent bin.
Note that at most $4^k\leq 4^G$ such gap bins are created in $\group_{g_k}$ for $k\leq G$ and the bound of $4^G$ also holds for the last group.
Clearly, we have $B_0 = B$ and we set $\psi = \psi_0$.
The observations above yield that the migration directly caused by the call of $\gapfill(B_k)$ is at most
\[\Oh(1/\eps((1 - s(B_k)) + \psi_k/\eps )) = \Oh(1/\eps((1 - s(B)) + (G + 1) \psi/\eps )),\]
because $\psi_r\leq\psi$ for each $r\in\set{0,\dots, K}$ (due to invariant \ref{DIvariants:distribution_big_small_groups}) and $K\leq G$.
Therefore the overall migration can be bounded by:
\begin{align*}
&\Oh(G\times 4^G\times 1/\eps((1 - s(B)) + G \psi/\eps )) \\
&= \Oh((1/\eps^3)\cdot \log(1/\eps)((1 - s(B)) + (1/\eps)\cdot \log(1/\eps)\psi))\\
&\leq \Oh((1/\eps^{4})\cdot \log^2(1/\eps)[(1 - s(B)) +\psi]) .
\end{align*}
In the case that an item $i^*$ is deleted, we have $(1 - s(B))\leq s(i^*)$ and $\psi \leq s(i^*)$, yielding the proposed bound for that case.
\end{proof}

\subsection{Big and Medium Items}

To deal with the arrival or departure of big and medium items, we mainly use the ideas developed in the static case.
However, changing bins belonging to $\BSbins$ is now much more complicated because we have to be careful to maintain the group and chain structure (invariants \ref{DInvariants:Groupstructure} and following).
For example, removing a big item and its bin from $\BSbins$ will decrement the size of the corresponding group and possibly $\ceil{\eps |\BSbins|}$.
Furthermore, the corresponding parallel chain may become too short.

In this section, we first design some auxiliary procedures to carefully manipulate the group structure.
Utilizing these, we next describe the deletion and insertion procedures of big items, and lastly the procedures for medium items.
The latter are much simpler but may involve the insertion or deletion of big items from $\BSbins$ and hence are dealt with in the end.

\paragraph{Preliminaries.}

We introduce a series of auxiliary procedures used for small local changes in the groups maintaining their internal structure.
These procedures use $\gapfill$ and the insertion procedure for small items as subroutines, and are used in situations in which the group structure (invariant \ref{DInvariants:Groupstructure}) may be slightly disturbed.
In particular, the sizes of the groups may be different and $G$~--~the total
number of groups~--~may be disturbed by some additive constant.
It is easy to see that this is not a problem regarding the analysis of $\gapfill$ and the insertion procedure for small items, because the sizes of the groups do not influence and are not influenced by these procedures; and a change to $G$ by some additive constant merely results in some additional constant factor in the analysis of their migration.
Furthermore, $\BSbins$ and $\BBbins$ may not be balanced (invariant \ref{DInvariant:BB=BS}) before or after the call of one of the auxiliary procedures.
This will be handled, when they are used.

\begin{itemize}

\item $\groupinsert(i,\group_{g})$ is called to insert a big item
$i$ into a group $\group_{g}$ with $g<G+1$. To do so, we create a
new bin $B$ containing $i$. There are two cases. If $\group_{g}$
contains a parallel chain $\mathtt{C}$ with length smaller than
$1/\eps + 1$, bin $B$ is inserted into $\mathtt{C}$ at a position
that maintains the correct distribution of big items in the chain
(invariant \ref{DInvariant:parallel_chain_structure}). More
precisely, right before the buffer bin, if all big items in
$\mathtt{C}$ are at least as big as $i$, or right in front of the
first bin containing a big item that is smaller than $i$. If $B$
was inserted right in front of the buffer bin, we adjust the size
of the virtual big item placed in the buffer bin of $\mathtt{C}$
to be $s(i)$, thereby maintaining invariant
\ref{DInvariants:group_buffer_bin}. Otherwise, a new parallel
chain is created in $\group_{g}$ containing only $B$ and a buffer
bin whose virtual big item has size $s(i)$. Lastly, $B$ now may be
a maximal gap bin and we call $\gapfill(B)$ in this case.
Afterwards, the parallel chain structure (invariant
\ref{DInvariant:parallel_chain_structure}), as well as the correct
distribution of small items (\ref{DInvariants:wellpacked_groups})
are guaranteed.

\item $\groupdelete(B,\group_{g})$ is used to repair the internal structure of a group~$\group_{g}$ with $g<G+1$ after the big item of a bin $B$ positioned in $\group_{g}$ was removed.
Let $\mathtt{C}$ be the parallel chain $B$ is positioned in.
We remove all small items from $B$ and delete the bin.
Next, we remove all small items from the buffer bin of $\mathtt{C}$.
If $\mathtt{C}$ now has length one, we delete the chain, and if additionally $\group_{g}$ has become empty, we also delete the group.
Otherwise, we consider two cases: If $\mathtt{C}$ is the only
chain with length smaller than $1/\eps + 1$, we adjust the size of
the virtual big item in its buffer to be the same as the size of
the big item in its direct predecessor, thereby maintaining
invariant \ref{DInvariants:group_buffer_bin}. In the second case
there is another chain $\mathtt{C}'$ with length smaller $1/\eps +
1$ and hence $\mathtt{C}$ has length $1/\eps$ (due to
\ref{DInvariant:parallel_chain_structure}). We remove the last
non-buffer bin $B'$ from $\mathtt{C}'$, and remove all small items
from $B'$ and the buffer bin of $\mathtt{C}'$. Then, we insert
$B'$ into $\mathtt{C}$ at the position maintaining the correct
distribution of big items
(\ref{DInvariant:parallel_chain_structure}). If $B'$ was inserted
right in front of the buffer bin, we adjust the size of the
virtual big item placed in the buffer of $\mathtt{C}$ (maintaining
\ref{DInvariants:group_buffer_bin}). If $\mathtt{C}'$ now has
length $1$, we delete it, and otherwise we adjust the size of the
corresponding virtual big item in its buffer bin as well, again
maintaining \ref{DInvariant:parallel_chain_structure} and
\ref{DInvariants:group_buffer_bin}. Now $B'$ may be a maximal gap
bin and we call $\gapfill(B')$ in this case. Lastly, we reinsert
all small items that have been removed. Note that afterwards the
parallel chain structure (invariant
\ref{DInvariant:parallel_chain_structure}), as well as the correct
distribution of small items (\ref{DInvariants:wellpacked_groups})
are guaranteed.

\item $\grouppush(\group_{g})$ is used to decrement the size of a group $\group_{g}$, with $g\leq G$, by pushing a big item of minimal size to the next group.
We select a non-buffer bin $B$ from $\group_{g}$ containing a big item $i$ of minimal size, remove $i$ from $B$, and call $\groupdelete(B,\group_{g})$.
If $g < G$, we then call $\groupinsert(i,\group_{g+1})$.
Otherwise, we create a new group belonging to $\BSbins$ with only one parallel chain.
The chain has two bins, namely a bin containing only $i$ and a buffer bin whose virtual item has size $s(i)$.

\item $\grouppull(\group_{g})$ is used to increment the size of a group $\group_{g}$, with $g< G$, by pulling a big item of maximal size from the next group.
We select a bin $B$ from $\group_{g + 1}$ containing a big item $i$ of maximal size, remove $i$ from $B$, call $\groupdelete(B,\group_{g+1})$, and then  $\groupinsert(i,\group_{g})$.

\end{itemize}

\begin{lemma}\label{lem:dyn_migration_auxil}
The overall size of items migrated due to a call of one of the
auxiliary procedures above is at most $\Oh((1/\eps^3)\cdot
\log^2(1/\eps))$.
\end{lemma}
\begin{proof}
In both $\groupdelete$ and $\groupinsert$ at most one big item and
small items with overall size of at most $4$ are removed. The
small items are reinserted causing migration of at most
$\Oh((1/\eps^3)\cdot \log^2(1/\eps))$ (see Lemma
\ref{lem:dyn_in_small_migration}), and there is a call of
$\gapfill$  yielding the same bound (see Lemma
\ref{lem:dyn_del_small_migration}). Both in $\grouppush$ and
$\grouppull$ one big item is moved and the other two auxiliary
procedures are called once.
\end{proof}

\paragraph{Deletion of Big Items.}

When a big item $i^*$ is deleted, there are three cases we have to
consider based on the bin $B^*$ that item $i^*$ was positioned in.
Firstly, if $B^*\in\BMbins$, we simply delete the bin and reinsert
the corresponding medium items starting with the largest one.
Next, if $B^*\in\BBbins$, we have to be a bit more careful in
order to keep the sets $\BSbins$ and $\BBbins$ balanced (invariant
\ref{DInvariant:BB=BS}). If $|\BBbins| \geq |\BSbins|$ did hold
before, we select a bin $B\in\BBbins$ containing a big item $i$ of
maximal size (with a preference for $B^*$ itself), replace $i$
with the remaining item from $B^*$, and reinsert $i$ (thereby
deleting $B^*$). If, on the other hand, $|\BBbins| = |\BSbins| -
1$ did hold, we select a big item $i$ of minimal size from
$\BSbins$ and move it to $B^*$. Afterwards, \ref{DInvariant:BB=BS}
holds, but the removal of $i$ from $\BSbins$ may have violated
several invariants concerning the structure of $\BSbins$. To deal
with this, we use the deletion procedure for the third case
$B^*\in\BSbins$, which is by far the most complicated and
described in the following.

Let $\group_{g}$ be the group $B^*$ is positioned in. We perform
$\groupdelete(B^*,\group_{g})$. Now, $|\group_{g}|$ is decreased
by one and $\ceil{\eps |\BSbins|}$ may have decreased by one. We
first deal with the former and then consider the latter. In
particular, if $g<G$, we iteratively perform
$\grouppull(\group_{g'})$ for one group after another with $g' =
g,\dots, G-1$. Afterwards, $G$ may have been decreased by one and
we have $|\group_g| = 2^{g-1}\ceil{\eps (|\BSbins| + 1)}$ for
$g<G$ and $|\group_G| \leq 2^{G-1}\ceil{\eps (|\BSbins| + 1)}$. If
$\ceil{\eps (|\BSbins| + 1)} = \ceil{\eps |\BSbins|}$ the group
structure, that is, \ref{DInvariants:Groupstructure}, again holds.
Otherwise, there are $2^{g-1}$ too many bins positioned in group
$\group_{g}$ for $g<G$ and up to $2^{G-1}$ too many in
$\group_{G}$. Hence, for each $g'\in[G-1]$, we perform $\sum_{r =
1}^{g'}2^{r-1} = 2^{g'}-1$ many calls of
$\grouppush(\group_{g'})$. Afterwards, there are $\ell\leq 2^G -
1$ too many bins in $\group_{G}$ and we perform $\ell$ calls of
$\grouppush(\group_{G})$, thereby creating up to one additional
group. The last step is to rebalance $\BSbins$ and $\BBbins$, if
necessary, that is, if $|\BSbins|=|\BBbins|-2$, we select a bin
$B\in\BBbins$ containing a big item of maximal size. If necessary,
we switch the second item in $B$ such that the big items in $B$
are as big as any other in $\BBbins$. Then we delete $B$ and
reinsert the big items. These big items will be inserted into
$\BSbins$ but will not incur further recursive calls for insertion
or deletion of big items.

\subparagraph{Insertion of Big Items.}

The steps that we take, when a big item $i^*$ has to be inserted in the dynamic case are very similar to the static case:
We consider inserting $i^*$ into $\BMbins$, $\BSbins$ and $\BBbins$ in this order and based on the same conditions.
The insertion for $\BMbins$ works exactly as in the static case and the one for $\BBbins$ is changed only slightly.
In particular, there are two cases considered in the procedure, namely $|\BSbins| = |\BBbins| + 1$ and $|\BSbins| < |\BBbins| + 1$.
In the former, a bin $B\in\BSbins$ is selected containing a big item $i$ of minimal size, and this item $i$ is used to form a new bin together with $i^*$.
The only change to the procedure is that we deal with the removal of $i$ from $\BSbins$ using the deletion procedure for big items.
We now describe how to insert $i^*$ into $\BSbins$.

When $i^*$ is to be inserted into $\BSbins$ we first select the
correct group $\group_{g}$ with respect to the distribution of big
items (invariant \ref{DIvariants:distribution_big_small_groups}),
that is, $g=G$ if $i^*$ is at most as big as any other big item in
$\BSbins$, or the first group containing a big item smaller than
$i^*$ otherwise. We call $\groupinsert(i^*, \group_g)$. Now,
$|\group_{g}|$ is increased by one and $\ceil{\eps |\BSbins|}$ may
have increased by one, that is, invariant
\ref{DInvariants:Groupstructure} may not hold. We first deal with
the former and then consider the latter. In particular, we
iteratively call $\grouppush(\group_{g'})$ for one group after
another with $g' = g,\dots, G-1$. Furthermore, if the last group
was of full size, that is $|\group_{G}| = 2^{G-1}\ceil{\eps
|\BSbins|}$ before the above operations, we also call
$\grouppush(\group_{G})$. Afterwards, $G$ may have been increased
by one and we have $|\group_g| = 2^{g-1}\ceil{\eps (|\BSbins| -
1)}$ for $g<G$ and $|\group_G| \leq 2^{G-1}\ceil{\eps (|\BSbins| -
1)}$. Hence, if $\ceil{\eps (|\BSbins| - 1)} = \ceil{\eps
|\BSbins|}$ the group structure(\ref{DInvariants:Groupstructure})
again holds. Otherwise, there are $2^{g-1}$ bins too few
positioned in group $\group_{g}$ for $g<G$.For each $g'\in[G-2]$, we perform $\sum_{r = 1}^{g'}2^{r-1} =
2^{g'}-1$ many calls of $\grouppull(\group_{g'})$. Furthermore, if
$G>1$, we perform $\min\set{2^{G-1}-1, |\group_{G}|}$ many calls
of $\grouppull(\group_{G-1})$, thereby possibly deleting the last
group. Lastly, the insertion of $i^*$ might have violated
\ref{DInvariant:BB=BS}, that is, we now have $|\BSbins| =
|\BBbins| + 2$. In this case, we select a big item $i$ from
$\BBbins$ and create a bin containing two dummy items of the same
size as $i$. Afterwards \ref{DInvariant:BB=BS} and all other
invariants hold. We then delete the two dummy items one after
another using the deletion procedures described above.

\paragraph{Medium items.}

Dealing with medium items in the dynamic case is not a big problem, because we can essentially use the same ideas as in the static case.
In fact, we have to make only one small adjustment to the insertion procedure in the static case:
When the arriving medium item has been inserted into $\Mbins$ and we have the situation that $s(\Mbins) \geq 1 - s_{\max}(\BSbins\cup\BBbins)$ and $\BSbins\neq\emptyset$, then a big item $i$ is to be removed from $\BSbins$.
We handle this as follows:
We replace $i$ by a dummy item of the same size, form a new bin $B$ containing $i$ and call $\greedypull(B,\Mbins)$.
Afterwards \ref{DInvariant:MediumBins} holds.
Then we delete the big dummy item and call the corresponding deletion procedure.

When a medium item $i^*$ is deleted from its bin $B^*$, we distinguish a few cases.
If $B^*$ remains barely covered, we do nothing.
Otherwise, if $B^*\in\Mbins$, we delete the bin and reinsert all the medium items.
Note that these items will simply be reinserted into $\Mbins$ via $\greedypush$ and cause no further migration, because of invariant \ref{DInvariant:MediumBins}.
Lastly if $B^*\in\BMbins$, we call $\greedypull(B^*,\Mbins)$.
If $B^*$ is barely covered afterwards, we do nothing else.
Otherwise, $\Mbins = \emptyset$ and we remove the big item $i$ from $B^*$ (yielding $\Mbins = \set{B^*}$), and reinsert $i$.
Note that in this case $i$ will not be inserted into $\BMbins$.

\begin{lemma}\label{lem:dyn_migration_big_med}
The overall size of items that are migrated due to the insertion
or deletion of a big or medium item $i^*$ is $\Oh((1/\eps)^4\cdot
\log^2(1/\eps))$.
\end{lemma}
\begin{proof}
We consider the same cases as above and in the same order.

If $i^*$ is big and deleted from a bin $B^*\in\BMbins$, then medium items with size at most $1$ are reinserted.
Because the largest item is reinserted first and because of invariants \ref{DInvariant:MediumBins} and \ref{DInvariant:BigItems}, only one of these items can trigger further migration and all others will be inserted into $\Mbins$.
Hence, we get the same migration as for medium items plus some constant.

If $B^*\in\BBbins$, the deletion may cause a reinsertion of a big item, or the deletion of a big item positioned in $\BSbins$, as well as $\Oh(1)$ further migration.
In the first case, the corresponding big item will be inserted into $\BSbins$, because of its size, and because of the considered case, the insertion cannot cause further deletions of big items.
Similarly, in the second case, the deletion cannot cause further insertions of big items.
Hence, the migration is given by the migration for insertions or deletions of big items into $\BSbins$ plus some constant.

If $B^*\in\BSbins$ we may perform one call of $\groupdelete$,
$\Oh(G) = \Oh(\log (1/\eps))$ calls of $\grouppull$, $2^{G+\Oh(1)}
= \Oh(1/\eps)$ calls of $\grouppush$ and some additional constant
migration. Furthermore, there may be two big items that are
reinserted. Because of the choice of the items and invariant
\ref{DInvariant:BigItems}, they are inserted into $\BSbins$, and
because of the corresponding case no deletions of big items are
caused by these insertions. Hence, the migration in this case is
at most $\Oh((1/\eps)^4\cdot \log^2(1/\eps))$ (see Lemma
\ref{lem:dyn_migration_auxil}) plus the migration due to the
possible insertions of big items.

We consider the case that $i^*$ is big and has to be inserted.
As described in Lemma \ref{lem:static_migration_big} inserting a big item into $\BMbins$ may directly cause $\Oh(1)$ migration, and the insertion of a big item into $\BSbins$ or $\BBbins$.
Furthermore, the insertion of $i^*$ into $\BBbins$ may cause $\Oh(1)$ migration and either a deletion of a big item from $\BSbins$ or the reinsertion of a big item, which will be placed into $\BSbins$.
In both cases this causes no further insertions or deletions of big items.
Hence, in both cases the migration is essentially given by the migration for the insertion or deletion of a big item into $\BSbins$.

If a big item $i^*$ is inserted into $\BSbins$, we may perform one
call of $\groupinsert$, $\Oh(G) = \Oh(\log (1/\eps))$ calls of
$\grouppush$, $2^{G+\Oh(1)} = \Oh(1/\eps)$ calls of $\grouppull$
and some additional constant migration. Furthermore, two deletions
of big items from $\BBbins$ may be caused. Because of the
corresponding case no reinsertions of big items are caused by
these deletions. Hence, the migration in this case is at most
$\Oh((1/\eps)^4\cdot\log^2(1/\eps))$ (see Lemma
\ref{lem:dyn_migration_auxil}) plus the migration due to the
possible insertions of big items.

If $i^*$ is a medium item and has to be inserted, we have $\Oh(1)$
migration and the deletion or insertion of a big item may be
caused. In case of a deletion the corresponding item was
positioned in $\BSbins$. If the item $i^*$ is deleted, this may
cause $\Oh(1)$ migration as well as the reinsertion of a big item.
Hence, the migration in this case is at most
$\Oh((1/\eps)^4\cdot\log^2(1/\eps))$.

Summarizing, the overall size of items that are migrated due to
the insertion or deletion of a big or medium item $i^*$ is
$\Oh((1/\eps)^4\cdot\log^2(1/\eps))$.
\end{proof}

\subsection{Analysis}

In this section, we finish the proof of Theorem
\ref{thm:wc_and_dynamic_alg}. We already proved bounds on the
migration that may be caused due to insertions and deletions of
algorithms (see Lemmas
\ref{lem:dyn_in_small_migration},\ref{lem:dyn_del_small_migration}
and \ref{lem:dyn_migration_big_med}). Note that medium items may
cause the migration of items with overall size
$\Oh((1/\eps)^4\cdot\log^2(1/\eps))$ but may have size $\eps$
themselves yielding the bound of
$\Oh((1/\eps)^5\cdot\log^2(1/\eps))$ stated in Theorem
\ref{thm:wc_and_dynamic_alg}. Furthermore, we argued that the
invariants are maintained by the insertion and deletion procedures
throughout the section. Using these invariants, we prove the
stated asymptotic competitive ratio and additive constant.

\begin{lemma}\label{lem:dynamic:ratio}
The algorithm for the dynamic case has an asymptotic competitive ratio of $1.5 + \eps$ with additive constant $\Oh(\log 1/\eps)$.
\end{lemma}
\begin{proof}
First, we consider the case $\BSwaitbins = \emptyset$.
In this case, the claim holds because the bins on average have not too much excess size.
More precisely, we obviously have $\Opt(I) \leq s(I)$, and the invariants \ref{DInvariant:Binstructure} and \ref{DInvariant:MediumoneUncovered} imply:
\begin{align*}
s(I) &= s(\BBbins) + s(\BSbins) + s(\Gbuffer) + s(\BMbins) + s(\Mbins) + s(\Sbins)\\
& < 2|\BBbins| + (1+\eps)|\BSfullbins| + s(\Gbuffer) + 1.5
|\BMbins| + 1.5|\Mfullbins| + 1 + (1+\eps)|\Sfullbins| +
s(\Sbins\setminus\Sfullbins) .
\end{align*}
Furthermore, we have $s(\Gbuffer) < (1/2 + \eps)(\eps |\BSfullbins| + G)\leq \eps |\BSfullbins| + G$, due to \ref{DInvariants:Groupstructure} and \ref{DInvariant:parallel_chain_structure}; $s(\Sbins\setminus\Sfullbins) < \eps |\Sfullbins| + 1$, due to \ref{DInvariants:Chainstructure}; $(0.5-\eps)|\BBbins|\leq (0.5-\eps)(|\BSbins| + 1)$, due to \ref{DInvariant:BB=BS}; and $|\BSbins|=|\BSfullbins|$, because of the case we are considering.
Hence:
\begin{align*}
\Opt(I) &< (1.5 + \eps)|\BBbins| + (1.5 + \eps)|\BSfullbins| + 1.5|\BMbins| + 1.5|\Mfullbins| + (1 + 2\eps)|\Sfullbins| + G + 3\\
& \leq (1.5 + \eps)(|\BBbins| + |\BSfullbins| + |\BMbins| + |\Mfullbins| + |\Sfullbins|) + \log(1/\eps) + 5\\
& =  (1.5 + \eps)\Alg(I) + \log(1/\eps) + 5 .
\end{align*}

A similar argument holds, if $\BSwaitbins \neq \emptyset$ and $\BBbins = \emptyset$.
In this case, we have $\BSbins=\BSwaitbins$ and $|\BSwaitbins| = 1$, because of invariant \ref{DInvariant:BB=BS}; $\Gbuffer= \emptyset$, because of \ref{DInvariants:Groupstructure} and \ref{DInvariant:parallel_chain_structure}; and $\Sbins = \emptyset$, because of invariant \ref{DInvariant:SmallBins}.
Hence:
\[\Opt(I) \leq s(I) < 1.5 |\BMbins| + 1.5|\Mfullbins| + 2 = 1.5 \Alg(I) +
2  . \]

Next, we consider the case $\BSwaitbins \neq \emptyset$ and $\BBbins \neq \emptyset$.
In this case, we have $\Mfullbins = \emptyset$, because of Invariant \ref{DInvariant:MediumBins}, and $\Sbins = \emptyset$, because of Invariant \ref{DInvariant:SmallBins}.
According to invariant \ref{DInvariants:wellpacked_groups}, there is a $g^*\in [G+1]$ such that each bin in each group preceding $\group_{g^*}$ is well-covered and each succeeding group does not contain small items, and because $\BSwaitbins \neq \emptyset$, we have $g^*< G+1$.
Let $\xi=s_{\min}(\fbig(\group_{g^*}))$ be the size of a big item from $\group_{g^*}$ with minimal size.
Because of the invariants \ref{DInvariant:BigItems} and \ref{DIvariants:distribution_big_small_groups}, the items in $\BBbins$ and each group succeeding $\group_{g^*}$ are upper bounded in size by $\xi$.
We construct a modified instance $I^*$ as follows:
\begin{enumerate}
\item The remaining big items (the ones bigger than $\xi$) are
split into a big item of size $\xi$ and a medium or small item.
Let $X$ be the set of items with size $\xi$. \item For each bin
from
$\BSfullbins\cup\BMbins\cup\sett{B}{B\in\bigcup_{g\in[G]}\VGbuffer_g,
B\text{ is well-covered}}$, we increase the size of a biggest item
that is not big to $0.5$. Let $Y$ be the set of these items and
$Z$ be the set of the remaining items not belonging to $X$ or $Y$.
\end{enumerate}
Note that $\Opt(I)\leq\Opt(I^*)$.
For some optimal solution for $I^*$ in which each bin is at most barely covered, let $k_2$, $k_1$ and $k_0$ be the numbers of covered bins with $2$, $1$ or $0$ items from $X\cup Y$ respectively.
We have:
\[2k_2 + k_1 = |X\cup Y| \overset{\text{\ref{DInvariant:parallel_chain_structure}}}{\leq} (2|\BBbins| + |\BSbins| + |\BMbins|) + (|\BMbins| + |\BSfullbins| + \eps|\BSfullbins| + G) . \]
Furthermore, $(1-\xi)k_1 + k_0 \leq s(Z)$ and, due to \ref{DInvariant:parallel_chain_structure}, \ref{DInvariants:Groupstructure}, \ref{DInvariant:MediumBins}, \ref{DInvariants:wellpacked_groups} and the choice of $\xi$, we have:
\begin{align*}
s(Z) &\leq (1-\xi)(|\BMbins| + |\BSfullbins| + |\group_{g^*}| + |\sett{B\in\bigcup_{g\in[G]}\VGbuffer_g}{\fsmall(B)\neq \emptyset}|) + s(\Mbins)\\
& \leq (1-\xi)(|\BMbins| + |\BSfullbins| + |\group_{g^*}| +\eps |\BSfullbins| + G) + (1 - s_{\max}(\BSbins))\\
& \leq (1-\xi)(|\BMbins| + (1+\eps)|\BSfullbins| + \ceil{\eps|\BSbins|} + \sum_{g'= 1}^{g^*-1} |\group_{g'}| + G + 1)\\
& \leq (1-\xi)(|\BMbins| + (1+\eps)|\BSfullbins| + \ceil{\eps|\BSbins|} + |\BSfullbins| + G + 1)\\
&\leq (1-\xi)(|\BMbins| + (2+\eps)|\BSfullbins| + \eps|\BSbins| +
G + 2) .
\end{align*}
Hence:
\begin{align*}
2\Opt(I)   &\leq 2(k_2 + k_1 + k_0)\\
&\leq (2k_2 + k_1) + (k_1 + (1-\xi)^{-1}k_0)\\
&\leq  (2|\BBbins| + |\BSbins| + |\BMbins|) + (|\BMbins| + |\BSfullbins| + \eps|\BSfullbins| + G)\\
&\quad +(|\BMbins| + (2+\eps)|\BSfullbins| + \eps|\BSbins| + G + 2)\\
& = 2|\BBbins| + (1+\eps)|\BSbins| + 3|\BMbins| + (3+2\eps)|\BSfullbins| + 2G + 2\\
&\overset{\text{\ref{DInvariant:BB=BS}}}{\leq} 2|\BBbins| + (1+\eps)(|\BBbins| + 1) + 3|\BMbins| + (3+2\eps)|\BSfullbins| + 2G + 2\\
& \overset{\text{\ref{DInvariants:Groupstructure}}}{\leq} (3+\eps)|\BBbins| + 3|\BMbins| + (3+2\eps)|\BSfullbins| + 2\log(1/\eps) + 8\\
& \leq (3+2\eps)\Alg(I) + 2\log(1/\eps) + 8 . \qedhere
\end{align*}
\end{proof}
Lastly, we conclude the proof of Theorem \ref{thm:wc_and_dynamic_alg} by arguing that the running time bound holds.
\begin{lemma}
The insertion and deletion procedures in the dynamic case have polynomial running time in both the input size and $1/\eps$.
\end{lemma}
\begin{proof}
We first consider the insertion or deletion of small items. When
one of these procedures is called, there may be $\Oh(1/\eps)$ many
recursive calls. For each of these calls it is rather ease to see
that the running time is linear in the number of items. Concerning
the insertion or deletion of big or medium items, we already saw
in the migration analysis (Lemma \ref{lem:dyn_migration_big_med})
that only a constant number of recursive calls of big and medium
items may be caused. Furthermore, small items (and there are at most $n$ such items) may removed or reinserted
due to these calls.
\end{proof}

\section{Amortized Migration}
We will now take a look at the setting, where the migration is amortized.
Here we are concerned with the maximum ratio of the migration performed
divided by the total size of items that have arrived up to some point.
This is easier than the standard non-amortized migration considered before.
\subsection{Dynamic Case}
For the dynamic case, we will show that the lower bound of $3/2$ on the
competitive ratio also holds if amortized migration is allowed. Hence, our
above algorithm that only uses non-amortized migration is optimal (except for an
extra $\eps$). This result can
also be proved for the other definition of asymptotic competitive
ratio.

\begin{proposition}
  \label{thm:ac_and_dynamic_alg}
There is no algorithm for dynamic online bin covering with a constant amortized
migration factor $\beta$ and an asymptotic competitive ratio smaller than $3/2$.
\end{proposition}
\begin{proof}
Let $\Alg$ be an algorithm for dynamic online bin covering with an amortized
migration factor of $\beta$. Suppose that $\Alg$ has asymptotic
competitive ratio $3/2-\delta$ for some constant value $1/2 \geq \delta > 0$, \ie for each instance
we have $\OPT(I)\leq (3/2-\delta)\cdot \Alg(I)+c$ for some constant $c$. We will now construct for each even integer $N$
 an instance $I^{(N)}$. This instance is comprised of $m + 1$ phases, where $m$
  is an even number  that is specified later on and depends on $N$ and $\beta$. In
  each of the phases some number of equally sized items either arrives or
  departs. The instance corresponding to the first phase is called $I^{(N)}_0$
  and the one corresponding to the first $j + 1$ phases is called $I^{(N)}_j$.
  The optimal (offline) objective value for each of these instances will be a
  multiple of $N$, \ie $\OPT(I^{(N)}_{j})\in \{3N, (3/2)N\}$. We will show that the ratio between this value and the objective
  value achieved by $\Alg$ is at least $3/2$ for at least one value of $j\in
  \{0,\ldots,m\}$ with constant amortized migration factor.

We now describe the instance $I^{(N)}$: In the first phase $3N$ big items with size $1-\eps$ arrive; and for each odd $j\in[m]$ there are $3N$ small items with size $\eps$ arriving in phase $j$ and departing in phase $j+1$.
We choose the parameter $\eps$ such that $3N$ small items cannot fill a whole
bin, and such that the total size of
$3N$ small items is upper bounded by the size of a big item, e.g., $\eps =
\min\set{(9\beta N + 2)^{-1}, (6N)^{-1}}$. Note that for each $j\in\set{0,\dots, m}$ the optimal (offline) objective value
for $I^{(N)}_{j}$ is $3N/2$ if $j$ is even, and $3N$ if $j$ is odd (see
Fig.~\ref{fig:lb_dynamic}).

\begin{claim}
  For each even integer $N$, there is some index $j\in \{1,\ldots,m\}$ such
  that $(3/2)\Alg(I^{(N)}_{j})\leq \OPT(I^{(N)}_{j})$.
\end{claim}

\begin{proof}
If $\Alg(I^{(N)}_{j})\leq N$ for some even $j$ or
$\Alg(I^{(N)}_{j}) \leq 2N$ for some odd $j$, we are done. Hence,
let $\mathfrak{P}$ denote the property that $\Alg(I^{(N)}_{j}) >
N$ for all even $j$ and $\Alg(I^{(N)}_{j}) > 2N$ for all odd $j$,
and we analyze whether $\mathfrak{P}$ can be obtained by $\Alg$.
Since two big items are needed to cover a bin and $3N$ small items
are not sufficient for  covering a bin, $\mathfrak{P}$ implies
that at least $N+1$ bins have to contain at least $2$ big items,
if $j$ is even; and at least $N+1$ bins have to contain exactly
one big item, if $j$ is odd.

Hence, to maintain that $\mathfrak{P}$ holds in two consecutive
phases, some big item has to be moved, and therefore the migration
performed in each phase is at least $1-\eps$. However, the total
size of items arriving or departing in every phase
$I^{(N)}_1,I^{(N)}_2,\dotsc$ is only
\begin{align*}
3N\eps \le 3N / (9\beta N + 2) < 3N/(9\beta N) = 1/(3\beta)\leq  (1-2\eps)/\beta.
\end{align*}
In the first inequality, we use $\eps \leq (9\beta N +2)^{-1}$. In
the last inequality, we use $N\geq 1$ implying $\eps \leq 1/6\leq
1/3$ and thus $(1-2\eps)\geq 1/3$.

Hence, the total size of migrated items is at least $m(1-\eps)$
and the total load of items arriving or departing is strictly
smaller than $(3N)(1-\eps)+m(1-2\eps)/\beta$. As $\Alg$ has
amortized migration factor $\beta$, it is allowed to reassign a
total size of $\beta S$, after items of size $S$ have arrived or
departed. The total size of arrived and departed items equals
$(3N)(1-\eps)+m(1-2\eps)/\beta$ in our setting. Hence, $\Alg$ is
allowed to migrate a total size $(3N\beta)(1-\eps)+m(1-2\eps)$.
Choosing $m$~--~the number of phases~--~such that
$m>(3N\beta)(1-\eps)/\eps$ then gives
$(3N\beta)(1-\eps)+m(1-2\eps) < m(1-\eps)$. Hence $\Alg$ can not
migrate a total size of $m(1-\eps)$. As it  needs to migrate at
least such a total size to maintain property $\mathfrak{P}$, we
can conclude that $\mathfrak{P}$ can not hold. Hence, there is
some index $j\in \{1,\ldots,m\}$ such that
$\Opt(I^{(N)}_{j})/\Alg(I^{(N)}_{j})\geq 3/2$, concluding our
proof.
\end{proof}

This is sufficient to prove the proposition: From above, we
know that for each even $N$, there is some value $j$ such that
$(3/2)\cdot \Alg(I^{(N)}_{j})\leq \Opt(I^{(N)}_{j})$. The
asymptotic competitive ratio of $\Alg$ then implies
$\Opt(I^{(N)}_{j}) \leq (3/2-\delta)\cdot \Alg(I^{(N)}_{j})+c$ and
thus
\begin{align*}
  (3/2)\cdot \Alg(I^{(N)}_{j}) \leq (3/2-\delta)\cdot \Alg(I^{(N)}_{j})+c.
\end{align*}
This is equivalent to $c/\delta \geq  \Alg(I^{(N)}_{j})$. As
$\Alg$ has asymptotic competitive ratio $3/2-\delta$ and
$\Opt(I^{(N)}_{j})=k_{j}N$, we have $c/\delta \geq
\Alg(I^{(N)}_{j})\geq \frac{k_{j}N-c}{3/2-\delta}$ and thus
$c(3/2-\delta)/\delta +c\geq k_{j}N$. As $N$ can be made
arbitrarily large, this is a contradiction to the assumption that
both $c$ and $\delta$ are constant values. Hence, an asymptotic
competitive ratio of $3/2-\delta$ is not achievable with constant
amortized migration factor $\beta$.

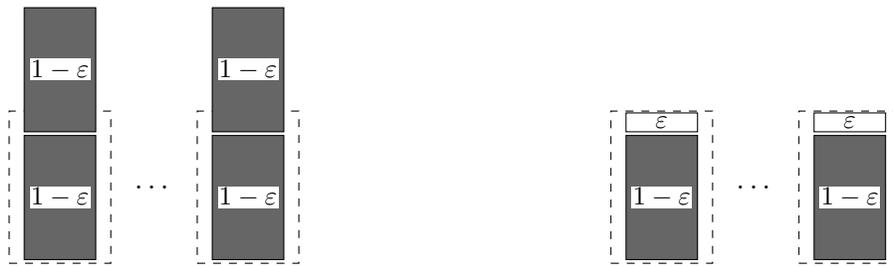
\begin{figure}[h]
  \centering
\begin{tikzpicture}
  \drawBin{2}{}{b1}
  \drawBin{4.5}{}{b2}
  \node at ($ (b1)!0.5!(b2)$) {\ldots};
  \addToBin{2}{0}{.85}{\tikz\node[fill=white,inner sep=.5]{\small $1-\eps$};}{fill=big}{}
  \addToBin{2}{.85}{.85}{\tikz\node[fill=white,inner sep=.5]{\small $1-\eps$};}{fill=big}{}
  \addToBin{4.5}{0}{.85}{\tikz\node[fill=white,inner sep=.5]{\small $1-\eps$};}{fill=big}{}
  \addToBin{4.5}{.85}{.85}{\tikz\node[fill=white, inner sep=.5]{\small $1-\eps$};}{fill=big}{}
  \node[yshift=-2.4cm] at ($ (b1)!0.5!(b2)$) {$\Opt(I_{2i})=3/2 N$};
  \drawBin{10}{}{b3}
  \drawBin{12.5}{}{b4}
  \node at ($ (b3)!0.5!(b4)$) {\ldots};
  \addToBin{10}{0}{.85}{\tikz\node[fill=white, inner sep=.5]{\small $1-\eps$};}{fill=big}{}
  \addToBin{10}{.85}{.15}{\small $\eps$}{fill=small}{}
  \addToBin{12.5}{0}{.85}{\tikz\node[fill=white, inner sep=.5]{\small $1-\eps$};}{fill=big}{}
  \addToBin{12.5}{.85}{.15}{\small $\eps$}{fill=small}{}
  \node[yshift=-2.4cm] at ($ (b3)!0.5!(b4)$) {$\Opt(I_{2i+1})=3N$};

\end{tikzpicture}

  \caption{Optimal packings of $I_{2i}$ and $I_{2i+1}$}
  \label{fig:lb_dynamic}
\end{figure}
\end{proof}

\subsection{Static Case}
In contrast to the dynamic case, where amortized migration does
not help to improve upon the lower bound of $3/2$ on the
competitive ratio, amortization allows to design a simple
algorithm for the static case, achieving a competitive ratio of
$1+\eps$ with amortized migration of $\Oh(1/\eps)$. Some of the
ideas applied here were used for bin packing
\cite{DBLP:conf/icalp/FeldkordFGGKRW18} in the past. The small
value $\eps$ will again satisfy $0<\eps \leq 0.5$.

In order to achieve this goal, we make use of an asymptotic fully
polynomial time approximation scheme (AFPTAS), denoted by
$\Alg_{\off}$, for the offline version of the problem due to
Jansen and Solis-Oba \cite{jansen2003bincovering}. For an instance
$I$ and an approximation parameter $\eps$, the algorithm
$\Alg_{\off}$ produces a packing that covers $\Alg_{\off}(I,\eps)$
bins with $ \Opt(I)\leq
(1+\eps)\Alg_{\off}(I,\eps)+\Oh(1/\eps^{3})$ in time that is
polynomial in the input size and $\frac{1}{\eps}$. Let $\mu \geq
1$ be a constant such that
\begin{align}
\label{eq:dynamic_amortized_offline} \Opt(I)\leq
(1+\eps)\Alg_{\off}(I,\eps)+\frac{\mu}{\eps^{3}}
\end{align}
 holds for every
input $I$.

Our algorithm will be called $\Alg_{\on}$ and its profit for input
$I$ and $\eps$ is denoted by $\Alg_{\on}(I,\eps)$. Denote the
instance after $\tau$ items were presented, that is, at time
$\tau$, as $I_{\tau}$. Our online algorithm $\Alg_{\on}$ will have
the same guarantee as the offline algorithm in terms of the order
of growth of the additive term, \ie we will show
$\Opt(I_{\tau})\leq
(1+\Oh(\eps))\Alg_{\on}(I_{\tau},\eps)+\Oh(1/\eps^{3})$, and more
specifically, we will show that $\Opt(I_{\tau})\leq
(1+3\eps)\Alg_{\on}(I_{\tau},\eps)+\frac{6\mu}{\eps^{3}}$ for
every time $\tau$. There will be a sequence of special times
denoted by $t_1,t_2,\cdots$ for which it will hold that
$\Alg_{\on}(I_{t_i},\eps)=\Alg_{\off}(I_{t_i},\eps)$ and therefore
$\Opt(I_{t_i})\leq
(1+\eps)\Alg_{\on}(I_{t_i},\eps)+\frac{\mu}{\eps^{3}}$ will hold
for every $i\geq 1$.

Algorithm $\Alg_{\on}$ will apply $\Alg_{\off}$ (with the same
value of $\eps$) in every step, \ie when a new item arrives. This
is done in order to obtain an approximated value of the profit of
an optimal solution. As we saw before, arguments that are based on
total size are valid for approximating this value by a factor of
$2$, but not within a smaller factor. In the special times
$t_1,t_2,\ldots$ defined in what follows, the output of
$\Alg_{\off}$ will be used in order to obtain a new solution, and
the input is completely repacked such that it is packed as in this
solution ($\Alg_{\off}(I_{t_i},\eps)$).  In this case, all items
are removed from their bins and repacked identically to the output
of $\Alg_{\off}$. This is done instead of assigning the new item.
In other times which are not special, only the profit of
$\Alg_{\off}$ will be used, and the new item is packed into a new
bin.

Thus, at times that are not special, algorithm $\Alg_{\on}$ will
always greedily assign a new item to a new non-covered bin, and on
occasion it will define the current time as special and then it
completely \emph{repacks} the instance with the help of
$\Alg_{\off}$.  We define special times as follows. The first
special time $t_1$, which is the first time in which we repack the
instance using $\Alg_{\off}$ is when $\Alg_{\off}(I_{t_1}) \ge
3\mu/\eps^3 > 1$ (by $\eps \leq 0.5$ and $\mu\geq 1$). We have
$t_1>1$ since after the arrival of just one item no solution has
profit above $1$. Before time $t_1$, the profit of $\Alg_{\on}$ is
zero. The other special times are defined as follows. Let $t_i$
(for $i\geq 1$) be the last special time, \ie the last time that a
repacking was performed. For the current time $\tau>t_i$, if
$\Alg_{\off}(I_{\tau},\eps) < (1+\eps)\Alg_{\off}(I_{t_i},\eps)$,
the new item is packed into a new bin and $\tau$ is not a special
time. Otherwise, we let $t_{i+1}=\tau$, \ie $\tau$ is a special
time, and instead of packing the new item immediately, we perform
the next repacking. Thus, time $t_{i+1}>t_i$, is the first time
$\tau>t_i$ for which $\Alg_{\off}(I_{{\tau}},\eps)\geq
(1+\eps)\Alg_{\off}(I_{t_i},\eps)$, and thus
$\Alg_{\off}(I_{t_{i+1}},\eps)\geq
(1+\eps)\Alg_{\off}(I_{t_i},\eps)$ while
$\Alg_{\off}(I_{\theta},\eps)< (1+\eps)\Alg_{\off}(I_{t_i},\eps)$,
where $t_{i}+1 \leq  \theta \leq t_{i+1}-1$.

Clearly, the running time of our online algorithm is polynomial in
$t$ and $1/\eps$, since it is the sum of running times of
$\Alg_{\off}$ for all prefixes of the input. Before we discuss the
migration factor, we show that this algorithm always maintains a
near-optimal solution. We will use the property that holds at
special times, and analyze the other times.

At any time $\tau < t_1$ (before the first repacking), we have
that $\Alg_{\off}(I_{\tau}) < 3\mu /\eps^{3}$ while
$\Alg_{\on}(I_{\tau})  \geq 0$, and by
\eqref{eq:dynamic_amortized_offline} we get $\Opt(I_{\tau})\leq
(1+\eps)\Alg_{\off}(I_{\tau},\eps)+\frac{\mu}{\eps^{3}} \leq
(1+\eps)\cdot \frac{3\mu}{\eps^3}+\frac{\mu}{\eps^{3}} <
\frac{6\mu}{\eps^{3}} \leq
(1+3\eps)\Alg_{\on}(I_{\tau})+\frac{6\mu}{\eps^{3}}$. For every
special time $t_i$ we have $\Opt(I_{t_i})\leq
(1+\eps)\Alg_{\off}(I_{t_i},\eps)+\frac{\mu}{\eps^{3}}
=(1+\eps)\Alg_{\on}(I_{t_i},\eps)+\frac{\mu}{\eps^{3}} \leq
(1+3\eps)\Alg_{\on}(I_{t_i},\eps)+\frac{6\mu}{\eps^{3}}$, since at
any special time, once the repacking is applied, the packing of
$\Alg_{\on}$ is identical to that of $\Alg_{\off}$.

Now consider a time  $\tau$ satisfying $\tau > t_i$ for some $i
\geq 1$, and if $t_i$ is not the last special time, then we also
require $\tau<t_{i+1}$. Thus, no special times are defined at
times $t_i+1,\ldots,\tau$, and we have
\begin{align}
  \label{eq:dynamic_amortized_offtpp_vs_offt}
\Alg_{\off}(I_{\tau},\eps)< (1+\eps)\Alg_{\off}(I_{t_i},\eps).
\end{align}

Since all items $t_i+1,\ldots,\tau$ are assigned to new bins, and
covered bins remain such, we have $\Alg_{\on}(I_{\tau},\eps) \geq
\Alg_{\on}(I_{t_i},\eps)$. Due to the repacking at the special
time $t_i$, we find
$\Alg_{\on}(I_{t_i},\eps)=\Alg_{\off}(I_{t_i},\eps)$, and get
\begin{align}
  \label{eq:dynamic_amortized_ontpp_vs_offt}
  \Alg_{\on}(I_{\tau},\eps)\geq \Alg_{\on}(I_{t_i},\eps) = \Alg_{\off}(I_{t_i},\eps).
\end{align}

 Combining the last two inequalities
(\ref{eq:dynamic_amortized_offtpp_vs_offt}) and
(\ref{eq:dynamic_amortized_ontpp_vs_offt}) gives
\begin{align}
  \label{eq:dynamic_amortized_offtpp_vs_ontpp}
\Alg_{\off}(I_{\tau},\eps)\leq (1+\eps)\Alg_{\on}(I_{\tau},\eps).
\end{align}

By \eqref{eq:dynamic_amortized_offline} for time $\tau$,
$\Opt(I_{\tau})\leq
(1+\eps)\Alg_{\off}(I_{\tau},\eps)+\frac{\mu}{\eps^{3}}$, and we
get
\begin{align*}
  &\Opt(I_{\tau})\leq (1+\eps)\Alg_{\off}(I_{\tau},\eps)+\mu/\eps^{3}\leq  (1+\eps)^{2}\Alg_{\on}(I_{\tau},\eps)+\mu/\eps^{3}\leq\\
  &(1+3\eps)\Alg_{\on}(I_{\tau},\eps)+6\mu/\eps^{3},
\end{align*}
as $\eps^2 <\eps$.

Hence, $\Alg_{\on}$ achieves an asymptotic competitive ratio of
$1+3\eps$, and an appropriate scaling of $\eps$ allows us to
obtain competitive ratio $1+\eps$.

Finally, we will show that $\Alg_{\on}$ has an amortized migration
factor of $\Oh(1/\eps)$. Let $k$ be the number of special times
(which are $t_1,t_2,\ldots,t_k$). Let $U_i$ denote the optimal
profit at time $t_i$, which is a monotonically non-decreasing
sequence. Let $S_i$ denote the total size of the first $t_i$ input
items. The total size of migrated items is at most $\sum_{i=1}^k
S_i$, while the total size of items is at least $S_k$. Let
$\Delta_i=\Alg_{\off}(I_{t_i},\eps)$, where for every $1 \leq i
\leq k-1$ we have $\Delta_{i+1} \geq (1+\eps)\Delta_{i}$, and
additionally, $\Delta_1 \geq \frac{3\mu}{\eps^3} >1 $.

We claim that $U_i$ satisfies $\Delta_i \leq U_i < 3\Delta_i$. The
inequality $\Delta_i \leq U_i$ holds since $U_i$ is the optimal
profit at a certain time while $\Delta_i$ is the profit of some
solution for the same input. Since $U_i \geq U_1 \geq \Delta_1
\geq \frac{3\mu}{\eps^3}$ and by
\eqref{eq:dynamic_amortized_offline} we get $ U_i \leq
(1+\eps)\Delta_i+\frac{\mu}{\eps^{3}} < 2\Delta_i+ U_i/3$, which
implies $U_i < 3\Delta_i$.

Moreover, there exists an optimal solution where every bin has
load in $[1,2)$, and there is at most one bin that is not covered.
Thus, $U_i \leq S_i  <2U_i+1 < 3U_i$.

We find by $\Delta_k \leq U_k \leq S_k$:
$$\sum_{i=1}^k S_i < 3 \sum_{i=1}^k U_i \leq 9\sum_{i=1}^k
\Delta_i .$$ By
$$\sum_{j=0}^{\infty} \frac
1{(1+\eps)^j}=\frac{1}{1-1/(1+\eps)}=\frac{1+\eps}{\eps}=1+\frac{1}{\eps},$$

and by
$$\Delta_i \leq \frac{\Delta_k}{(1+\eps)^{k-i}},$$ we find $$ \frac
19 \cdot \sum_{i=1}^k S_i \leq \Delta_k \cdot \sum_{i=1}^{k} \frac
1 {(1+\eps)^{k-i}} = \Delta_k \cdot \sum_{j=0}^{k-1} \frac
1{(1+\eps)^j} \leq S_k \sum_{j=0}^{\infty} \frac 1{(1+\eps)^j}
=S_k \cdot (1+\frac{1}{\eps}) .$$

This analysis holds also for every prefix of the input if the
input was repacked at least once (and otherwise the amortized
migration factor is zero). Hence, the amortized migration factor
of $\Alg_{\on}$ is at most $9+\frac{9}{\eps}=\Oh(1/\eps)$ and we
obtain the following theorem.
\begin{theorem}
  \label{thm:ac_and_static_alg}
  For every $\eps > 0$, there is an algorithm for static online bin covering
  with polynomial running time, an asymptotic competitive ratio of $1 + \eps$,
  and an amortized migration factor of $\Oh(1/\eps)$.
\end{theorem}

In the following, we will show that an error in the amortized
setting is inevitable, even if we allow an unbounded running time.
\begin{theorem}
  \label{thm:ac_and_static_lb}
There is no (possibly exponential time)
algorithm for static online bin covering that maintains an optimal solution
with constant amortized migration factor $\beta$.
\end{theorem}
\begin{proof}
We say that a bin is \emph{perfectly covered}, if it contains items of total size exactly $1$, and call a solution perfect if all of its bins are perfectly covered.
Obviously, if there is a perfect solution to a bin covering instance, it is optimal and each optimal solution for this instance has to be perfect.

Let $N$ be a positive even integer.
Our input consists of $N + 1$ phases.
The first phase (phase $0$) consists of items of total size $N$ and phase $j$ (for $j\geq 1$) consists of items of total size $2$.
Let $I_{j}$ be the instance after phase $j$ arrived.

We will ensure that the optimum $\Opt(I_{j})=N+2j$ is realized by a unique perfect solution (up to some transpositions).
Furthermore, the total size of items that need to be moved to transform an optimal solution of $I_{j}$ to an optimal solution of $I_{j+1}$ is at least $(N-j+2)/3$.
Hence, items with total size $\Omega(N^2)$ have to be moved to maintain an
optimal solution at all times while the total size of all arriving items is $3N$
yielding that the needed amortized migration is $\Omega(N)$.
Since the construction works for arbitrarily big values of $N$, this suffices to
prove the claim, as this rules out a constant amortized migration factor.

We will proceed by first describing all occurring items in detail and noting some basic properties, then we define the phases, describe the optimal solutions and finally we prove their uniqueness and the necessary migration.

Let $\eps = 1/(N+1)^4$.
The items are either huge, large, medium or tiny and each of these classes contains items of $N$ levels.
More precisely, for each level $i\in [N]$ we have:

\begin{itemize}
\item One huge item with size $\ihuge(i) = 1 - (N-i+1)\eps(1+(i-1)(N+1))$.
\item One large item of type A with size $\ilarge_A(i) = 1/2 + i(N+1)\eps$, and one large item of type B with size $\ilarge_B(i) = \ilarge_A(i) + \eps$.
\item One medium item of type A with size $\imedium_A(i) = 1/2 - i(N+1)\eps$, and one large item of type B with size $\imedium_B(i) = \imedium_A(j) - \eps$.
\item $N - i + 1$ tiny items each with size $\itiny(i) = \eps(1 + (i-1)(N+1))$.
\end{itemize}
We now state some simple properties of these items:
\begin{claim}
\label{claim:item_properties}
  It holds
  \begin{enumerate}
  \item $\ilarge_A(i) + \imedium_A(i) = \ilarge_B(i) + \imedium_B(i) = \ihuge(i)
    +  (N - i + 1)\itiny(i) = 1$;
  \item $\sum_{i=1}^{N}(N-i+1)\cdot \itiny(i) < 1$ (\ie all tiny items
    together cannot fill a bin completely);
  \item $\ihuge(N) < \ihuge(1)$, $\ilarge_{A}(1) < \ilarge_{A}(N)$,
    $\ilarge_{B}(1) < \ilarge_{B}(N)$, $\imedium_{A}(N) < \imedium_{A}(1)$,
    $\imedium_{B}(N) < \imedium_{B}(1)$, $\itiny(1) < \itiny(N)$;
  \item $\ihuge(N) > \ilarge_{B}(N)$, $\ilarge_{A}(1) > \imedium_{A}(1)$,
    $\imedium_{B}(N) > \itiny(N)$;
  \item $\imedium_{B}(N) > 1/3$, $\ilarge_{A}(1) > 1/2$, $\imedium_{A}(1) <
    1/2$, $\ihuge(N)+\imedium_{B}(N) >~1$
  \item $\ihuge(i)+\imedium_{X}(i')>1$ and $\ihuge(i)+\ilarge_{X}(i')>1$ for all
    $X\in \{A,B\}$ and all $i,i'\in \{1,\ldots,N\}$.
  \end{enumerate}

\end{claim}
For completeness, the simple proof of this fact can be found in the appendix. 

Hence, large and medium items of the same type and level can be used to pack a bin perfectly and every other combination of two large or medium items does not cover a bin perfectly.
Let $s_h, s_\ell, s_m$ and $s_t$ be sizes of some huge, large, medium and tiny
item respectively. Summarizing the claim, we have
(i) all tiny items together cannot fill a whole bin, (ii) $s_h > s_\ell > s_m >
s_t$, (iii) $s_m > 1/3$, (iv) $s_\ell > 1/2 > s_m$, and (v) $s_h + s_m > 1$.
Hence, in a perfect solution to an instance $I$ comprised of a subset of the above items, a bin may contain
\begin{enumerate*}[label=(\arabic*)]
\item exactly one large item together with the unique matching medium item;
\item at most two items that are medium or large, and at least one tiny item; or
\item exactly one huge, and at least one tiny item.
\end{enumerate*}
Furthermore, if $I$ contains the first $n\leq N$ levels of huge items, we have:
\begin{claim}\label{claim:huge_items}
In a perfect solution of $I$, the first $n$ levels of huge items have to be packed together with the first $n$ levels of tiny items in a unique way.
\end{claim}
\begin{proof}
To see this, consider the bin $B_{k}$ containing the huge item of level $k$ in a
perfect solution. Due to point 6 in Claim~\ref{claim:item_properties}, $B_{k}$ cannot contain any
large or medium item, as the solution is perfect. Hence, $B_{k}$ must contain
some tiny items. These tiny items in $B_{k}$ must have a total size of $r_{k}=1-\ihuge(k)=(N-k+1)\eps(1+(k-1)(N+1))$.
As $(r_{k}/\epsilon) \bmod (N+1)=(N-k+1)$ and $(\itiny(i)/\eps) \bmod (N+1)=1$
for all $i=1,\ldots,N$, we need at least $N-k+1$
tiny items to cover bin $B_{k}$.

Now, suppose that there is some bin $B_{k}$ that does not contain all $N-k+1$
tiny items of level $k$ and let this $k$ be minimal. Hence, for each $i\in
\{1,\ldots,k-1\}$, the bin $B_{i}$ contains the $N-i+1$ tiny items of level $i$
and all of these tiny items are thus packed.
As $B_{k}$ needs to contain $N-k+1$ tiny items, it hence must contain at least one
tiny item of size at least $\itiny(k+1)$ (due to the monotonicity of $\itiny(i)$). Hence, the sum of the tiny items in
$B_{k}$ is at least
\begin{align*}
  &\itiny(k+1)+(N-k)\itiny(k) = \eps(1 + k(N+1))+ (N-k)\eps(1 + (k-1)(N+1)) >\\
  &(N-k+1)\eps(1+(k-1)(N+1)) = r_{k}.
\end{align*}
Hence, $B_{k}$ is not perfectly covered, a contradiction to the assumption.

\end{proof}

We proceed with the definition of the phases:
Phase $0$ consists of all large items of type A, all medium items of type B and
the tiny items of level $1$. Note that the large items of type $B$ and the
medium items of type $A$ are not yet available.
As $\ilarge_A(i) + \imedium_B(i) + \itiny(1) = 1$ (see Claim~\ref{claim:item_properties}), the total volume of items in phase $0$ is $N$.
This (perfect) packing also shows that $\Opt(I_{0})=N$ (see
Fig.~\ref{fig:lb_static_nc_batch0}).

\begin{claim}\label{claim:packing_phase0}
  Putting one item of size $\ilarge_{A}(i)$, one item of size $\imedium_{B}(i)$,
  and one item of size $\itiny(1)$ into a bin is the unique optimal perfect
  packing (up to the packing of the tiny items) for phase $0$.
\end{claim}
\begin{proof}
Due to the definition of the sizes, two large items have total size
strictly larger than $1$ and two medium items have total size strictly smaller
than $1$. The total size of items in phase $0$ is exactly $N$, the total size of
all tiny items is not enough to cover a single bin
(Claim~\ref{claim:item_properties}), and we have exactly $N$ large and $N$
medium items. Hence, every bin in a perfect packing must contain one large item
and one medium item. Now let $B_{i}$ be the bin containing $\ilarge_{A}(i)$.

Suppose that the solution described above is not unique. Hence, there is some
bin $B_{i}$ that contains an item of size $\imedium_{B}(j)$ with $i\neq j$. If
$j > i$, the monotonicity of $\imedium_{B}(\cdot)$ shows that this bin is not
perfectly covered. If $j < i$, the pidgeonhole principle shows that there must
be some bin $B_{i'}$ containing an item of size $\imedium_{B}(j')$ with $j' >
i'$, which is again a contradiction.
\end{proof}

For each $j\in[N]$, phase $j$ consist of the huge item of level $j$, the large type B item of level $j$, the medium type A item of level $N-j +1$ and the tiny items of level $j+1$ if $j<N$, and no tiny items otherwise.
The items of instance $I_j$ can be packed perfectly as follows.
The huge and tiny items of level $1$ to $j$ are packed in the unique possible way described above; the large and medium type B items of level $1$ to $j$ are matched to form $j$ perfectly packed bins,  and similarly the large and medium type A items of level $N$ to $N-j+1$ are matched.
Lastly, for each $k\in\{1,\dots,N-j\}$, we have $\ilarge_A(k) + \imedium_B(j +
k) + \itiny(j+1) = 1$ and we pack the remaining items correspondingly (see
Fig.~\ref{fig:lb_static_nc_batchj}). This packing has $N+2j$ perfectly covered
bins.

\begin{claim}
The packing described above is unique.
\end{claim}
\begin{proof}
Claim~\ref{claim:huge_items} implies that the huge items are packed together
with the fitting tiny ones in $j$ perfectly covered bins.
There are only $N-j$ tiny items left and each bin
that does not correspond to a perfect match of a medium and large item has to
contain a tiny item as described
in the proof of Claim~\ref{claim:packing_phase0}. As we still need to fill $N+j$
bins perfectly, $2j$ of these bins need to be packed by perfect matches. This is
exactly the number of perfect matches available and we thus need to pack all of
them.

Furthermore, large items cannot be combined in one perfectly covered bin (Claim~\ref{claim:item_properties}) and therefore each of
them has to be packed into a single bin together with a tiny item.
For each such bin there is exactly one medium item left that fits perfectly, namely the one chosen above.
\end{proof}

In order to transform the unique perfect solution of instance $I_{j-1}$ into the
one for $I_j$, at least $N-j+2$ non-tiny items have to be moved (see Fig.~\ref{fig:lb_static_nc_batchj}).
Each of these items has a size of at least $1/3$. Hence, a total migration of
$\sum_{j=1}^{N}(N-j+2)/3\geq \Omega(N^{2})$ is needed while the total size of
the instance is $3N$. As we are only allowed to repack a total size of
$\beta 3N$, this implies $\beta \geq \Omega(N)$, a contradiction  to the
assumption that $\beta$ is a constant.

This concludes the proof.

\begin{figure}[h]
  \centering
\begin{tikzpicture}[xscale=2.3, yscale=1.8]
  \drawBin{2}{}{b1}
  \drawBin{4}{}{b2}
  \drawBin{7}{}{b3}
  \node at ($ (b2)!0.5!(b3)$) {\ldots};

  \foreach \i/\p in {1/2,2/4,N/7}{
  \addToBin{\p}{0}{.5}{$\ilarge_{A}(\i)$}{fill=medium}{}
  \addToBin{\p}{.5}{.35}{$\imedium_{B}(\i)$}{fill=medium}{}
  \addToBin{\p}{.85}{.15}{$\itiny(\i)$}{fill=medium}{}
 }
\end{tikzpicture}
  \caption{Optimal packing of phase $I_{0}$}
  \label{fig:lb_static_nc_batch0}
\end{figure}
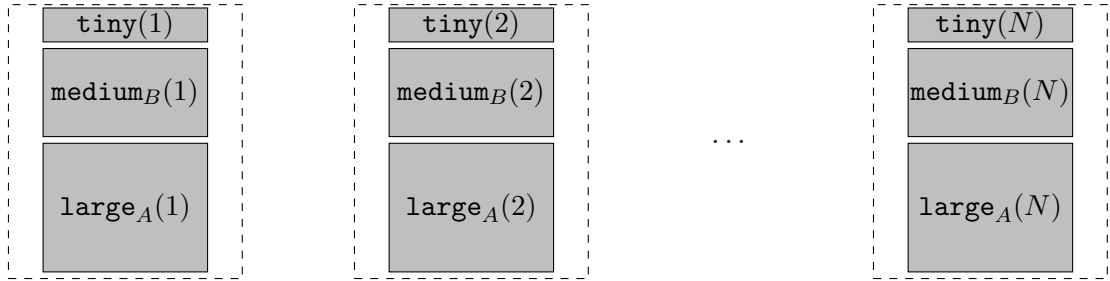

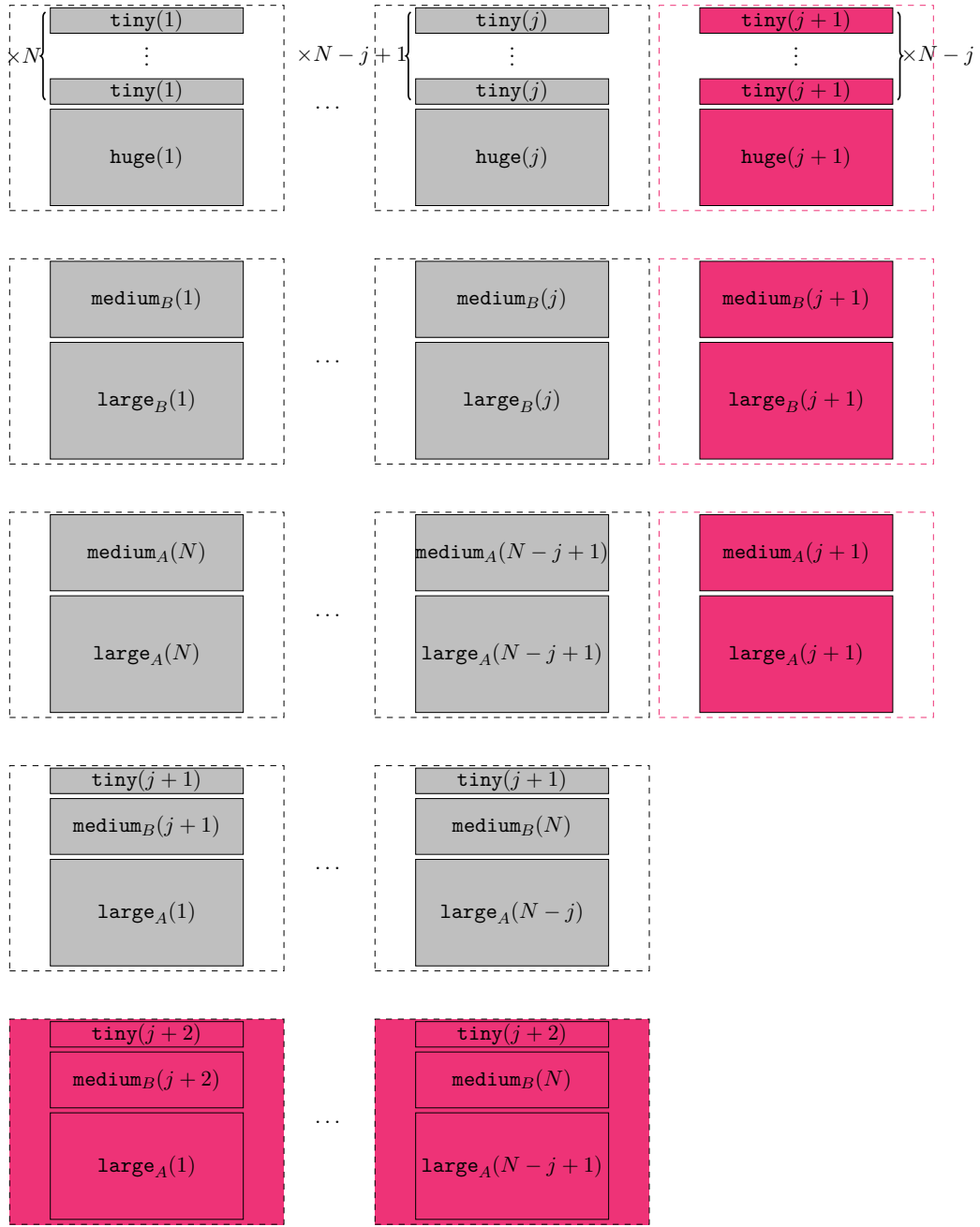
\begin{figure}[h]
  \centering
  \scalebox{.8}{
\begin{tikzpicture}[xscale=3.6, yscale=1.8]
  \drawBin{2}{}{b1}
  \drawBin{3.8}{}{b2}
  \node at ($ (b1)!0.5!(b2)$) {\ldots};

  \addToBin{2}{0}{.5}{$\ihuge(1)$}{fill=medium}{}
  \addToBin{2}{.5}{.15}{$\itiny(1)$}{fill=medium}{i1}
  \addToBin{2}{.85}{.15}{$\itiny(1)$}{fill=medium}{i2}
  \node[yshift=.1cm] at ($ (i1)!0.5!(i2) $) {$\vdots$};

  \draw[decoration={brace,raise=5pt},decorate, thick] ($ (i1.south) - (.45cm,-.1cm) $) -- node[left,xshift=-.1cm] {$\times N$} ($ (i2.north) - (.45cm,.1cm) $);

  \addToBin{3.8}{0}{.5}{$\ihuge(j)$}{fill=medium}{}
  \addToBin{3.8}{.5}{.15}{$\itiny(j)$}{fill=medium}{j1}
  \addToBin{3.8}{.85}{.15}{$\itiny(j)$}{fill=medium}{j2}
  \node[yshift=.1cm] at ($ (j1)!0.5!(j2) $) {$\vdots$};

  \draw[decoration={brace,raise=5pt},decorate, thick] ($ (j1.south) - (.45cm,-.1cm) $) -- node[left,xshift=-.1cm] {$\times N-j+1$} ($ (j2.north) - (.45cm,.1cm) $);

  \drawBin{5.2}{draw=myred}{b3}

    \addToBin{5.2}{0}{.5}{$\ihuge(j+1)$}{fill=myred}{}
  \addToBin{5.2}{.5}{.15}{$\itiny(j+1)$}{fill=myred}{k1}
  \addToBin{5.2}{.85}{.15}{$\itiny(j+1)$}{fill=myred}{k2}
  \node[yshift=.1cm] at ($ (k1)!0.5!(k2) $) {$\vdots$};

  \draw[decoration={brace,raise=5pt, mirror},decorate,  thick] ($ (k1.south) - (-.45cm,-.1cm) $) -- node[right,xshift=.1cm] {$\times N-j$} ($ (k2.north) - (-.45cm,.1cm) $);

  \begin{scope}[yshift=-2.5cm]
    \drawBin{2}{}{c1}
    \drawBin{3.8}{}{c2}
    \node at ($ (c1)!0.5!(c2)$) {\ldots};

    \addToBin{2}{0}{.6}{$\ilarge_{B}(1)$}{fill=medium}{}
    \addToBin{2}{.6}{.4}{$\imedium_{B}(1)$}{fill=medium}{}
    \addToBin{3.8}{0}{.6}{$\ilarge_{B}(j)$}{fill=medium}{}
    \addToBin{3.8}{.6}{.4}{$\imedium_{B}(j)$}{fill=medium}{}

    \drawBin{5.2}{draw=myred}{c3}
    \addToBin{5.2}{0}{.6}{$\ilarge_{B}(j+1)$}{fill=myred}{}
    \addToBin{5.2}{.6}{.4}{$\imedium_{B}(j+1)$}{fill=myred}{}
  \end{scope}

  \begin{scope}[yshift=-5cm]
    \drawBin{2}{}{c1}
    \drawBin{3.8}{}{c2}
    \node at ($ (c1)!0.5!(c2)$) {\ldots};

    \addToBin{2}{0}{.6}{$\ilarge_{A}(N)$}{fill=medium}{}
    \addToBin{2}{.6}{.4}{$\imedium_{A}(N)$}{fill=medium}{}
    \addToBin{3.8}{0}{.6}{$\ilarge_{A}(N-j+1)$}{fill=medium}{}
    \addToBin{3.8}{.6}{.4}{$\imedium_{A}(N-j+1)$}{fill=medium}{}

    \drawBin{5.2}{draw=myred}{c3}
    \addToBin{5.2}{0}{.6}{$\ilarge_{A}(j+1)$}{fill=myred}{}
    \addToBin{5.2}{.6}{.4}{$\imedium_{A}(j+1)$}{fill=myred}{}
  \end{scope}

  \begin{scope}[yshift=-7.5cm]
    \drawBin{2}{}{c1}
    \drawBin{3.8}{}{c2}
    \node at ($ (c1)!0.5!(c2)$) {\ldots};

    \addToBin{2}{0}{.55}{$\ilarge_{A}(1)$}{fill=medium}{}
    \addToBin{2}{.55}{.3}{$\imedium_{B}(j+1)$}{fill=medium}{}
    \addToBin{2}{.85}{.15}{$\itiny(j+1)$}{fill=medium}{}

    \addToBin{3.8}{0}{.55}{$\ilarge_{A}(N-j)$}{fill=medium}{}
    \addToBin{3.8}{.55}{.3}{$\imedium_{B}(N)$}{fill=medium}{}
    \addToBin{3.8}{.85}{.15}{$\itiny(j+1)$}{fill=medium}{}

  \end{scope}

  \begin{scope}[yshift=-10cm]
    \drawBin{2}{fill=myred}{c1}
    \drawBin{3.8}{fill=myred}{c2}
    \node at ($ (c1)!0.5!(c2)$) {\ldots};

    \addToBin{2}{0}{.55}{$\ilarge_{A}(1)$}{fill=myred}{}
    \addToBin{2}{.55}{.3}{$\imedium_{B}(j+2)$}{fill=myred}{}
    \addToBin{2}{.85}{.15}{$\itiny(j+2)$}{fill=myred}{}

    \addToBin{3.8}{0}{.55}{$\ilarge_{A}(N-j+1)$}{fill=myred}{}
    \addToBin{3.8}{.55}{.3}{$\imedium_{B}(N)$}{fill=myred}{}
    \addToBin{3.8}{.85}{.15}{$\itiny(j+2)$}{fill=myred}{}

  \end{scope}

\end{tikzpicture}}
\caption{Optimal packing of phase $I_{j}$ and phase $I_{j+1}$ (in red)}
\label{fig:lb_static_nc_batchj}

\end{figure}
\end{proof}

\clearpage

\appendix
\section{Proof of Claim 24}

\begin{proof}[Proof of \autoref{claim:item_properties}]
  \begin{enumerate}
  \item We simply calculate:
    \begin{compactitem}
    \item $\ilarge_{A}(i)+\imedium_{A}(i)=1/2+i(N+1)\eps+1/2 - i(N+1)\eps=1$
    \item $\ilarge_{B}(i)+\imedium_{B}(i)=1/2+i(N+1)\eps+1/2+\eps -
      i(N+1)\eps-\eps=1$
    \item $\ihuge(i)+(N-i+1)\itiny(i)=\\
      {1 - (N-i+1)\eps(1+(i-1)(N+1))+ (N-i+1)\eps(1 + (i-1)(N+1))}=1$
    \end{compactitem}
  \item We have
    \begin{align*}
      &\sum_{i=1}^{N}(N-i+1)\cdot \itiny(i) = \sum_{i=1}^{N}\underbrace{(N-i+1)}_{\leq N}\cdot \eps(1 + \underbrace{(i-1)}_{\leq N-1}(N+1)) \leq\\
      &\sum_{i=1}^{N}N\cdot \eps (1+(N-1)(N+1))= \sum_{i=1}^{N}N\cdot \eps(1+N^{2}-1)= N^{4}\eps = N^{4}/(N+1)^{4} <1.
    \end{align*}
  \item We also calculate:
    \begin{compactitem}
    \item $\ihuge(N)=1-N\eps < 1-N^{2}\eps=1-\eps(1+(N-1)(N+1)) = \ihuge(1)$
    \item $\ilarge_{A}(1)=1/2 + (N+1)\eps < 1/2 + N(N+1)\eps = \ilarge_{A}(N)$
    \item $\ilarge_{B}(1)=1/2 + (N+1)\eps+\eps < 1/2 + N(N+1)\eps+\eps =
      \ilarge_{B}(N)$
    \item $\imedium_{A}(N)=1/2 - N(N+1)\eps < 1/2-(N+1)\eps = \imedium_{A}(1)$
    \item $\imedium_{B}(N)=1/2 - N(N+1)\eps-\eps < 1/2-(N+1)\eps-\eps =
      \imedium_{B}(1)$
      \item $\itiny(1)=\eps < \eps(1+(N-1)(N+1))=\itiny(N)$
      \end{compactitem}
    \item Again, we calculate
      \begin{compactitem}
        \item Note that $(N+1)^{4} > (N+1)^{4}/2 + N(N+3)$ for all $N\geq 1$ and
          hence
          $\ihuge(N) = 1-N\eps = 1-N/(N+1)^{4} > 1/2 + N(N+2)/(N+1)^{4} =\\ 1/2 + N(N+1)\eps+\eps =
          \ilarge_{B}(N)$.
        \item $\ilarge_{A}(1)=1/2 + (N+1)\eps > 1/2-(N+1)\eps = \imedium_{A}(1)$
        \item Note that $(N+1)^{4}/2 > N^{2}+N(N+2)$ for $N\geq 1$ and hence
          $\imedium_{B}(N)= 1/2 - N(N+1)\eps -\eps = 1/2-N(N+2)\eps =
          1/2-N(N+2)/(N+1)^{4} > N^{2}/(N+1)^{4} = \eps(1+(N-1)(N+1))=\itiny(N)$
      \end{compactitem}
    \item Finally, some calculations show that
      \begin{compactitem}
        \item Note that $(N+1)^{4}/2 > (N+1)^{4}/3 + N(N+2)$ and hence
          $\imedium_{B}(N)= 1/2 - N(N+1)\eps -\eps = 1/2-N(N+2)\eps > 1/3$.
        \item $\ilarge_{A}(1)=1/2 + (N+1)\eps > 1/2$
        \item $ \imedium_{A}(1) = 1/2-(N+1)\eps  < 1/2$
        \item From $\imedium_{B}(N) > 1/3$, we have $\ihuge_{N}+\imedium_{B}(N)
          >1-N\eps+1/3 = 1-N/(N+1)^{4}+1/3 > 1$.
        \end{compactitem}
      \item This follows from the above points. Note that $\ihuge(N)$ is
        the smallest huge item and $\imedium_{B}(1)$ the smallest medium item.
        Hence, point 5 implies that the sum of these two already exceeds $1$.
        As every medium or large item is larger than $\imedium_{B}(1)$ due to
        points 3 and 4, the claim follows.
  \end{enumerate}
\end{proof}

\clearpage

\bibliography{references}

\end{document}